\documentclass[pra,aps,preprint,showpacs,preprintnumbers,
amsmath,amssymb,floatfix]{revtex4-1}
\usepackage{enumerate}
\usepackage{graphicx}
\usepackage{amsmath, amsthm, amssymb}
\newtheorem{theorem}{Theorem}
\newtheorem{corollary}{Corollary}
\newtheorem{lemma}{Lemma}
\newtheorem{conjecture}{Conjecture}

\theoremstyle{remark}
\newtheorem*{remark}{Remark}

\usepackage{enumerate}
\usepackage{graphicx}
\usepackage{amsmath, amssymb}
\usepackage{graphics}
\usepackage{bm}
\begin{document}
\newcommand{\real}{\textrm{Re}\:}
\newcommand{\sto}{\stackrel{s}{\to}}
\newcommand{\Tr}{\textrm{Tr}\:}
\newcommand{\Ran}{\textrm{Ran}}
\newcommand{\supp}{\textrm{supp}\:}
\newcommand{\evs}{\textnormal{evs}\:}
\newcommand{\Ker}{\textnormal{Ker}\:}
\newcommand{\sign}{\textnormal{sign}\:}
\newcommand{\wto}{\stackrel{w}{\to}}
\newcommand{\ssto}{\stackrel{s}{\to}}
\newcommand{\epstr}{{\tilde \epsilon}_0}
\newcommand{\beps}{\pmb{\epsilon}}
\newcounter{foo}
\newcommand{\sslim}{\textnormal{s--}\lim}
\newcommand{\wlim}{\textnormal{w--}\lim}
\providecommand{\norm}[1]{\lVert#1\rVert}
\providecommand{\abs}[1]{\lvert#1\rvert}
\providecommand{\absbm}[1]{\pmb{|}\mspace{1mu}#1\mspace{1mu}\pmb{|}}

\title{Why there is no Efimov effect for four bosons and related results on the
finiteness of the discrete spectrum.}

\author{Dmitry K. Gridnev}
\affiliation{FIAS, Ruth-Moufang-Stra{\ss}e 1, D--60438 Frankfurt am Main,
Germany}
\altaffiliation[On leave from:  ]{ Institute of Physics, St. Petersburg State
University, Ulyanovskaya 1, 198504 Russia}

\begin{abstract}
We consider a system of $N$ pairwise interacting particles described by
the Hamiltonian $H$, where $\sigma_{ess} (H) = [0,\infty)$
and none of the particle pairs has a zero energy resonance. The pair
potentials are allowed to take both signs and obey certain restrictions regarding the fall off.
It is proved that if $N \geq 4$ and none of the Hamiltonians corresponding to the 
subsystems containing $N-2$ or less particles has an eigenvalue equal to zero
then $H$ has a finite number of negative energy bound states. This result
provides a positive proof to a long--standing conjecture of Amado and Greenwood stating
that four bosons with an empty negative continuous spectrum have at most a
finite number of negative energy bound states. Additionally, we give a short proof to the theorem of Vugal'ter and Zhislin on the finiteness
of the discrete spectrum and pose a conjecture regarding the existence of the ``true'' four--body Efimov effect.
\end{abstract}

\maketitle

\section{Introduction}\label{sec:1}

In 1970 V. Efimov predicted \cite{vefimov} a remarkable and counterintuitive phenomenon now called the Efimov effect, which can be stated as follows.
If the negative continuous spectrum of a three-particle Hamiltonian $H$ is empty but
at least two of the particle pairs have a resonance at zero energy then $H$ has an infinite number of negative energy bound states. Thereby the pair--interactions' range
can be finite. This effect was in striking contradiction with the general knowledge of that time saying
that an infinite number of bound states can only be produced by long--range interactions. The first sketch of mathematical proof
of the Efimov effect was done by
L. D. Faddeev shortly after V. Efimov told him about his discovery \cite{efimovprivate}. The first published proof, which was not
mathematically flawless, appeared in \cite{amadonoble}. Later D. R. Yafaev \cite{yafaev} basing on the Faddeev's idea presented a complete mathematical proof;
in \cite{ovchinnikov,tamura,fonseca} one finds other proofs by different methods.
The original Faddeev's argument and his derivation of the discrete spectrum asymptotics in the case of three identical particles can be found in \cite{merkuriev}.
The spectral asymptotics for particles with unequal masses was analyzed in \cite{sobolev}.
In \cite{wang} the author claimed having generalized the result in \cite{yafaev,sobolev} to the case of three clusters but the proof in \cite{wang}
contains a mistake \cite{wangwrong}.

After the Efimov effect was proved to exist in the three--particle case the researchers had an eye on its
most straightforward generalization, namely, the case of four bosons with an empty negative continuous spectrum.
Amado and Greenwood in \cite{greenwood} claimed having proved that the Efimov effect is impossible for $N \geq 4$ bosons.
For four bosons their prediction later got numerical confirmation, see f. e. \cite{universal}. This showed that somehow
the Efimov effect appeared to be possible only for three bosons, not more not less.
The ``proof'' in \cite{greenwood} is invalid: the authors make various unclear and ungrounded assumptions (in particular, 
one could question the validity of the expansion Eq.~(10), where somehow $B\neq 0$, or the finiteness of the nominator in Eq.~6 in \cite{greenwood}, etc.). 
The reader should also be warned against
the misused terminology in \cite{amadonoble,greenwood}, where
the authors use the term ``zero energy bound states'' both for zero energy
resonances and $L^2$ states, which may lead to a controversy already in the three--particle case \cite{1}. The aim of the present paper is to give a
correct mathematical proof to the Amado--Greenwood's conjecture. For a more detailed explanation of our result we refer the reader to Sec.~\ref{25.9;15}.

In a series of papers (in particular, see \cite{vugzhisl1,vugzhisl2,vugzhisl3,vugzhisl4})
Vugal'ter and Zhislin using the variational approach derived several theorems concerning the finiteness of the discrete spectrum of Schr\"odinger
operators. Applied to systems, where subsystems may have virtual levels, their results are the most advanced ones so far.
In Theorem~1.3 in \cite{vugzhisl1} these authors prove the following: suppose that
the system of $N$ particles is described by the Hamiltonian $H$ with $\sigma_{ess} (H) = [0, \infty)$. Suppose also that the particles can be partitioned into two clusters
$C_1$ and $C_2$ in such a way that any subsystem having the particles both from $C_1$ and $C_2$ does not have a virtual level at zero energy. Then $H$ has at most a finite
number of negative eigenvalues. Note, that subsystems having the particles only from $C_1$ (or equivalently only from $C_2$) are allowed to have virtual
levels! Already in \cite{yafaev} as a byproduct of the main result one finds the proof
that if $N=3$ and at most one particle pair has a zero energy resonance then the number of negative energy levels is finite; in \cite{yafaev29} 
this case is analyzed in more detail. 
In that sense one can view the theorem of Vugal'ter and Zhislin as a
generalization of Yafaev's result regarding the finiteness of the discrete spectrum to $N \geq 4$.
The proof of Theorem~1.3 in \cite{vugzhisl1} is rather involved
(a proposed simplification in \cite{ahia} helps only a little since the method is practically the same as in \cite{vugzhisl1}
but various results from \cite{vugzhisl1,vugzhisl2,vugzhisl3} are cited without a proof). Sec.~\ref{1.08;33} of this paper
contains a relatively simple proof of Theorem~\ref{30.07;1}, which implies the theorem of Vugal'ter and Zhislin, however,  restricted to the class of potentials considered here.
Our proof uses the Birman--Schwinger (BS) principle and we believe that it contributes to a better understanding of the results in \cite{vugzhisl1}.

The Amado--Greenwood's conjecture \cite{greenwood} that is most interesting from a physics point of view
is not covered by the results of Vugal'ter and Zhislin. Here one can mention two sorts of arising difficulties.
One is, quoting Zhislin, a lack of information on the near--threshold resolvent behavior of the
N-body system for $N \geq 3$. Another difficulty appears if one tries to extend Yafaev's analysis
to the case $N \geq 4$. Namely, in \cite{yafaev,sobolev} the eigenvalues were ``counted'' using
symmetrized Faddeev equations, which have a  compact kernel away from the
resonances. In the case $N \geq 4$ Faddeev equations seize being compact, which is a
known problem \cite{motovilov,yakovlev,merkuriev}. Their exists a generalization in form of
Faddeev--Yakubovsky equations \cite{yakubovsky,merkuriev}, but it is not at all obvious how one can adopt Faddeev--Yakubovsky equations
to the purpose of counting eigenvalues.

In our proof of the Amado--Greenwood's conjecture we use an approach, which is in the spirit of \cite{yafaev}.
We reduce the problem to the analysis of the spectrum of
an integral operator, which arises after successively applying $N$ times the BS principle. In distinction from \cite{yafaev,sobolev},
where integral equations had a $3\times 3$ matrix form,
the resulting integral equation in our approach is written in one line; this equation is subsequently used for counting eigenvalues. One should also
mention the fundamental difference between the cases of 3 and 4 identical particles. In the case $N=3$
the two--particle subsystems have at most a zero energy resonance but cannot have a zero energy $L^2$ state.
On the contrary, for $N=4$ the 3--particle subsystems may have square integrable zero energy
bound states \cite{1,jpasubmit}. These zero energy three--body ground states have a power fall--off,
the most trivial lower bound \cite{1} being $|\psi_0 (x)| \geq (1+|x|)^{-4}$, where
$x \in \mathbb{R}^6$. The hard part of the proof is to
show in some sense that these states fall off rapidly enough in some sense so that they cannot lead to an infinite number of bound states. The presence of a bound state at zero energy also
affects the dependence of the energy on the coupling constant near the threshold \cite{klaus1,klaus2,simonfuncan} and shapes the low--energy behavior of the resolvent.

The paper is organized as follows. In Sec.~\ref{1.08;33} we prove Theorem~\ref{30.07;1}, which implies Theorem~1.3 in \cite{vugzhisl1}.
Sec.~\ref{25.9;15} contains the statement of the main theorem and an intriguing conjecture concerning the existence of the true four--body Efimov effect.
Sec.~\ref{1.08;1} deals with the situation when the N--particle system is at critical coupling (for the definition of critical coupling, see \cite{jpasubmit}).
The results of Sec.~\ref{1.08;1} are then applied to the subsystems
containing $N-1$ particles in Sec.~\ref{1.08;2}, where we prove the Amado--Greenwood's conjecture.
The Appendix reviews the BS principle, thereby, we extensively use the results from \cite{klaus2}.

Here are some of the notations used in the paper. We define
$\mathbb{R}_+ := \{x\in \mathbb{R}|x \geq 0\}$ and $\mathbb{Z}_+ := \{n \in \mathbb{Z}| n \geq 0\}$, where $\mathbb{R} , \mathbb{Z}$ are reals and integers respectively.
The function $\chi_R : \mathbb{R}_+ \to \mathbb{R}_+$ for $R > 0$ is such that
$\chi_R (r) = 1$ for $r \leq R$ and  $\chi_R (r) = 0$ otherwise.
For $v : \mathbb{R}^n \to \mathbb{R}$ the positive and negative parts are $v_{\pm} := \max[\pm v , 0]$ so
that $v= v_+ - v_-$. For a self--adjoint operator $A$ following \cite{imsigal} we denote $\# (\textrm{evs}(A) >
\lambda)$ the number of eigenvalues of $A$ (counting multiplicities) larger
than $\lambda \in \mathbb{R}$. In the case when $\sigma_{ess} (A) \cap (\lambda, \infty) \neq \emptyset$ we set by definition $\# (\textrm{evs}(A) > \lambda) = \infty$.
The versions of this definition using other relation symbols like $<, \geq, \leq$ are self--explanatory.
The linear space of bounded
operators on a Hilbert space $\mathcal{H}$ is always denoted as $\mathfrak{B}(\mathcal{H})$.  For a self--adjoint operator $A$ the notation $A \ngeqq 0$
means that there exists $f_0 \in D(A)$ such that $(f_0, Af_0) < 0$. $\|A \|_{HS}$ denotes the Hilbert--Schmidt norm of the Hilbert--Schmidt operator $A$. The notation
$f \in L^\infty_\infty (\mathbb{R}^n)$ means that $f: \mathbb{R}^n \to \mathbb{C}$ is a bounded Borel function going to zero at infinity.

\section{Theorem of Vugal'ter and Zhislin}\label{1.08;33}

We consider the Schr\"odinger operator of $N$ particles in $\mathbb{R}^3$
\begin{equation}\label{30.07;8}
 H = H_0 + \sum_{1 \leq i< j \leq N} v_{ij} ,
\end{equation}
where $H_0$ is the kinetic energy operator with the center of mass removed and
$v_{ij}$ are operators of multiplication by $v_{ij} (r_i -r_j)$. Here and further $r_i \in \mathbb{R}^3$ and $m_i \in \mathbb{R}_+ /\{0\}$ always denote particle position vectors and masses. 
For pair--interactions we shall require (throughout this section) that $(v_{ij})_+ \in L^2 (\mathbb{R}^3) + L^\infty_\infty (\mathbb{R}^3)$ and 
$(v_{ij})_- \in L^3 (\mathbb{R}^3) \cap L^{ 3/2 - \alpha} (\mathbb{R}^3)$, where  $\alpha \in (0, 1/2)$ has a fixed value throughout this section. 
By the Kato--Rellich theorem \cite{kato,teschl}
$H$ is self--adjoint on $D(H_0) = \mathcal{H}^2 (\mathbb{R}^{3N-3}) \subset L^2 (\mathbb{R}^{3N-3})$ (the symbol $\mathcal{H}^2$ denotes the corresponding Sobolev space \cite{teschl,liebloss}).
The conditions on pair--interactions guarantee that for all $\gamma \in [0, \alpha(3-2\alpha)^{-1}]$ there is a constant $c_\gamma \in \mathbb{R}_+$ such that 
\begin{equation}\label{30.07;4}
\left[ \int \frac{\bigl(v_{ij}(r)\bigr)_-   \bigl(v_{ij}(r')\bigr)_- } {|r-r'|^{2 - 8\gamma}} d^3r d^3 r' \right]^{1/2} < c_\gamma \quad \quad (1 \leq i < j \leq N), 
\end{equation}
whereby the constant $c_\gamma$ can be determined from the Hardy--Littlewood--Sobolev inequality \cite{liebloss}. 
Incidentally, $c_0$ is the Rollnik norm of $\bigl(v_{ij}\bigr)_- $ \cite{quadrforms}.

Suppose all particles are partitioned into two non--empty clusters
$\mathfrak{C}_1$ and $\mathfrak{C}_2$ each containing $\# \mathfrak{C}_1$ and
$\# \mathfrak{C}_2$ particles respectively. Then we can write
\begin{gather}
 H = h_{\mathfrak{C}_1 \mathfrak{C}_2} - \Delta_R + V_{\mathfrak{C}_1 \mathfrak{C}_2}^+ - V_{\mathfrak{C}_1 \mathfrak{C}_2}^- , \label{31.07;4}\\
V_{\mathfrak{C}_1 \mathfrak{C}_2}^\pm := \sum_{\substack{i \in \mathfrak{C}_1 \\ j \in \mathfrak{C}_2 }}
\bigl(v_{ij}\bigr)_\pm , \label{31.07;22}
\end{gather}
where $h_{\mathfrak{C}_1 \mathfrak{C}_2}$ is the Hamiltonian of internal motion in the clusters, $R \in \mathbb{R}^3$ is a vector pointing from the
center of mass of $\mathfrak{C}_1$ to the center of mass of $\mathfrak{C}_2$ (for convenience we set in (\ref{31.07;4}) the coefficient in front  of the Laplace operator
equal to unity). $V_{\mathfrak{C}_1 \mathfrak{C}_2}^+ $ (resp. $V_{\mathfrak{C}_1 \mathfrak{C}_2}^- $) are the sums of positive (resp. negative) parts of
interactions between the clusters. For each partition we define
\begin{gather}
  H(\lambda) := h_{\mathfrak{C}_1 \mathfrak{C}_2} - \Delta_R + V_{\mathfrak{C}_1 \mathfrak{C}_2}^+ - \lambda V_{\mathfrak{C}_1 \mathfrak{C}_2}^-  \quad \quad (\lambda \geq
1), \label{31.07;12}\\
 E_{thr} (\lambda) := \inf \sigma_{ess} \bigl( H (\lambda)\bigr), \label{30.07;9}
\end{gather}
where, clearly, $H = H(1)$.  Our aim in this section is to prove the following
\begin{theorem}\label{30.07;1}
 Let $H$ be defined as in (\ref{30.07;8}). Suppose that there exists a partition into two clusters $\mathfrak{C}_{1,2}$ and  $\lambda_0 > 1$ such that
$E_{thr}(\lambda_0)=\inf \sigma_{ess} (H)$, where $E_{thr}(\lambda_0)$ is defined in (\ref{30.07;9}). Then $\#(\evs(H) < \inf \sigma_{ess} (H)) < \infty$.
\end{theorem}
The proof would be given later in this section, where at the end we would also
demonstrate that in the case when $\inf \sigma_{ess} (H) = 0$ Theorem~1.3 in \cite{vugzhisl1} (restricted to the class of potentials discussed here)
follows from Theorem~\ref{30.07;1}.

Let us introduce the following operator function
\begin{equation}
 D (\gamma, \epsilon) := \bigl[H_w + \epsilon \bigr]^{-1/2-\gamma} (V_{\mathfrak{C}_1 \mathfrak{C}_2}^-)^{1/2} ,
\end{equation}
where $\gamma\geq 0$, $\epsilon >0$ and for a  shorter notation we set
\begin{equation}
H_w := (h_{\mathfrak{C}_1 \mathfrak{C}_2} - E_{thr}(1)) -
\Delta_R + V_{\mathfrak{C}_1 \mathfrak{C}_2}^+  .
\end{equation}
The relevant properties of $D (\gamma, \epsilon)$ are established
in the following two lemmas.
\begin{lemma}\label{31.07;1}
For all $\gamma \in [0, \alpha(3-2\alpha)^{-1}]$  
\begin{equation}
\Lambda_\gamma := \sup_{\epsilon>0} \left\|   D(\gamma, \epsilon)\right\|  < \infty
\end{equation}
\end{lemma}
\begin{proof}
Note that for a self--adjoint operator $A \geq 0$ and $\epsilon, p >0$ we have
\begin{equation}\label{30.07;10}
 (A+ \epsilon)^{-p} = [p\Gamma(p)]^{-1} \int_0^\infty  \exp\{-t^{1/p}(A+ \epsilon)\} dt ,
\end{equation}
where the integral is to be understood in the strong sense. Eq.~(\ref{30.07;10}) can be easily verified using the spectral theorem.
Therefore, for any $f \in L^2 (\mathbb{R}^{3N-3}), \gamma \geq 0$
\begin{equation}
 \bigl(H_w + \epsilon \bigr)^{-1/2-\gamma} f = \left[\left(1/2 +
\gamma\right)\Gamma \left(1/2 + \gamma\right)\right]^{-1} \int_0^\infty
 \exp\{-t^{\frac 2{1+2\gamma}}(H_w + \epsilon)\}  f  dt
\end{equation}
The operator under the integral is positivity preserving \cite{lpestim}, which
allows us to write
\begin{equation}\label{3.8;1}
 \left| \bigl(H_w + \epsilon \bigr)^{-1/2-\gamma} f \right| \leq
\left[\left(1/2 + \gamma\right)\Gamma \left(1/2 +
\gamma\right)\right]^{-1}  \int_0^\infty
 \exp\{-t^{\frac 2{1+2\gamma}}(H_w + \epsilon)\}  |f|  dt
\end{equation}
By the Lie--Trotter product formula (see Sec. VIII of vol.~1 in \cite{reed})
\begin{equation}
e^{-t(H_w + \epsilon)} = \textnormal{s--}\negmedspace\lim_{\negthickspace
\negthickspace \negthickspace n \to \infty} \left( e^{-(t/n)(H'_w +
\epsilon)}e^{-(t/n)V_{\mathfrak{C}_1 \mathfrak{C}_2}^+}\right)^n ,
\end{equation}
where $H'_w := (h_{\mathfrak{C}_1 \mathfrak{C}_2} - E_{thr}(1)) -\Delta_R$. This gives us the inequality
\begin{equation}
 \left| \bigl(H_w + \epsilon \bigr)^{-1/2-\gamma} f \right| \leq \bigl(H'_w +
\epsilon \bigr)^{-1/2-\gamma} |f| .
\end{equation}
Hence,
\begin{gather}
\Lambda_\gamma \leq \sup_{\epsilon>0} \left\|  \bigl(H'_w + \epsilon
\bigr)^{-1/2-\gamma}  \bigl(V_{\mathfrak{C}_1 \mathfrak{C}_2}^-\bigr)^{1/2}\right\| \nonumber\\
  = \sup_{\epsilon>0} \left\|
\bigl(V_{\mathfrak{C}_1 \mathfrak{C}_2}^-\bigr)^{1/2} \bigl(H'_w + \epsilon \bigr)^{-1-2\gamma}
\bigl(V_{\mathfrak{C}_1 \mathfrak{C}_2}^-\bigr)^{1/2}\right\|^{1/2} . \label{31.07;5}
\end{gather}
By the spectral theorem
\begin{equation}\label{27.01/1}
  \bigl(H'_w + \epsilon \bigr)^{-1-2\gamma} \leq \bigl(-\Delta_R + \epsilon
\bigr)^{-1-2\gamma} . 
\end{equation}
(here we benefit from the fact that $[h_{\mathfrak{C}_1 \mathfrak{C}_2}, -\Delta_R] = 0$). Using
(\ref{27.01/1}) we get from (\ref{31.07;5})
\begin{gather}
\Lambda_\gamma^2  \leq  \sup_{\epsilon>0} \left\|  \bigl(V_{\mathfrak{C}_1 \mathfrak{C}_2}^-\bigr)^{1/2}
\bigl(-\Delta_R  + \epsilon \bigr)^{-1-2\gamma}  \bigl(V_{\mathfrak{C}_1 \mathfrak{C}_2}^-\bigr)^{1/2}\right\|
\nonumber \\
= \sup_{\epsilon>0} \left\|  \bigl(-\Delta_R + \epsilon \bigr)^{-1/2-\gamma}
V_{\mathfrak{C}_1 \mathfrak{C}_2}^- \bigl(-\Delta_R + \epsilon \bigr)^{-1/2-\gamma} \right\|\nonumber\\
\leq \bigl(\# \mathfrak{C}_1 \bigr) \bigl(\# \mathfrak{C}_2 \bigr) \max_{i \in
\mathfrak{C}_1 , j \in \mathfrak{C}_2} \sup_{\epsilon>0} \left\|
\bigl(v_{ij}\bigr)_-^{1/2} \bigl(-\Delta_R + \epsilon \bigr)^{-1-2\gamma}
\bigl(v_{ij}\bigr)_-^{1/2}\right\| . \label{31.07;6}
\end{gather}
Note that for $i \in  \mathfrak{C}_1 , j \in  \mathfrak{C}_2$ we can write $r_j
- r_i = R + \sum_k \mathbf{m}^{(ij)}_k x_k$, where $\mathbf{m}^{(ij)}_k $ are
real
coefficients depending on masses and $x_k \in \mathbb{R}^3$ for $k = 1, \ldots, N-2$ are intercluster coordinates. (It is easy to see that the coefficient in
front of $R$ is always 1 by fixing all $|x_k|$ and taking $|R|\gg 1$).
Thus we can trivially estimate the norm on the rhs of (\ref{31.07;6})
\begin{equation}
 \left\|  \bigl(v_{ij}\bigr)_-^{1/2} \bigl(-\Delta_R + \epsilon
\bigr)^{-1-2\gamma}  \bigl(v_{ij}\bigr)_-^{1/2}\right\|^2  \leq \int \bigl(v_{ij}(R)\bigr)_-
G_\gamma^2 (\epsilon; R-R') \bigl(v_{ij}(R')\bigr)_- d^3R d^3R' , \label{31.07;7}
\end{equation}
where $G_\gamma (\epsilon; R-R')$ is the integral kernel of the operator
$\bigl(-\Delta_R + \epsilon \bigr)^{-1-2\gamma}$ on $L^2(\mathbb{R}^3)$, which is positive \cite{reed,lpestim}.
Using
(\ref{30.07;10}) and the formula on p. 59 in \cite{reed}, vol. II (c.f . the formula on the top of
page 262 in \cite{klaus1}) this integral kernel can be written
as
\begin{gather}
 G_\gamma (\epsilon; R) = (4\pi)^{-3/2} [p\Gamma(p)]^{-1} \int_0^\infty t^{-\frac{3}{2p}}
\exp\{-\epsilon t^{\frac 1p} -2^{-2}|R|^2 t^{-\frac 1p}\}  dt\nonumber\\
\leq (4\pi)^{-3/2} [p\Gamma(p)]^{-1} \frac{8p2^{-2p}}{|R|^{3-2p}} \int_0^\infty y^{-p + \frac 12}e^{-y} =\frac{2^{-2p}\Gamma( 3/2 -p )}{\pi^{3/2}\Gamma(p)}  \frac{1}{|R|^{3-2p}} , \label{2.10;1}
\end{gather}
where $p= (1+2\gamma)$ and in the integral we used the substitution 
$y = |R|^2 /(4 t^{1/p})$. 
Substituting this upper bound into (\ref{31.07;7}) and using (\ref{30.07;4}) finishes the proof. 
\end{proof}
Let us remark that 
\begin{equation}\label{5.10;1}
 \sup_{\epsilon>0} \left\|  \bigl[H_w + \epsilon \bigr]^{-1/2} (V_{\mathfrak{C}_1 \mathfrak{C}_2}^-)\right\|  < \infty. 
\end{equation}
Eq.~(\ref{5.10;1}) follows from the proof of Lemma~\ref{31.07;1} since for all $1 < i \leq j \leq N$ 
\begin{equation}\label{5.10;2}
 \sup_{\epsilon >0} \left\|\bigl(v_{ij}(r)\bigr)_- \bigl(-\Delta_r + \epsilon\bigr)^{- \frac 12}\right\| < \infty ,  
\end{equation}
where the norm on the rhs is that of $L^2 (\mathbb{R}^3)$. To check that (\ref{5.10;2}) holds note that 
$(-\Delta_r + \epsilon)^{-1/2} f = h * f$ for any $f \in L^2 (\mathbb{R}^3)$, where by (\ref{2.10;1}) 
\begin{equation}
 h(r) \leq \frac 1{2 \pi^2 |r|^2}  \int_0^\infty e^{-y} dy = \frac 1{2 \pi^2 |r|^2} 
\end{equation}
By the Sobolev's inequality (eq.~(4.2) in \cite{traceideals}) 
\begin{equation}\label{9.10;1}
 \left\|\bigl(v_{ij}(r)\bigr)_- \bigl(-\Delta_r + \epsilon\bigr)^{- \frac 12}\right\| \leq  \left(\frac{2}{3\pi^2}\right)^{2/3} C_3\|\bigl(v_{ij}\bigr)_-\|_3 
\end{equation}
(for an estimation of the constant $C_3$ see f. e. \cite{liebloss}, where this inequality is called the weak Young inequality). 
\begin{lemma}\label{31.07;2}
The function $D(0, \epsilon)$ is norm--continuous for $\epsilon >0 $ and has a norm limit for $\epsilon\to +0$.
\end{lemma}
\begin{proof}
$D(0, \epsilon)$ is uniformly norm--bounded for $\epsilon >0$ by Lemma~\ref{31.07;1}. Below we prove the following
inequality
\begin{equation} 
 \left\| D(0, \epsilon_2 ) - D(0, \epsilon_1)\right\| \leq \frac{\Lambda_{\gamma_0} |\epsilon_1 -
\epsilon_2|^{1/2}}{\bigl(\max \{\epsilon_1, \epsilon_2\} \bigr)^{1/2 - \gamma_0}} \quad \quad (\textnormal{for
$\epsilon_{1,2} >0$}), \label{xwdcont}
\end{equation}
where $\gamma_0 := \alpha(3-2\alpha)^{-1}$ and $\Lambda_{\gamma_0} $ is defined in Lemma~\ref{31.07;1}. From (\ref{xwdcont}) it follows that the
norm limit
$D(0, 0):= \lim_{\epsilon\to +0} D(0, \epsilon)$ exists because it is a
norm--limit of a Cauchy sequence (the norm--continuity 
is also a trivial consequence of (\ref{xwdcont})). For $\epsilon_2 \geq \epsilon_1 >0$ we can write
\begin{gather}
D(0, \epsilon_1) -  D(0, \epsilon_2) = \Bigl[ \bigl(H_w + \epsilon_2\bigr)^{1/2}
-
\bigl(H_w + \epsilon_1\bigr)^{1/2}  \Bigr] \bigl(H_w + \epsilon_1 \bigr)^{-1/2}
\bigl(H_w + \epsilon_2\bigr)^{-1/2}\bigl(V_{\mathfrak{C}_1 \mathfrak{C}_2}^-\bigr)^{1/2} \nonumber\\
= \Bigl[ \bigl(H_w + \epsilon_2\bigr)^{1/2}  -
\bigl(H_w + \epsilon_1\bigr)^{1/2}  \Bigr] \bigl(H_w + \epsilon_2\bigr)^{-1/2 +
\gamma_0} \bigl(H_w + \epsilon_1 \bigr)^{-1/2} \bigl(H_w +
\epsilon_2\bigr)^{-\gamma_0} \bigl(V_{\mathfrak{C}_1 \mathfrak{C}_2}^-\bigr)^{1/2}.  \label{xwnonu}
\end{gather}
The following inequality for $f \in D(H_0)$ is a trivial consequence of the spectral theorem
\begin{equation}
 \left\| \Bigl[ \bigl(H_w + \epsilon_2\bigr)^{1/2}  -
\bigl(H_w + \epsilon_1\bigr)^{1/2} \Bigr] f \right\| \leq |\epsilon_2 -
\epsilon_1|^{1/2} \|f\| .
\end{equation}
By the same reasons
\begin{equation}
 \| \bigl(H_w + \epsilon_2 \bigr)^{-1/2 + \gamma_0} \| \leq |\epsilon_2|^{-1/2 +
\gamma_0}.
\end{equation}
Thus from (\ref{xwnonu})
\begin{equation}
 \left\| D(0, \epsilon_2 ) - D(0, \epsilon_1)\right\| \leq |\epsilon_2 -
\epsilon_1|^{1/2} |\epsilon_2|^{-1/2 + \gamma_0} \left\| \bigl(H_w + \epsilon_1
\bigr)^{-1/2} \bigl(H_w + \epsilon_2\bigr)^{-\gamma_0}
\bigl(V_{\mathfrak{C}_1 \mathfrak{C}_2}^-\bigr)^{1/2}\right\| .\label{xwbotw2}
\end{equation}
For the norm on the rhs we can write
\begin{gather}
 \left\| \bigl(H_w + \epsilon_1 \bigr)^{-1/2} \bigl(H_w +
\epsilon_2\bigr)^{-\gamma_0} \bigl(V_{\mathfrak{C}_1 \mathfrak{C}_2}^-\bigr)^{1/2}\right\|^2 \nonumber\\
= \left\| \bigl(V_{\mathfrak{C}_1 \mathfrak{C}_2}^-\bigr)^{1/2} \bigl(H_w + \epsilon_1\bigr)^{-1/2}  \bigl(H_w +
\epsilon_2 \bigr)^{-2\gamma_0} \bigl(H_w + \epsilon_1\bigr)^{-1/2}
\bigl(V_{\mathfrak{C}_1 \mathfrak{C}_2}^-\bigr)^{1/2}\right\| \nonumber\\
\leq  \left\| \bigl(V_{\mathfrak{C}_1 \mathfrak{C}_2}^-\bigr)^{1/2}  \bigl(H_w + \epsilon_1 \bigr)^{-1-2\gamma_0}
\bigl(V_{\mathfrak{C}_1 \mathfrak{C}_2}^-\bigr)^{1/2}\right\| =  \left\| \bigl(H_w + \epsilon_1
\bigr)^{-1/2-\gamma_0}  \bigl(V_{\mathfrak{C}_1 \mathfrak{C}_2}^-\bigr)^{1/2}\right\|^2 \leq \Lambda^2_{\gamma_0}
\label{xwbotw}.
\end{gather}
where we used the inequality $ \bigl(H_w + \epsilon_2 \bigr)^{-2\gamma_0} \leq
\bigl(H_w + \epsilon_1 \bigr)^{-2\gamma_0}$,
which again follows from the spectral theorem.
Now (\ref{xwdcont}) follows from (\ref{xwbotw2}) and (\ref{xwbotw}).
\end{proof}

The following lemma is trivial.
\begin{lemma}\label{30.07;2}
 Let $A$ be a Hilbert--Schmidt operator on a Hilbert space $\mathcal{H}$.
Suppose there is $\delta >0$ and
an orthonormal set $\varphi_1 , \ldots, \varphi_n \in \mathcal{H}$ such that
$|(\varphi_i, A\varphi_i)| \geq \delta$ for $i = 1,\ldots, n$.
Then $n \leq \|A\|_{HS}^2/\delta^2$.
\end{lemma}
\begin{proof}
From Lemma's conditions it follows that
\begin{equation}
 \delta^2 \leq |(\varphi_i , A\varphi_i)|^2 \leq \|A \varphi_i\|^2 = (\varphi_i
, A^*  A \varphi_i)
\end{equation}
From the last inequality it follows that
\begin{equation}
 n \delta^2 \leq \sum_{i=1}^n (\varphi_i , A^*  A \varphi_i) \leq \|A\|_{HS}^2 .
\end{equation}
\end{proof}

Now we pass to the
\begin{proof}[Proof of Theorem~\ref{30.07;1}]
The Hamiltonian in (\ref{31.07;12}) has the form $H(\lambda) = H_w - \lambda V_{\mathfrak{C}_1 \mathfrak{C}_2}^-$. 
Before we apply Theorem~\ref{31.07;16} we need to verify its conditions. Note that $D(V_{\mathfrak{C}_1 \mathfrak{C}_2}^-)\supseteq D(H_w^{\frac 12})$ due to (\ref{5.10;1}) 
(see a first Remark in Sec.~\ref{31.07;14}). Besides, $H_w + \mu V_{\mathfrak{C}_1 \mathfrak{C}_2}^-$ is self--adjoint on $D(H_0) = D(H_w)$ for all $\mu \in [0, \infty)$. 

The
associated BS operator for $\epsilon >0$ is given by 
\begin{equation}
 K(\epsilon)= \bigl(H_w + \epsilon \bigr)^{-1/2} V_{\mathfrak{C}_1 \mathfrak{C}_2}^- \bigl(H_w + \epsilon
\bigr)^{-1/2} = D(0, \epsilon)  D^* (0, \epsilon).
\end{equation}
From Lemmas~\ref{31.07;1}, \ref{31.07;2} it follows that $K(\epsilon)$ is norm--continuous on $[0, \infty)$. 
Using that
$\sigma_{ess} \bigl( H(\lambda_0)\bigr) = [E_{thr}(1), \infty)$ and Theorem~\ref{31.07;16} we
conclude that
\begin{equation}\label{xwthisholds}
 \sigma_{ess} \bigl( K(\epsilon)\bigr) \cap \left(\lambda^{-1}_0 , \infty\right)
= \emptyset
\end{equation}
for $\epsilon >0$.
By Theorem~9.5 in \cite{weidmann} (\ref{xwthisholds}) holds also for $\epsilon =
0$. Hence, due to $\lambda_0 >1$ there must exist $\delta > 0$ such that
\begin{equation}\label{xwthisholds4}
 K(0) \leq 1-2\delta + \mathcal{C}_f ,
\end{equation}
where $\mathcal{C}_f$ is a finite rank operator. Now let us assume by
contradiction  that $\#(\evs(H(1)) < E_{thr}(1))$ is infinite. Then by the
BS principle
\begin{equation}\label{xwthisholds2}
 \lim_{\epsilon \to + 0} \#\bigl(\evs(K(\epsilon)) > 1 \bigr) \to \infty .
\end{equation}
Due to the norm--continuity of $K(\epsilon)$ from (\ref{xwthisholds2}) it follows that for any
$n = 1, 2, \ldots$ one can find an orthonormal set $\psi_1, \psi_2 , \ldots ,
\psi_n \in L^2 (\mathbb{R}^{3N-3})$
such that
\begin{equation}\label{xwthisholds3}
 \bigl( \psi_i, K(0)\psi_i\bigr) \geq 1-\delta \quad \quad (i = 1, 2, \ldots, n).
\end{equation}
Then (\ref{xwthisholds4}) gives
\begin{equation}
  \bigl( \psi_i, \mathcal{C}_f \psi_i\bigr) \geq \delta \quad \quad (i = 1, 2, \ldots, n).
\end{equation}
By Lemma~\ref{30.07;2} $n \leq \|\mathcal{C}_f \|^2_{HS}/\delta^2$,
which contradicts $n$ being arbitrary positive integer.
\end{proof}

Let us briefly show that Theorem~1.3 in \cite{vugzhisl1} follows from Theorem~\ref{30.07;1}. Suppose that $d_1 \subset \{1, \ldots, N\}$ and $d_2 \subset \{1, \ldots, N\}$ are
two nonempty disjoint clusters. Similar to (\ref{31.07;12})--(\ref{30.07;9}) the Hamiltonian of the subsystem $d_1 \cup d_2$ can be written as
\begin{equation}\label{31.07;21}
 H_{d_1 d_2} (\lambda) := h_{d_1 d_2} - \Delta_{r_{12}} + V_{d_1 d_2}^+ - \lambda V_{d_1 d_2}^- ,
\end{equation}
where $r_{12} \in \mathbb{R}^3$ points from the center of mass of $d_1$ in the direction to the center of mass of $d_2$; the scale is chosen so as to make
(\ref{31.07;21}) hold. The meaning of other notations is clear from (\ref{31.07;4})--(\ref{31.07;22}).
The Hamiltonian (\ref{31.07;21}) acts on $L^2 (\mathbb{R}^n)$, where $n = 3 \bigl( (\# d_1) + (\# d_2) - 1\bigr)$.
Putting aside the comparison of restrictions on the potentials, Theorem~1.3
in \cite{vugzhisl1} can be equivalently reformulated as follows
\begin{theorem}[S. Vugal'ter and G. Zhislin 1986]
 Suppose the Hamiltonian in (\ref{30.07;8}) is  such that $\sigma_{ess} (H) = [0, \infty)$. Suppose there exists a partition in two clusters $\mathfrak{C}_{1,2}$ 
where for all 
subsystems $d_1 \cup d_2$ such that $d_1 \subseteq \mathfrak{C}_1$, $d_2 \subseteq \mathfrak{C}_2$ and $(\#d_1) + (\#d_2) \leq N-1$,
the Hamiltonian $H_{d_1 d_2} (1)$ does not have a virtual level at zero energy. Then $\#(\evs(H) < 0) < \infty$.
\end{theorem}
\begin{proof}
We need to show that the conditions of Theorem~\ref{30.07;1} are fulfilled. We would say that the subsystem $d_1 \cup d_2$ is at critical coupling if
$H_{d_1 d_2} (1) \geq 0$ and $H_{d_1 d_2} (1+ \epsilon) \ngeqq 0$ for $\epsilon >0$ (this is different from the Definition~1 in \cite{jpasubmit}; for the definition of virtual
level see, for example, Definition~3 in \cite{jpasubmit}). By the HVZ theorem
conditions of Theorem~\ref{30.07;1} would be verified if we can prove that
$d_1 \cup d_2$ is not at critical coupling for all $d_{1,2}$ described in the conditions of the theorem to prove. Assume
by contradiction that there exist $d_{1,2}$ such that $H_{d_1 d_2}$ is at critical coupling.
Without loosing generality we can assume that the subsystems $d'_1 \cup d'_2$, where
$d'_1 \subseteq d_1$, $d'_2 \subseteq d_2$ and $(\#d'_1) + (\#d'_2) < (\#d_1) + (\#d_2) $ are \textit{not} at critical
coupling (otherwise we can pass to an appropriate sub/subsystem).
Thus there must exist $\omega > 0$ such that $\sigma_{ess} \bigl( H_{d_1 d_2} (1+\omega)\bigr) = [0, \infty)$.
For $d_1 \cup d_2$ we construct the BS operator
\begin{equation}
 K(\epsilon) := \left[h_{d_1d_2} - \Delta_{r_{12}} + V_{d_1 d_2}^+ + \epsilon \right]^{-\frac 12} V_{d_1 d_2}^-
\left[h_{d_1d_2} - \Delta_{r_{12}} + V_{d_1 d_2}^+ + \epsilon \right]^{-\frac 12} . 
\end{equation}
By the analysis above $K(\epsilon)$ is positivity preserving, $K(\epsilon) \to K(0)$ in norm, where $K(0)$ is a positivity preserving operator as well. 
By Theorem~\ref{31.07;16} and Theorem~9.5 in \cite{weidmann}
$ \sigma_{ess} \bigl(K(0)\bigr) \subseteq [0 ,(1+\omega)^{-1}]$.
Because $d_1 \cup d_2$ is at critical coupling we have $\|K(0)\| = 1$ (this follows from the BS principle and norm--continuity of $K(\epsilon)$).
Due to the location of the essential spectrum  $\|K(0)\|$ is an eigenvalue. 
Thus there exists $\phi_0 \in L^2 (\mathbb{R}^n)$ such that $K(0)\phi_0 = \phi_0$ and $\phi_0 >0$, see \cite{reed}. From this fact and
from the variational principle it follows that for all $a >0$ there exists $\epsilon > 0$ such that $\|K' (\epsilon,a) \| >1$, where
\begin{gather}
 K' (\epsilon,a) := \left[h_{d_1d_2} - \Delta_{r_{12}} + V_{d_1 d_2}^+ + \epsilon \right]^{-\frac 12} \left\{ V_{d_1 d_2}^- + ae^{-|x|} \right\}
\left[h_{d_1d_2} - \Delta_{r_{12}} + V_{d_1 d_2}^+ + \epsilon \right]^{-\frac 12}.
\end{gather}
By the BS principle this means that $H_{d_1 d_2} - ae^{-|x|} \ngeq 0$ for any $a>0$. Now it is clear that $H_{d_1 d_2}$ has a virtual level at zero energy contrary to the conditions of the
theorem.
\end{proof}

\section{Main Result and Discussion}\label{25.9;15}

We shall consider the Hamiltonian of $N$ particles in $\mathbb{R}^3$
\begin{gather}
 H = H_0 +  V \label{1.08;3}\\
V = \sum_{i< k} v_{ik},   \quad\quad \quad v_{ik} \in L^1 (\mathbb{R}^3) \cap L^3 (\mathbb{R}^3) \label{1.08;4},
\end{gather}
where $H_0$ is the kinetic energy operator with the center of mass removed and
$v_{ik}$ are operators of multiplication by $v_{ik} (r_i - r_k)$. 
By Kato--Rellich's theorem \cite{kato,teschl}
$H$ is self--adjoint on $D(H_0) = \mathcal{H}^2 (\mathbb{R}^{3N-3}) \subset L^2 (\mathbb{R}^{3N-3})$.

For an ordered multi--index $\pmb{i} = \{i_1, \ldots , i_s\}$, where $1 \leq s \leq N$ and $1 \leq i_1 < i_2 < \cdots < i_s \leq N$ let us define the set
$S_{\pmb{i}} = \{1, 2, \ldots, N\} / \{i_1, \ldots , i_s\}$, which is the subsystem containing $N - \# \pmb{i}$ particles.
The sums of pair--interactions for the subsystem $S_{\pmb{i}}$ are defined as
\begin{gather}
 V_{\pmb{i}} := \sum_{\substack{j < k\\ j,k \in S_{\pmb{i}}}} v_{jk} \label{7.8;1a} \\
V^{\pm}_{\pmb{i}} := \sum_{\substack{j < k\\ j,k \in S_{\pmb{i}}}} (v_{jk})_\pm \label{7.8;1}
\end{gather}
In particular, $V_{\{j\}}$ is the sum of pair--interactions in
the subsystem, where particle $j$ is removed; $V_{\{j, s\}}$ is the sum of
pair--interactions in the subsystem,
where particles $j$ and $s$ are removed, etc. Obviously, $V_{\pmb{i}} = V^{+}_{\pmb{i}} - V^{-}_{\pmb{i}}$.
We shall make the following assumption
\begin{list}{R\arabic{foo}}
{\usecounter{foo}
    \setlength{\rightmargin}{\leftmargin}}
\item $\sigma_{ess} (H) = [0, \infty)$. There exists $\omega >0$ such that $H_0 + V_{\{j, s\}}^+ -
(1+\omega)V_{\{j, s\}}^- \geq 0$ for all $1  \leq j < s \leq N $.
\end{list}
In particular, R1 implies that a
subsystem containing $N-2$ or less particles is not at critical coupling \cite{jpasubmit}.
In Sec.~\ref{1.08;2} we prove the following
\begin{theorem}\label{1.08;9}
 Suppose that $H$ defined in (\ref{1.08;3})--(\ref{1.08;4}) satisfies R1 and $N \geq 4$. Then
$\#(\evs(H) < 0)$ is finite.
\end{theorem}
The first thing worth noting is that Theorem~\ref{1.08;9} does not hold for $N=3$ because
of the Efimov effect \cite{vefimov,yafaev,sobolev,merkuriev}.
\begin{corollary}\label{1.08;10}
 Suppose the system of 4 identical particles is described by the Hamiltonian $H$ in (\ref{1.08;3})--(\ref{1.08;4}) and $\sigma_{ess} (H) = [0, \infty)$.
Then the number of negative
energy bound states is finite.
\end{corollary}
\begin{proof}
 We need only to check the second part of R1. If one pair of particles would be
at critical coupling then this would be true for all particle pairs because the
particles are
identical. Then due to the Efimov effect \cite{yafaev,tamura} three--particle subsystems would have negative energy bound states thereby violating the condition
$\sigma_{ess} (H) = [0, \infty)$.
Therefore, none of the particle pairs is at critical coupling and Theorem~\ref{1.08;9} applies.
\end{proof}
Unfortunately, we did not succeed in extending Corollary~\ref{1.08;10} to $N\geq 5$. To
explain the difficulty let us consider $N=5$. Like in the proof of Corollary~\ref{1.08;10} we
conclude that
none of the particle pairs is at critical coupling. It can happen, however, that
all particle triples would be at critical coupling, each of them having a zero
energy
bound state \cite{1}. It is natural to assume that in most cases this would lead to negative
energy bound states in 4--particle subsystems. But it is unclear how to prove
that even for negative pair--interactions.

 It is natural to ask whether instead of assumption R1 in the condition of Theorem~\ref{1.08;9} one could simply require $\sigma_{ess} (H) = [0, \infty)$.
In this regard we pose the following conjecture
\begin{conjecture}\label{1.08;14}
 Suppose that  in (\ref{1.08;3})--(\ref{1.08;4}) $N=4$ and all $v_{ik} \leq 0$ are bounded and finite--range potentials. Suppose also that $\sigma_{ess} (H) = [0, \infty)$,
the particle pair $\{1,2\}$ and the particle triples $\{1,2,3\}$ and $\{1,2,4\}$ are at critical coupling (in the sense of Definition~1 in \cite{jpasubmit}),
and the subsystems $\{1,3,4\}$ and $\{2,3,4\}$ do not have zero energy bound states. Then $H$
has an infinite number of negative energy bound states.
\end{conjecture}
The conditions in Conjecture~\ref{1.08;14} can always be met by appropriate tuning of the coupling constants, see Sec.~6 in \cite{1}. Note, that by Theorem~3 in \cite{1}
none of the subsystems has zero energy bound states and, therefore, the no--clustering theorem applies (see Theorem~3 in \cite{jpasubmit}).
Conjecture~\ref{1.08;14}, if true, would mean the existence of a ``true'' Efimov effect for
four particles (``true'' means that it does not trivially reduce to the case of three clusters).

Let us remark that using Theorem~2 from \cite{jpasubmit} Theorem~\ref{1.08;9} can be reformulated in the following way
\begin{theorem}\label{1.08;9aa}
 Suppose that $N \geq 4$ and $H$ in (\ref{1.08;3})--(\ref{1.08;4}) is such that $\sigma_{ess} (H) = [0, \infty)$. Suppose also that
none of the particle pairs is at critical coupling and none of the subsystems consisting of $N-2$ or less particles has a square--integrable zero energy ground state.
Then $H$ has a finite number of bound states with negative energies.
\end{theorem}

Let us now discuss how the material is arranged in the next sections. The main
tool of our analysis is the BS operator, see Sec.~\ref{31.07;14}. The somewhat uncommon form of the BS operator, which we
adopt in this paper, has an advantage that the BS operator of an $N$--particle
system can be expressed through the BS operators of the subsystems. In Sec.~\ref{1.08;1} we
analyze the spectrum of the BS operator, which corresponds to the
$N$--particle system at critical coupling, whose subsystems are not at critical coupling. From Theorem~2 in \cite{jpasubmit} we know that the Hamiltonian of 
such system has eigenvalue equal to zero. Here of special
interest is the behavior of the  eigenfunction corresponding to the largest positive
eigenvalue of the BS operator. Sec.~\ref{1.08;2} is devoted to the proof of Theorem~\ref{1.08;9}.
From Sec.~\ref{1.08;33} we already know that in proving Theorem~\ref{1.08;9}
we need to focus on the case when some of the
 $N-1$--particle subsystems are at critical coupling (otherwise the proof is
 accomplished by applying Theorem~\ref{30.07;1}). For these subsystems we shall need the results of Sec.~\ref{1.08;1}.

We define the BS operator associated with (\ref{1.08;3}) as
\begin{equation}\label{1.08;35}
 K(\epsilon) := \bigl(H_0 + \epsilon \bigr)^{-1/2} V \bigl(H_0 + \epsilon
\bigr)^{-1/2} \quad \quad (\epsilon > 0) ,
\end{equation}
The main object of our interest is
\begin{equation}\label{7.8;55}
 N_\epsilon := \# (\textrm{evs}(H) < -\epsilon)= \# (\textrm{evs}(K(\epsilon)) >
1),
\end{equation}
where the last equation follows from the BS principle. (The applicability of Theorem~\ref{31.07;16} can always be checked in the same way it is done 
in the proof of Theorem~\ref{30.07;1}). 
\begin{lemma}\label{1.8;54}
One can define $K(0)$ so that $K(\epsilon)$ is
norm--continuous on $[0, \infty)$.
\end{lemma}
\begin{proof}
We can write
\begin{equation}\label{1.08;32}
 K(\epsilon) = \sum_{1 \leq i < k\leq N} d_{ik} (\epsilon) \sign (v_{ik})
d^*_{ik} (\epsilon) ,
\end{equation}
where
\begin{equation}
 d_{ik} (\epsilon) := \bigl(H_0 + \epsilon\bigr)^{-1/2} |v_{ik}|^{1/2} .
\end{equation}
and $\sign (v_{ik})$ is the operator of multiplication by the sign of $v_{ik}$.
Repeating the
arguments from the proof of Lemma~\ref{31.07;2} we prove the following inequality analogous to (\ref{xwdcont}) (the restrictions on the pair--potentials 
allow us to set $\gamma_0 = 1/10$)
\begin{equation}\label{dik}
 \|d_{ik} (\epsilon_1) - d_{ik} (\epsilon_2)\|  \leq \Lambda_{\frac 1{10}} |\epsilon_2 -
\epsilon_1|^{1/2} \epsilon_2^{-2/5},
\end{equation}
where $\epsilon_2 \geq \epsilon_1 >0$. From (\ref{dik}) it
follows that $d_{ik} (\epsilon)$, and, hence, $K(\epsilon)$ is norm continuous on $[0, \infty)$. 
\end{proof}

\section{N--Particle System at Critical Coupling}\label{1.08;1}

Here we shall analyze the BS operator (\ref{1.08;35}) in the case when $H$ defined in (\ref{1.08;3})--(\ref{1.08;4}) 
is at critical coupling and has a bound state at zero energy.
Let us define $\mu : \mathbb{R}_+  \to \mathbb{R}$ 
\begin{equation}\label{1.08;41}
 \mu(\epsilon):= \sup \sigma\bigl( K (\epsilon)\bigr) .
\end{equation}
We shall make the following assumption
\begin{list}{R\arabic{foo}}
{\usecounter{foo}
    \setlength{\rightmargin}{\leftmargin}}\setcounter{foo}{1}
\item $H \geq 0$. There exists $\omega >0$ such that $H_0 + V_{\{j\}}^+
- (1+\omega)V_{\{j\}}^- \geq 0$ for $j = 1, \ldots,  N $. Besides, $H_0 + \delta V \ngeqq 0$ for all $\delta >0$. 
\end{list}
By Theorem~2 in \cite{jpasubmit} $H$ satisfying R2 has zero as an eigenvalue.
\begin{theorem}\label{3.8;3}
Suppose $H$ defined in (\ref{1.08;3})--(\ref{1.08;4}) satisfies R2.
Then there is $\beps > 0$ such that $\mu (\epsilon)$ defined in (\ref{1.08;41}) is an eigenvalue of $K(\epsilon)$ for $\epsilon \in [0, \beps]$.
As eigenvalue $\mu(\epsilon)$ is isolated and non--degenerate, as a function it is continuous and monotone decreasing on $[0,
\beps]$. Besides, $\mu(0) = 1$, $\mu(\beps) \geq (1+\omega/2)^{-1}$ and there is $a_\mu >0$ such that
\begin{equation}\label{mubound}
 1 - \mu(\epsilon) \geq a_\mu \epsilon \quad \quad \textnormal{for
$\epsilon \in [0, \beps]$}. 
\end{equation}
\end{theorem}
\begin{proof}
By R2 and the HVZ theorem $\sigma_{ess} (H_0 + (1+\omega) V) = [0, \infty)$. By Theorem~\ref{31.07;16} and Theorem~9.5 in \cite{weidmann}
\begin{equation}
 \sigma_{ess} (K(\epsilon)) \cap ((1+\omega)^{-1} , \infty) = \emptyset \quad
\quad \textnormal{for  $\epsilon \in [0, \infty)$}.
\end{equation}
On one hand, from the BS principle and the condition $H \geq 0$ it follows that
\begin{equation}
 \sigma \bigl(K(\epsilon)\bigr) \cap (1 , \infty) = \emptyset \quad
\quad \textnormal{for  $\epsilon >0$}.\label{1.8;51}
\end{equation}
On the other hand, since $H$ is at critical coupling, there must exist the sequences $\omega_n \to +0$, $\epsilon_n \to +0$, where $\omega_n < \omega$, such that
$\inf \sigma \bigl(H_0 + (1+\omega_n)V\bigr) = - \epsilon_n$. Hence, by the BS principle
\begin{equation}
 \sigma \bigl( (1+\omega_n)K(\epsilon_n)\bigr) \cap (1, \infty) \neq \emptyset. \label{1.8;52}
\end{equation}
Comparing (\ref{1.8;51}) and (\ref{1.8;52}) and using the continuity of $K(\epsilon)$ (Lemma~\ref{1.8;54}) we conclude
that $\mu(0) = 1$ is an eigenvalue of $K(0)$ lying aside from the essential
spectrum. (The non--degeneracy of this eigenvalue would easily follow from $V \leq 0$, for in this case $K(\epsilon)$
is a positivity
preserving operator).
Let us assume by contradiction that the eigenvalue $\mu(0)$ is degenerate. Then by
continuity for $\epsilon_n \to +0$ there must exist a sequence $\mu_n \nearrow
1$ such that
$K (\epsilon_n)$ has at least two eigenvalues in the interval $(\mu_n, 1)$.
We can choose the sequences $\epsilon_n, \mu_n$ so that $(1+\omega)^{-1} < \mu_n$. Thereby we guarantee that
\begin{equation}\label{23.03:01}
 \sigma_{ess} \left(H_0 + V^+ - \mu_n^{-1} V^- \right) = [0, \infty),
\end{equation}
where by definition
\begin{equation}
 V^{\pm} := \sum_{i < j} \bigl( v_{ij} \bigr)_\pm  .
\end{equation}
Since $\mu_n^{-1} K(\epsilon_n)$ has at least 2 eigenvalues in the interval $(1,
\infty)$, by the BS principle (Theorem~\ref{31.07;16}) the operator
\begin{equation}
 H_0 + \mu^{-1}_n \bigl(V^+ - V^-   \bigr)
\end{equation}
has at least 2 eigenvalues in the interval $(-\infty, -\epsilon_n)$. From the
operator inequality
\begin{equation}\label{31.01/1}
 H_0 + \mu^{-1}_n \bigl(V^+ - V^-   \bigr) \geq  H_0 + V^+ - \mu^{-1}_n V^- ,
\end{equation}
maxmin principle and (\ref{23.03:01}) it follows that the operator on the rhs of
(\ref{31.01/1}) has at least 2 eigenvalues in the interval $(-\infty,
-\epsilon_n)$.
The BS operator associated with  the operator on the rhs of (\ref{31.01/1}) is $
\mu_n^{-1} \tilde K(\epsilon)$, where
\begin{equation}
\tilde K (\epsilon) =\left[H_0 + V^+ +
\epsilon\right]^{-1/2}\bigl(V^-\bigr)\left[H_0 + V^+ + \epsilon\right]^{-1/2}.
\end{equation}
By the BS principle
\begin{equation}\label{23.03:02}
 \# (\textrm{evs}(\tilde K(\epsilon_n)) > \mu_n) \geq 2.
\end{equation}
The proof that $\tilde K(\epsilon)$ is norm--continuous on $[0, \infty)$ repeats that of Lemma~\ref{1.8;54}. Using the inequality
\begin{equation}
\sigma_{ess} (H_0 + V^+ - (1+\omega) V^- ) = [0, \infty),
\end{equation}
Theorem~\ref{31.07;16} and Theorem~9.5 in \cite{weidmann}  we conclude that
\begin{equation}
  \sigma_{ess} (\tilde K(\epsilon)) \cap ((1+\omega)^{-1} , \infty) = \emptyset
\quad \quad (\textnormal{for  $\epsilon \in [0, \infty]$} ).
\end{equation}

Repeating the arguments in the beginning of the proof we infer that $\|\tilde K
(0) \| = 1$ is an eigenvalue of $\tilde K(0)$. By norm--continuity and
(\ref{23.03:02})
we know that this eigenvalue must be at least two--fold degenerate. However,
$\tilde K(\epsilon)$ for $\epsilon >0$ is a product of positivity preserving operators
(c. f. formula (\ref{3.8;1})), and $\tilde K(0)$ is also positivity
preserving being the norm limit of positivity preserving operators. By
Theorem~XIII.43 in vol.~4 \cite{reed} $\|\tilde K(0)\|$ must be a non--degenerate eigenvalue, a
contradiction. The existence of $\beps > 0$ such that $\mu(\epsilon)$ is continuous on $[0, \beps]$ and $\mu(\beps) \geq (1+\omega/2)^{-1}$ is a trivial consequence
of the norm--continuity of $K(\epsilon)$. Thus there exists $\varphi(\epsilon) : \mathbb{R}_+ \to L^2 (\mathbb{R}^{3N-3})$ such that
\begin{equation}\label{3.8;5}
 K(\epsilon) \varphi (\epsilon) = \mu (\epsilon) \varphi(\epsilon) \quad \quad (\| \varphi (\epsilon) \| = 1, \epsilon \in [0, \beps]) 
\end{equation}
and by definition $\varphi(\epsilon) \equiv 0$ for $\epsilon \in (\beps, \infty)$. 
Let us define
\begin{equation}
 \psi (\epsilon) := \bigl( H_0 + \epsilon\bigr)^{-1/2} \varphi(\epsilon) \quad \quad (\epsilon >0). \label{11.01.12/1}
\end{equation}
Due to (\ref{3.8;5}) $\varphi(\epsilon) \in \Ran \bigl( (H_0 + \epsilon)^{-1/2}\bigr) $, hence, $\psi(\epsilon)\in D(H_0)$. Besides, $\psi(\epsilon)$
satisfies the Schr\"odinger
equation
\begin{equation}
 \bigl( H_0 + \epsilon\bigr)  \psi (\epsilon)  + \frac 1{\mu(\epsilon)}V  \psi
(\epsilon) = 0 \quad \quad (\epsilon \in (0, \beps]). \label{11.01.12/1.3}
\end{equation}
Monotonicity of $\mu(\epsilon)$ follows from the fact that $-\epsilon = \inf \sigma \left( H_0 + \mu^{-1}(\epsilon) V\right) $ is monotone 
decreasing with $\mu^{-1}$ \cite{reed}. 
 Inequality (\ref{mubound}) is a direct
consequence of Lemma~3.1 in \cite{simonfuncan}.
\end{proof}
\begin{remark}
 With additional effort instead of (\ref{mubound}) one can possibly prove
that $\mu(\epsilon) = 1 -(\psi (0) , V  \psi (0))^{-1} \epsilon +
\hbox{o}(\epsilon)$, that is $\mu(\epsilon)$ has a derivative at $\epsilon =
0$. Here $\psi(0) = \lim_{\epsilon \to 0} \psi(\epsilon)$, where the limit is in norm (its existence is proved in Corollary~\ref{18.9;4} below). 
\end{remark}

Let us introduce the projection operator
\begin{equation}
P(\epsilon):= \bigl( \varphi(\epsilon),\cdot\bigr)\varphi(\epsilon) , \label{7.8;31}
\end{equation}
where $\varphi(\epsilon)$ is defined in (\ref{3.8;5}). So far we have defined $\mu(\epsilon)$ by equation (\ref{1.08;41}). Now we redefine $\mu(\epsilon)$ setting 
\begin{equation}\label{7.01;41}
 \mu(\epsilon):= \begin{cases}
\sup \sigma\bigl( K (\epsilon)\bigr) & \text{if $\epsilon \in [0, \beps]$},\\
\mu(\beps) & \text{if $\epsilon \in (\beps , \infty)$}.
\end{cases} 
\end{equation}
By Lemma~\ref{3.8;3}
\begin{gather}
 K(\epsilon) = \mu(\epsilon)P(\epsilon) + K(\epsilon)\bigl(1-P(\epsilon)\bigr) ,\\
\|K(\epsilon)(1-P(\epsilon)) \| \leq \eta , \label{19.9;11}
\end{gather}
where $\eta \in (0, 1)$ is some constant.
\begin{lemma}\label{7.8;53}
 For $\epsilon >0 $ the following formula holds
\begin{equation}\label{wurzel}
 \Bigl[1 - \mu(\epsilon) P(\epsilon)\Bigr]^{-1/2} = 1 + \left(\frac
1{\sqrt{1-\mu(\epsilon)}} - 1\right)P(\epsilon).
\end{equation}
\end{lemma}
\begin{proof}
 Using von Neumann series we get
\begin{equation}\label{wurzel2}
 \Bigl[1 - \mu(\epsilon) P(\epsilon)\Bigr]^{-1} = 1 +
\frac{\mu(\epsilon)}{1-\mu(\epsilon)}P(\epsilon) .
\end{equation}
The operator on the rhs of (\ref{wurzel}) is positive and by the direct check
one finds that its square is equal to the operator on the rhs of
(\ref{wurzel2}).
\end{proof}

In the rest of this section we shall derive various estimates on $\varphi(\epsilon), \psi(\epsilon)$, the key result in this respect being
Theorem~\ref{3.8;8}.
\begin{figure}
\begin{center}
\includegraphics[width=0.7\textwidth]{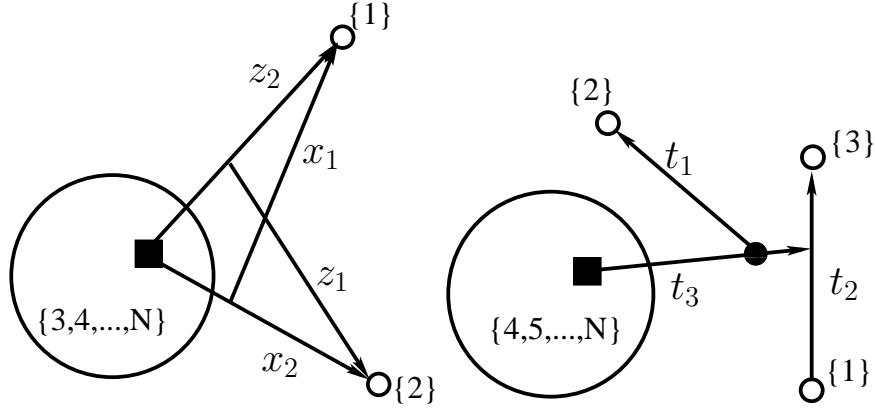}
\caption{Illustration to the choice of orthogonal Jacobi coordinates. The filled square symbolizes the center of mass
of the particles $\{3,4, \ldots, N\}$. Coordinates' scales
are set so that $H_0 = -\sum_i \Delta_{x_i} = -\sum_i \Delta_{z_i}= -\sum_i \Delta_{t_i}$ holds.}
\label{Fig:1}
\end{center}
\end{figure}

We shall use $N$ sets of Jacobi coordinates each associated with the particle number $1, \ldots, N$.
Let us construct the Jacobi coordinates $x_1, \ldots, x_{N-1} \in
\mathbb{R}^3$ associated with the first particle. Here $x_1$ points from the center of mass of particles $\{2, 3,
\ldots, N\}$ in the
direction of particle $1$, $x_i$ points from the center of mass of particles $\{i+1, \ldots, N\}$ in the
direction of particle $i$. The coordinates' scales are chosen so that $H_0 = -
\sum_i \Delta_{x_i}$ holds. The coordinates $z_1, z_2, \ldots, z_{N-1} \in
\mathbb{R}^3$ are associated with particle $2$. Here
$z_1$ points from the center of mass of particles $\{1, 3,
\ldots, N\}$ in the
direction of particle $2$, $z_2$ points from the center of mass of particles $\{ 3,
\ldots, N\}$ in the direction of particle $1$, $z_i$ for $i \geq 3$ points from the center of mass of particles $\{i+1, \ldots, N\}$ in the
direction of particle $i$. The coordinates' scales are chosen so that $H_0 = -
\sum_i \Delta_{z_i}$ holds. This choice of coordinates is illustrated in Fig.~\ref{Fig:1}.
The two sets of coordinates are connected through
\begin{gather}
 z_1 = a_{11} x_1 + a_{12} x_2 , \label{24.9;28}\\
 z_2 = a_{21} x_1 + a_{22} x_2 , \label{24.9;29}\\
z_i = x_i \quad \quad (i \geq 3) ,
\end{gather}
where the $2\times 2$ real matrix $a_{ik}$ is orthogonal. In fact,
\begin{gather}
a_{11} =  -\left[ \frac{m_1 m_2}{(M - m_1)(M-m_2)} \right]^{\frac 12} \label{21.9;1}\\
a_{12} = \left[ \frac{M(M- m_1 -m_2 )}{(M - m_1)(M-m_2)} \right]^{\frac 12}, \label{21.9;2}
\end{gather}
where $M := \sum_{i=1}^N m_i$ and $a_{22} = -a_{11}$ and $a_{12} = a_{21}$.

For each set of coordinates $j = 1, \ldots , N$ we introduce the full and the partial Fourier transforms denoted as $F_j$ and $\mathcal{F}_j$ respectively.
In particular,
\begin{gather}
 \hat f (p_1 , \ldots, p_{N-1}) = (F_1 f) = \frac 1{(2 \pi)^{\frac {3N-3}2}} \int e^{-i \sum_{k=1}^{N-1} p_k \cdot x_k }f(x_1, \ldots, x_{N-1})d^3 x_1 \ldots d^3 x_{N-1} , \label{6.8;51}\\
\hat f (p_1 , x_2 , \ldots, x_{N-1}) = (\mathcal{F}_1 f) = \frac 1{(2 \pi)^{3/2}} \int e^{-ip_1 \cdot x_1 }f(x_1, \ldots, x_{N-1})d^3 x_1 , \label{7.8;51}\\
\hat f (q_1,  \ldots, q_{N-1}) = (F_2 f) = \frac 1{(2 \pi)^{\frac {3N-3}2}} \int e^{-i \sum_{k=1}^{N-1} q_k \cdot z_k }f(z_1, \ldots, z_{N-1})d^3 z_1 \ldots d^3 z_{N-1} , \\
\hat f (q_1 , z_2 , \ldots, z_{N-1}) = (\mathcal{F}_2 f) = \frac 1{(2 \pi)^{3/2}} \int e^{-iq_1 \cdot z_1 }f(z_1, \ldots, z_{N-1})d^3 z_1 , \label{20.9;27}
\end{gather}
For shorter notation let us define the following tuples $x_r := (x_2, x_3, \ldots, x_{N-1})\in \mathbb{R}^{3N-6}$, $x_c := (x_3, \ldots, x_{N-1}) \in \mathbb{R}^{3N-9}$ and
$p_r := (p_2, p_3, \ldots, p_{N-1})\in \mathbb{R}^{3N-6}$, $p_c := (p_3, \ldots, p_{N-1})\in \mathbb{R}^{3N-9}$. Similarly, we define $z_r, q_r \in \mathbb{R}^{3N-6}$ and
$z_c, q_c \in \mathbb{R}^{3N-9}$.

For $j=1,2,\ldots, N$, $s \in \mathbb{R}^3$ and $\epsilon > 0$ let us define the positive operator
$G_j (s, \epsilon)$, where $G_1 (s, \epsilon)$ acts on $f \in L^2(\mathbb{R}^{3N-3})$ as follows
\begin{equation}\label{6.8;1}
G_1 (s, \epsilon) f = \mathcal{F}_1^{-1} \left[(p_1 + s)^2 + \epsilon
\right]^{-1/2} (\mathcal{F}_1 f).
\end{equation}
The operators $G_j (s, \epsilon)$ for $j \geq 2$ are constructed analogously using appropriate coordinates.
\begin{theorem}\label{3.8;8}
 Suppose the conditions of Theorem~\ref{3.8;3} are fulfilled. Then for all $\alpha \in [1,
\frac 32 )$ and $\varphi(\epsilon)$ defined in (\ref{3.8;5}) the following bound holds
\begin{equation}\label{6.8;2}
 \sup_{\epsilon >0} \sup_{\epsilon' >0} \sup_{s \in \mathbb{R}^3}  \|  G_j^\alpha (s, \epsilon')
\varphi(\epsilon) \| < \infty \quad (j = 1, 2, \ldots , N).
\end{equation}
\end{theorem}
Before we proceed with the proof we shall need a couple of technical lemmas
\begin{lemma}\label{7.8;8}
 Suppose $H$ defined in (\ref{1.08;3})--(\ref{1.08;4}) satisfies R2. For a multi--index $\pmb i$ and $\epsilon >0$ let us define the operators
\begin{gather}
 \mathcal{K}_{\pmb i} (\epsilon) := \mu^{-1} (\epsilon) \bigl( H_0 + V^+_{\pmb i} + \epsilon \bigr)^{-1/2} V^-_{\pmb i} \bigl( H_0 + V^+_{\pmb i} + \epsilon \bigr)^{-1/2} \\
Q_{\pmb i} (\epsilon) := [1-\mathcal{K}_{\pmb i} (\epsilon)]^{-1} ,
\end{gather}
where $V^\pm_{\pmb i}$, $\mu (\epsilon)$ were defined in (\ref{7.8;1}), (\ref{7.01;41}) respectively. 
Then $\mathcal{K}_{\pmb i} (\epsilon) , Q_{\pmb i} (\epsilon) \in \mathfrak{B}(L^2 (\mathbb{R}^{3N-3}))$ and $\|\mathcal{K}_{\pmb i} (\epsilon)\| \leq (1+\omega)^{-1} (1+\omega/2)$,
$\|Q_{\pmb i} (\epsilon)\| \leq 2 \omega^{-1} (1+\omega)$.
\end{lemma}
\begin{proof}
From R2 and HVZ theorem it follows that
\begin{equation}\label{7.8;3}
 H_0 + \mu^{-1} (\epsilon) V^+_{\pmb i} - (1+ \omega)V^-_{\pmb i} \geq 0 , 
\end{equation}
since $\mu^{-1}(\epsilon) > 1$. 
Therefore, by the BS principle (Theorem~\ref{31.07;16}) we get
$\sigma \bigl( (1+ \omega)\mu(\epsilon) \mathcal{K}_{\pmb i} (\epsilon) \bigr) \cap (1, \infty) = \emptyset $. Together 
with $\mathcal{K}_{\pmb i} (\epsilon) \geq 0 $ this gives
$\|\mathcal{K}_{\pmb i} (\epsilon)\| \leq (1+\omega)^{-1} (1+\omega/2) <1 $, see Theorem~\ref{3.8;3}. The rest of the proof is trivial.
\end{proof}
\begin{lemma}\label{6.8;22}
Suppose $H$ defined in (\ref{1.08;3})--(\ref{1.08;4}) satisfies R2. For $f \in D(H^{1/2}_0)$, $\epsilon >0$ and any ordered multi--index $\pmb i$ 
the following inequality holds
\begin{equation}
 \left\| \bigl( H_0 + \mu^{-1}(\epsilon) V^+_{\pmb i} + \epsilon  \bigr)^{-1/2} \bigl( H_0 + \epsilon
\bigr)^{1/2} f\right\| \leq \| f\| ,
\end{equation}
where $V^+_{\pmb i}$ is defined in (\ref{7.8;1}).
\end{lemma}
\begin{proof}
Indeed,
\begin{gather}
\left\| \bigl( H_0 + \epsilon + \mu^{-1}(\epsilon) V^+_{\pmb i} \bigr)^{-1/2} \bigl( H_0 + \epsilon
\bigr)^{1/2} f\right\|^2 \nonumber\\
= \Bigl(\bigl( H_0 + \epsilon \bigr)^{1/2} f ,  \bigl(
H_0 + \epsilon + \mu^{-1}(\epsilon) V^+_{\pmb i} \bigr)^{-1} \bigl( H_0 + \epsilon \bigr)^{1/2} f \Bigr)
\nonumber\\
\leq \Bigl(\bigl( H_0 + \epsilon \bigr)^{1/2} f ,  \bigl(H_0 + \epsilon
\bigr)^{-1} \bigl( H_0 + \epsilon \bigr)^{1/2} f \Bigr) = \|f\|^2 ,
\label{17.01/02}
\end{gather}
In (\ref{17.01/02}) we have used the operator inequality
\begin{equation}
  \bigl( H_0 + \epsilon + \mu^{-1}(\epsilon) V^+_{\pmb i} \bigr)^{-1} \leq \bigl(H_0 + \epsilon  \bigr)^{-1}, 
\end{equation}
which follows from $H_0 + \mu^{-1}(\epsilon) V^+_{\pmb i} \geq H_0 \geq 0 $ (see, for example, Proposition
A.2.5 on page 131 in \cite{glimmjaffe}).
\end{proof}

\begin{proof}[Proof of Theorem~\ref{3.8;8}]
Without loosing generality we can set $j = 1$; the power $\alpha \in [1, \frac 32 )$ remains fixed throughout the proof. We shall further on
assume that $\epsilon' \leq \epsilon$ because from (\ref{6.8;1}) clearly follows that
\begin{equation}
  \sup_{\epsilon >0} \sup_{\epsilon' >0} \sup_{s \in \mathbb{R}^3}  \|  G_j^\alpha (s, \epsilon')
\varphi(\epsilon) \| =  \sup_{\epsilon >0} \sup_{0 <\epsilon' \leq \epsilon} \sup_{s \in \mathbb{R}^3}  \|  G_j^\alpha (s, \epsilon')
\varphi(\epsilon) \|
\end{equation}
Using (\ref{11.01.12/1}), (\ref{11.01.12/1.3}) and $\mu(\epsilon) \geq (1+\omega/2)^{-1}$ we can write
\begin{gather}
 \|  G^\alpha_1 (s, \epsilon') \varphi(\epsilon) \| \leq (1+\omega/2) \sum_{i = 2}^N  \|  G^\alpha_1 (s, \epsilon')  \bigl(
H_0 + \epsilon\bigr)^{-1/2} v_{1i} \psi(\epsilon) \|  \nonumber\\
+ (1+\omega/2) \sum_{2\leq i < k \leq N} \|  G^\alpha_1 (s, \epsilon') \bigl( H_0 + \epsilon\bigr)^{-1/2}
v_{ik}  \psi(\epsilon) \| .
\end{gather}
Thus without
loss of generality, (\ref{6.8;2}) follows from the following inequalities
\begin{gather}
 \sup_{\epsilon >0} \sup_{0 <\epsilon' \leq \epsilon}  \sup_{s \in \mathbb{R}^3}  \|  G^\alpha_1 (s, \epsilon')  \bigl( H_0 +
\epsilon\bigr)^{-1/2} v_{12} \psi(\epsilon) \| < \infty \label{12.01/01}\\
\sup_{\epsilon >0} \sup_{0 <\epsilon' \leq \epsilon}  \sup_{s \in \mathbb{R}^3} \|  G^\alpha_1 (s, \epsilon')  \bigl( H_0 +
\epsilon\bigr)^{-1/2} v_{23} \psi(\epsilon) \| < \infty \label{12.01/02}
\end{gather}
We shall follow the method developed in \cite{1,jpasubmit,lmp}. Let us start with
(\ref{12.01/01}).
We introduce another set of Jacobi coordinates $y_1 = \sqrt{2\mu_{12}} (r_2 -
r_1)$, $y_2 = (\sqrt{2 M_{12;3}})\bigl[ r_3 -
m_1/(m_1+m_2) r_1 - m_2/(m_1+m_2) r_2\bigr]$ etc.,
where $M_{ik;j} := (m_i + m_k)m_j / (m_i + m_k + m_j)$ and
$\mu_{ik} := m_i m_k /(m_i + m_k)$ denote the reduced masses.
The coordinate $y_i \in \mathbb{R}^3$ is proportional to the
vector pointing from the centre of mass of the particles $[1, 2, \ldots , i]$
to the particle $i+1$, and the scales are set in order to guarantee that $H_0 = - \sum_i \Delta_{y_i} $. The full and partial Fourier transforms have the form
\begin{gather}
 (F_{12} f ) (p_{y_1}, p_{y_2}, \ldots, p_{y_{N-1}}):= \frac 1{(2\pi)^{\frac{3N-3}2}} \int e^{-i\sum_{k = 1}^{N-1} p_{y_k} \cdot y_k} f(y_1, \ldots, y_{N-1}) d^3 y_1 \ldots d^3 y_{N-1} \label{6.8;52}\\
(\mathcal{F}_{12} f ) (y_1, p_{y_2}, \ldots, p_{y_{N-1}}):= \frac 1{(2\pi)^{\frac{3N-6}2}} \int e^{-i\sum_{k = 2}^{N-1} p_{y_k} \cdot y_k} f(y_1, \ldots, y_{N-1}) d^3 y_2 \ldots d^3 y_{N-1}
\end{gather}
For shorter notation we shall
denote by
$y_r , p_{y_r}\in \mathbb{R}^{3N-6}$ the following tuples $y_r := (y_2, y_3,
\ldots, y_{N-1})$ and $p_{y_r} := (p_{y_2}, p_{y_3}, \ldots, p_{y_{N-1}})$. The coordinate set $y_i$ can be expressed through $x_i$ as follows
\begin{equation}\label{6.8;53}
 x_i = \sum_{k=1}^{N-1} b_{ik} y_k \quad (i = 1, \ldots, N-1),
\end{equation}
where the $(N-1)\times (N-1)$ real orthogonal matrix $b_{ik}$ depends on the mass ratios. The expressions for these coefficients are complicated, we just mention that
\begin{equation}
 b_{11} = - \left[ \frac{m_2 M}{(M - m_1)(m_1 + m_2)} \right]^{\frac 12} . 
\end{equation}

Similar to \cite{jpasubmit,lmp} we introduce the operator, which acts on $f \in
L^2(\mathbb{R}^3)$ according to the rule
\begin{equation}\label{7.8;15}
 B_{12} (\epsilon) f = (1+ \epsilon^{\zeta/2})f + \mathcal{F}_{12}^{-1} (|p_{y_r}|^\zeta -
1)\chi_1 (|p_{y_r}|) (\mathcal{F}_{12} f),
\end{equation}
where
\begin{equation}\label{defzeta28}
 \zeta = \frac \alpha2  + \frac 14 < 1
\end{equation}
For all $\epsilon >0$ the operators $B_{12} (\epsilon)$ and $B_{12}^{-1}
(\epsilon)$ are bounded. Inserting the identity $B_{12} (\epsilon') B_{12}^{-1} (\epsilon') = 1$  into
(\ref{12.01/01})
and using $[B_{12} (\epsilon') , v_{12}] =0$ we get
\begin{equation}
 \|  G^\alpha_1 (s, \epsilon') \bigl( H_0 + \epsilon\bigr)^{-1/2} v_{12} \psi(\epsilon) \| \leq
\| \mathcal{C} (s, \alpha, \epsilon, \epsilon') \| \, \| B_{12}^{-1}
(\epsilon')\bigl|v_{12}\bigr|^{1/2} \psi (\epsilon)\| ,
\end{equation}
where by definition
\begin{equation}\label{12.01/04}
 \mathcal{C} (s, \alpha, \epsilon, \epsilon')  :=  G^\alpha_1 (s, \epsilon') \bigl( H_0 +
\epsilon\bigr)^{-1/2} B_{12} (\epsilon') \bigl|v_{12}\bigr|^{1/2}
\end{equation}
By Lemma~\ref{6.8;5} below to prove (\ref{12.01/01}) it suffices to show that
\begin{equation}\label{6.8;6}
  \sup_{\epsilon >0 } \sup_{0 < \epsilon' \leq \epsilon }  \| B_{12}^{-1} (\epsilon') \bigl|v_{12}\bigr|^{1/2} \psi
(\epsilon)\| < \infty .
\end{equation}
Using the method in \cite{jpasubmit} we shall prove that
\begin{equation}\label{6.8;7}
  \sup_{\epsilon >0 } \sup_{0 < \epsilon' \leq \epsilon }  \| B_{12}^{-1} (\epsilon') \bigl((v_{12})_-\bigr)^{\frac 12} \psi
(\epsilon)\| < \infty .
\end{equation}
We first prove that (\ref{6.8;6}) follows from (\ref{6.8;7}) and afterwards prove that (\ref{6.8;7}) holds.
After rearranging the terms in the Schr\"odinger equation (\ref{11.01.12/1.3}) we obtain
\begin{gather}
 B_{12}^{-1} (\epsilon') \bigl|v_{12}\bigr|^{1/2} \psi(\epsilon) = \mu^{-1} (\epsilon)\bigl|v_{12}\bigr|^{1/2}
B_{12}^{-1} (\epsilon') \bigl[H_0 + \epsilon + \mu^{-1} (\epsilon) (v_{12})_+ \bigr]^{-1} \\
\times  
\bigl((v_{12})_- - \sum_{i=3}^N v_{1i} - \sum_{2 \leq i < j \leq N} v_{ij}\bigr)
\psi(\epsilon)
\end{gather}
This leads to the upper bound
\begin{gather}
 \| B_{12}^{-1} (\epsilon') \bigl|v_{12}\bigr|^{1/2} \psi (\epsilon)\| \nonumber\\
\leq (1+ \omega/2)\Bigl\|
\bigl|v_{12}\bigr|^{1/2} \bigl[H_0 + \epsilon + \mu^{-1} (\epsilon) (v_{12})_+ \bigr]^{-1}
\bigl((v_{12})_-\bigr)^{1/2}\Bigr\|
\Bigl\| B_{12}^{-1} (\epsilon') \bigl((v_{12})_-\bigr)^{1/2} \psi(\epsilon) \Bigr\| \nonumber\\
+ (1+ \omega/2) \sum_{i=3}^N \Bigl\| \bigl|v_{12}\bigr|^{1/2} B_{12}^{-1} (\epsilon')\bigl[H_0 + \epsilon
+ \mu^{-1} (\epsilon) (v_{12})_+ \bigr]^{-1} \bigl| v_{1i} \bigr|^{1/2}\Bigr\| \Bigl\| \bigl| v_{1i}
\bigr|^{1/2}\psi(\epsilon) \Bigr\| \nonumber\\
+ (1+ \omega/2) \sum_{2 \leq i < j \leq N} \Bigl\| \bigl|v_{12}\bigr|^{1/2} B_{12}^{-1} (\epsilon')
\bigl[H_0 + \epsilon + \mu^{-1} (\epsilon) (v_{12})_+ \bigr]^{-1} \bigl| v_{ij} \bigr|^{1/2}\Bigr\|
\Bigl\| \bigl| v_{ij} \bigr|^{1/2}\psi(\epsilon) \Bigr\|,   \label{6.8;11}
\end{gather}
where we have used $\mu (\epsilon) \leq (1+\omega/2)$. Note, that for $1 \leq i < j \leq N$ the terms $\bigl\| \bigl| v_{ij}
\bigr|^{1/2}\psi(\epsilon) \bigr\|$ are uniformly bounded. Indeed, by
(\ref{11.01.12/1})
\begin{gather}
\sup_{\epsilon>0}\Bigl\| \bigl| |v_{ij} \bigr|^{1/2}\psi(\epsilon) \Bigr\| =
\sup_{\epsilon>0} \Bigl\| \bigl| v_{ij} \bigr|^{1/2} \bigl( H_0 +
\epsilon\bigr)^{-1/2} \varphi(\epsilon) \Bigr\| \nonumber \\
\leq
C_v \equiv \max_{i <j} \sup_{\epsilon>0} \Bigl\| \bigl| v_{ij} \bigr|^{1/2} \bigl( H_0 +
\epsilon\bigr)^{-1} \bigl| v_{ij} \bigr|^{1/2}  \Bigr\|^{1/2} < \infty , \label{16.01/1}
\end{gather}
where we have used $\| \varphi(\epsilon) \| \leq  1$ and $|v_{ij} (r)|^{1/2} \in L^3 (\mathbb{R}^3)$, c. f. (\ref{9.10;1}). 
Applying Lemma~\ref{6.8;14} and (\ref{16.01/1}) we conclude
that
all terms under the sums in (\ref{6.8;11}) are uniformly bounded. For the first operator norm on the rhs of
(\ref{6.8;11}) using (\ref{16.01/1}) we get
\begin{gather}
\sup_{\epsilon>0} \Bigl\| \bigl|v_{12}\bigr|^{1/2} \bigl[H_0 + \epsilon +  \mu^{-1} (\epsilon) (v_{12})_+ \bigr]^{-1}
\bigl((v_{12})_-\bigr)^{1/2}\Bigr\| \nonumber \\
\leq
\sup_{\epsilon>0}  \Bigl\| \bigl|v_{12}\bigr|^{1/2} \bigl[H_0 + \epsilon \bigr]^{-1}
\bigl|v_{12}\bigr|^{1/2}\Bigr\| \leq C_v^2 . \label{6.8;16}
\end{gather}
Thus from (\ref{6.8;11})--(\ref{6.8;16}) and (\ref{6.8;7}) inequality (\ref{6.8;6}) follows. It remains to prove (\ref{6.8;7}).
Using Eq.~16 in \cite{jpasubmit} (where one has to set $k_n^2 =
\epsilon$ and $\lambda_n = 1$) we get
\begin{gather}
B_{12}^{-1} (\epsilon') \bigl((v_{12})_-\bigr)^{\frac 12} \psi
(\epsilon) \nonumber\\
= -  \mu^{-1} (\epsilon) Q_{\{3,\ldots, N\}}(\epsilon) \sum_{i=3}^N \bigl((v_{12})_-\bigr)^{\frac 12} B_{12}^{-1} (\epsilon') \bigl[H_0 + \epsilon +  \mu^{-1} (\epsilon) (v_{12})_+ \bigr]^{-1} v_{1i} \psi(\epsilon) \nonumber\\
-  \mu^{-1} (\epsilon) Q_{\{3,\ldots, N\}}(\epsilon) \sum_{2 \leq i < j \leq N}  \bigl((v_{12})_-\bigr)^{\frac 12} B_{12}^{-1} (\epsilon') \bigl[H_0 + \epsilon +  \mu^{-1} (\epsilon) (v_{12})_+ \bigr]^{-1} v_{ij}  \psi(\epsilon) ,
\end{gather}
where
\begin{equation}
 Q_{\{3,\ldots, N\}}(\epsilon) = \left\{1 -  \mu^{-1} (\epsilon) \bigl((v_{12})_-\bigr)^{\frac 12} \bigl[H_0 + \epsilon +  \mu^{-1} (\epsilon) (v_{12})_+ \bigr]^{-1} \bigl((v_{12})_-\bigr)^{\frac 12} \right\}^{-1}
\end{equation}
was defined in Lemma~\ref{7.8;8}. By Lemma~\ref{7.8;8} and (\ref{16.01/1})
\begin{gather}
 \left\|B_{12}^{-1} (\epsilon') \bigl((v_{12})_-\bigr)^{\frac 12} \psi(\epsilon) \right\| \nonumber \\
\leq  2 \omega^{-1} (1+\omega) (1+\omega/2) C_v \sum_{i=3}^N \left\| \bigl((v_{12})_-\bigr)^{\frac 12} B_{12}^{-1} (\epsilon') \bigl[H_0 + \epsilon + \mu^{-1} (\epsilon) (v_{12})_+ \bigr]^{-1} |v_{1i}|^{\frac12} \right\| \nonumber \\
+ 2 \omega^{-1} (1+\omega) (1+\omega/2)  C_v \sum_{2 \leq i < j \leq N}   
\left\| \bigl((v_{12})_-\bigr)^{\frac 12} B_{12}^{-1} (\epsilon') \bigl[H_0 + \epsilon + \mu^{-1} (\epsilon) (v_{12})_+ \bigr]^{-1} |v_{ij}|^{\frac 12}  \right\| \nonumber
\end{gather}
Now (\ref{6.8;7}) follows from Lemma~\ref{6.8;14} and (\ref{12.01/01}) is proved.

Let us now consider (\ref{12.01/02}).
After rearranging the terms in the Schr\"odinger equation (\ref{11.01.12/1.3}) we obtain
\begin{gather}
 G^\alpha_1 (s, \epsilon)  \bigl( H_0 + \epsilon\bigr)^{-1/2} v_{23}  \psi(\epsilon) \nonumber \\
= \mu^{-1} (\epsilon) \bigl(
H_0 + \epsilon\bigr)^{-1/2} v_{23} \bigl( H_0 + \epsilon + \mu^{-1} (\epsilon) V_{\{1\}}^+ \bigr)^{-1}
 G^\alpha_1 (s, \epsilon) \left\{ V_{\{1\}}^- - \sum_{i=2}^N v_{1i}\right\} \psi(\epsilon) \nonumber
\\
= \mu^{-1} (\epsilon)\mathcal{J}_1(\epsilon) \psi_1 (\epsilon, \epsilon') - \mu^{-1} (\epsilon)\mathcal{J}_2(\epsilon) \psi_2
(\epsilon, \epsilon')  \label{17.01/0}
\end{gather}
where by definition
\begin{gather}
 \mathcal{J}_1(\epsilon)  := \bigl( H_0 + \epsilon\bigr)^{-1/2} v_{23} \bigl(
H_0 + \epsilon + \mu^{-1} (\epsilon)V_{\{1\}}^+ \bigr)^{-1} \bigl(V_{\{1\}}^-\bigr)^{1/2} \\
\mathcal{J}_2(\epsilon)  := \bigl( H_0 + \epsilon\bigr)^{-1/2} v_{23} \bigl( H_0
+ \epsilon + \mu^{-1} (\epsilon)V_{\{1\}}^+ \bigr)^{-1/2}
\end{gather}
and
\begin{gather}
\psi_1(\epsilon, \epsilon') := G^\alpha_1 (s, \epsilon') \bigl(V_{\{1\}}^-\bigr)^{1/2} \psi(\epsilon) \label{6.8;27}\\
\psi_2(\epsilon, \epsilon') := \sum_{i=2}^N \bigl( H_0 + \epsilon + \mu^{-1} (\epsilon) V_{\{1\}}^+ \bigr)^{-1/2}
G^\alpha_1 (s, \epsilon') v_{1i} \psi(\epsilon) \label{6.8;28}
\end{gather}
In (\ref{17.01/0}) we have used that $[ G^\alpha_1 (s, \epsilon), V_{\{1\}}^\pm] = 0$ due to
$V_{\{1\}}^\pm$ being dependent only on $x_2, \ldots, x_{N-1}$.
It is easy to show that $\|\mathcal{J}_{1,2} (\epsilon)\|$ is uniformly bounded.
For example,
\begin{gather}
 \|\mathcal{J}_1 (\epsilon)\| \leq \left\||v_{23}|^{1/2} (H_0 + \epsilon)^{-1}
|v_{23}|^{1/2}\right\|^{1/2} \,
\left\||v_{23}|^{1/2} (H_0 + \mu^{-1} (\epsilon) V_{\{1\}}^+ + \epsilon )^{-1}
|v_{23}|^{1/2}\right\|^{1/2} \nonumber\\
\times \left\|(V_{\{1\}}^-)^{1/2} (H_0 + \mu^{-1} (\epsilon) V_{\{1\}}^+ + \epsilon )^{-1} (V_{\{1\}}^-)^{1/2}
\right\|^{1/2} \leq C_v (1+\omega)^{-1/2} ,
\end{gather}
where we have used Lemma~\ref{7.8;8}.
It remains to prove that
\begin{equation}
 \sup_{\epsilon > 0} \sup_{0 < \epsilon' \leq \epsilon} \|\psi_{1,2} (\epsilon, \epsilon')\| < \infty. \label{6.8;35}
\end{equation}
By Lemma~\ref{6.8;22} and (\ref{12.01/01})
\begin{equation}
 \sup_{\epsilon >0 } \sup_{0 < \epsilon' \leq \epsilon}  \| \psi_2(\epsilon, \epsilon')\| \leq \sum_{i=2}^N \sup_{\epsilon >0 }
\sup_{0 < \epsilon' \leq \epsilon} \bigl\|  \bigl( H_0 + \epsilon
\bigr)^{-1/2} G^\alpha_1 (s, \epsilon') v_{1i} \psi(\epsilon)\bigr\| < \infty .  \label{6.8;32}
\end{equation}
After
rearranging the terms in (\ref{11.01.12/1.3}) we obtain
\begin{equation}
 \psi(\epsilon) = \mu^{-1} (\epsilon)\bigl[H_0 + \mu^{-1} (\epsilon)V_{\{1\}}^+ + \epsilon\bigr]^{-1}
\bigl(V_{\{1\}}^-\bigr)^{1/2}  \bigl(V_{\{1\}}^-\bigr)^{1/2} \psi(\epsilon) - \bigl[H_0 +
\mu^{-1} (\epsilon)V_{\{1\}}^+ + \epsilon\bigr]^{-1} \sum_{i=2}^N v_{1i} \psi(\epsilon) . \nonumber
\end{equation}
Therefore,
\begin{gather}
 \bigl(V_{\{1\}}^-\bigr)^{1/2} \psi(\epsilon) = -\mu^{-1} (\epsilon) Q_{\{1\}} (\epsilon)
\bigl(V_{\{1\}}^-\bigr)^{1/2} \bigl[H_0 + \mu^{-1} (\epsilon) V_{\{1\}}^+ + \epsilon\bigr]^{-1/2} \nonumber \\
\times \sum_{i=2}^N
\bigl[H_0 + \mu^{-1} (\epsilon) V_{\{1\}}^+ + \epsilon\bigr]^{-1/2} v_{1i} \psi(\epsilon) , \label{6.8;25}
\end{gather}
where $Q_{\{1\}} (\epsilon)$ was defined in Lemma~\ref{7.8;8}. Using (\ref{6.8;25}) and Lemma~\ref{7.8;8} we obtain from (\ref{6.8;27}), (\ref{6.8;28})
\begin{gather}
 \| \psi_1(\epsilon, \epsilon')\| \leq 2 (1+\omega) \omega^{-1} \mu^{-1} (\epsilon) \left\|
\bigl(V_{\{1\}}^-\bigr)^{1/2} \bigl[H_0 + \mu^{-1} (\epsilon) V_{\{1\}}^+ + \epsilon\bigr]^{-1/2}\right\|
\|\psi_2 (\epsilon, \epsilon')\| \nonumber\\
\leq 2 (1+\omega) \omega^{-1} \mu^{-1} (\epsilon) \left\|\bigl(V_{\{1\}}^-\bigr)^{1/2} \bigl[H_0 + \mu^{-1} (\epsilon) V_{\{1\}}^+ +
\epsilon\bigr]^{-1} \bigl(V_{\{1\}}^-\bigr)^{1/2} \right\|^{1/2} \|\psi_2 (\epsilon, \epsilon')\| \nonumber\\
\leq 2[(1+\omega)(1+\omega/2)]^{1/2} \omega^{-1} \|\psi_2 (\epsilon, \epsilon')\|.  \label{6.8;31}
\end{gather}
 Now (\ref{6.8;35}) follows from (\ref{6.8;32}) and (\ref{6.8;31}).
\end{proof}
\begin{lemma}\label{6.8;5}
For all $\alpha \in [1, \frac 32 )$  the following inequality holds
\begin{equation}
 \sup_{\epsilon >0 } \sup_{0 < \epsilon' \leq \epsilon } \sup_{s \in \mathbb{R}^3} \| \mathcal{C} (s, \alpha,
\epsilon, \epsilon') \| < \infty ,
\end{equation}
where $ \mathcal{C} (s, \alpha,\epsilon, \epsilon')$ is defined in (\ref{12.01/04}).
\end{lemma}
\begin{proof}
Without loosing generality we can consider $0 < \epsilon < 1$. For the dual
coordinates defined in (\ref{6.8;51}), (\ref{6.8;52}) the following relation holds $p_1 = \sum_{i=1}^{N-1} b_{1i} p_{y_i} $,
where we have used that the matrix $b_{ik}$ in (\ref{6.8;53}) is orthogonal.
 The Fourier--transformed operator $M= F_{12} \mathcal{C} (s, \alpha, \epsilon, \epsilon')
F_{12}^{-1}$ acts on $f (p_{y_1}, p_{y_r}) \in L^2 (\mathbb{R}^{3N-3})$ as
follows
\begin{equation}
 (Mf)(p_{y_1}, p_{y_r}) = \int \mathcal{M}(p_{y_1} , p'_{y_1} ; p_{y_r}) f(p'_{y_1} ,
p_{y_r}) d^3 p'_{y_1} ,
\end{equation}
where
\begin{gather}
 \mathcal{M} (p_{y_1} , p'_{y_1} ; p_{y_r}) =
\frac{(2\mu_{12})^{3/2}}{(2\pi)^{3/2}}\left[(b_{11} p_{y_1} + \sum_{k=2}^{N-1} b_{1k}
p_{y_k} +s)^2 +  \epsilon' \right]^{-\alpha/2}
\left(p_{y_1}^2 + p_{y_r}^2 + \epsilon \right)^{-1/2} \nonumber \\
\times \left\{1 + (\epsilon')^{\zeta/2} + (|p_{y_r}|^\zeta -1|)\chi_1
(|p_{y_r}|)\right\} \widehat{|v_{12}|^{\frac 12}}\bigl(\sqrt{2\mu_{12}}(p_{y_1}
- p'_{y_1})\bigr) ,
\end{gather}
and the hat denotes standard Fourier transform in $L^2 (\mathbb{R}^3)$. We
estimate the norm as
\begin{equation}
 \|M\|^2 \leq \sup_{|p_{y_r}| \leq 1} \int \bigl|\mathcal{M}(p_{y_1} , p'_{y_1} ;
p_{y_r})\bigr|^2 d^3 p'_{y_1} d^3 p_{y_1} +  \sup_{|p_{y_r}| > 1} \int
\bigl|\mathcal{M}(p_{y_1} , p'_{y_1} ; p_{y_r})\bigr|^2 d^3 p'_{y_1} d^3 p_{y_1} \label{6.8;56}
\end{equation}
For the first term on the rhs in (\ref{6.8;56}) we get
\begin{gather}
 \sup_{|p_{y_r}| \leq 1} \int \bigl|\mathcal{M}(p_{y_1} , p'_{y_1} ; p_{y_r})\bigr|^2 d^3
p'_{y_1} d^3 p_{y_1} \nonumber \\
\leq  C_0 \sup_{|p_{y_r}| \leq 1} \sup_{s' \in \mathbb{R}^3}
\int \left[(b_{11} p_{y_1} + s')^2 + \epsilon' \right]^{-\alpha} \frac{
\left(|p_{y_r}|^\zeta + (\epsilon')^{\zeta/2}\right)^2 }{\left(p_{y_1}^2 +
p_{y_r}^2 + \epsilon' \right)} d^3 p_{y_1} , \label{6.8;59}
\end{gather}
where
\begin{equation}
 C_0 := \frac{(2\mu_{12}^3)}{2\pi^3} \int \left|\widehat{|v_{12}|^{\frac
12}}\bigl(\sqrt{2\mu_{12}} p_{y_1}\bigr)\right|^2 d^3 p_{y_1}
\end{equation}
is finite due to $|v_{12}|^{1/2} \in L^2 (\mathbb{R}^3)$. In (\ref{6.8;59}) we have also used that $\epsilon' \leq \epsilon$. Let us use the following inequality
\begin{equation}
 \left(|p_{y_r}|^\zeta + (\epsilon')^{\zeta/2}\right)^2  \leq 2|p_{y_r}|^{2\zeta} +
2(\epsilon')^\zeta  \leq 4\left(|p_{y_r}|^2 + \epsilon'\right)^\zeta ,\label{8.6;61}
\end{equation}
which follows from $a^\gamma + b^\gamma \leq 2(a+b)^\gamma$ for any $a, b \geq
0$ and $0 \leq \gamma\leq 1$. Using (\ref{8.6;61}) and (\ref{defzeta28}) we obtain
from (\ref{6.8;59})
\begin{gather}
  \sup_{|p_{y_r}| \leq 1} \int \bigl|\mathcal{M} (p_{y_1} , p'_{y_1} ; p_{y_r})\bigr|^2 d^3
p'_{y_1} d^3 p_{y_1}  \leq 4 C_0 |b_{11}|^{-2\alpha} \mathcal{J} \sup_{|p_{y_r}|
\leq 1} \left(|p_{y_r}|^2 + \epsilon'\right)^{\frac{3-2\alpha}4} \nonumber\\
\leq 4 C_0 |b_{11}|^{-2\alpha} 2^{\frac{3-2\alpha}4} \mathcal{J}  , \label{6.8;72}
\end{gather}
where
\begin{equation}
 \mathcal{J} := \sup_{s' \in \mathbb{R}^3} \int \frac{d^3
p_{y_1}}{|p_{y_1}|^{2\alpha} \left[(p_{y_1} + s')^2 +1\right]} . \label{8.6;64}
\end{equation}
It remains to show that the expression in (\ref{8.6;64}) is finite. This becomes clear from he following upper bound
\begin{gather}
 \mathcal{J} \leq   \int_{|p_{y_1}| \leq 2} \frac{d^3
p_{y_1}}{|p_{y_1}|^{2\alpha}} + \sup_{s' \in \mathbb{R}^3}  \int_{|p_{y_1}| \geq
2} \frac{d^3p_{y_1}}{((1/2)|p_{y_1}|^{2\alpha} + 1) [(p_{y_1} + s')^2 + 1]} \nonumber \\
\leq \frac{32\pi}{(3-2\alpha)4^\alpha} + \left[\int
\frac{d^3p_{y_1}}{((1/2)|p_{y_1}|^{2\alpha} + 2)^2}\right]^{1/2} \left[\int
\frac{d^3p_{y_1}}{(|p_{y_1}|^2 + 1)^2}\right]^{1/2}
\end{gather}
where we have used the Cauchy--Schwarz inequality and the last
two integrals are obviously convergent.
For the second term in (\ref{6.8;56}) we obtain
\begin{gather}
  \sup_{|p_{y_r}| > 1} \int \bigl|\mathcal{M}(p_{y_1} , p'_{y_1} ; p_{y_r})\bigr|^2 d^3
p'_{y_1} d^3 p_{y_1} \leq C_0 |b_{11}|^{-2\alpha} \mathcal{J} \label{6.8;71}
\end{gather}
Because the expression on the rhs of (\ref{6.8;72} ) and (\ref{6.8;71}) do not depend on
$\epsilon , \epsilon', s$ the lemma is proved.
\end{proof}

\begin{lemma}\label{6.8;14}
The following inequalities hold
\begin{gather}
 \sup_{\epsilon > 0} \sup_{\epsilon' > 0} \Bigl\| \bigl|v_{12}\bigr|^{1/2} B_{12}^{-1} (\epsilon') \bigl[H_0 +
\epsilon + \mu^{-1} (\epsilon) (v_{12})_+ \bigr]^{-1} \bigl| v_{1i} \bigr|^{1/2}\Bigr\| < \infty
\quad (i=3,\ldots, N ) \label{6.8;75}\\
\sup_{\epsilon > 0} \sup_{\epsilon' > 0} \Bigl\| \bigl|v_{12}\bigr|^{1/2} B_{12}^{-1} (\epsilon')\bigl[H_0 +
\epsilon + \mu^{-1} (\epsilon) (v_{12})_+ \bigr]^{-1} \bigl| v_{ij} \bigr|^{1/2}\Bigr\| < \infty
\quad (2 \leq i < j \leq N )
\end{gather}
\end{lemma}
\begin{proof}
 The proof practically repeats that of Lemma~2 in \cite{jpasubmit}, where
in the definition of $B_{12}$ we used $\zeta=1/2$
(note, that coordinates' notations here are different from definitions in
\cite{jpasubmit}). So we shall restrict ourselves to the proof of (\ref{6.8;75}) for $i=3$, which is equivalent to
$\sup_{\epsilon, \epsilon' >0}\|D(\epsilon, \epsilon')\|< \infty$, where
\begin{equation}
 D (\epsilon, \epsilon') = \mathcal{F}_{12}^{-1} \bigl|v_{12}\bigr|^{1/2} B_{12}^{-1} (\epsilon')
\bigl[H_0 + \epsilon + \mu^{-1} (\epsilon) (v_{12})_+ \bigr]^{-1} \bigl| v_{13} \bigr|^{1/2}
\mathcal{F}_{12}.
\end{equation}
We split $D (\epsilon, \epsilon')$ as follows
\begin{gather}
D (\epsilon, \epsilon') =  D^{(1)} (\epsilon, \epsilon') + D^{(2)}(\epsilon, \epsilon') , \\
D^{(1)} (\epsilon, \epsilon')  := \mathcal{F}_{12}^{-1}  \bigl|v_{12}\bigr|^{\frac 12}
\left\{ B_{12}^{-1} (\epsilon') - (1+(\epsilon')^{\frac \zeta2})^{-1} \right\}
\bigl[H_0 + \epsilon + \mu^{-1} (\epsilon) (v_{12})_+ \bigr]^{-1}\bigl| v_{13} \bigr|^{\frac 12}
\mathcal{F}_{12}  , \\
 D^{(2)} (\epsilon, \epsilon')  := (1+(\epsilon')^{\frac \zeta2})^{-1} \mathcal{F}_{12}^{-1}
\bigl|v_{12}\bigr|^{\frac 12}
 \bigl[H_0 + \epsilon + \mu^{-1} (\epsilon) (v_{12})_+ \bigr]^{-1}\bigl| v_{13} \bigr|^{\frac 12}
\mathcal{F}_{12}.
\end{gather}
For convenience we define the tuple $p_{y_c} := (p_{y_3} , p_{y_4}, \ldots, p_{y_{N-1}})
\in \mathbb{R}^{3N-9}$.
The operator $D^{(1)} (\epsilon, \epsilon')$ acts on $\phi(y_1, p_{y_2},  p_{y_c} ) \in
L^2(\mathbb{R}^{3N-3})$ as follows
\begin{equation}
\bigl( D^{(1)} (\epsilon, \epsilon') \phi\bigr) (y_1 ,  p_{y_2}, p_{y_c} ) = \int d^3 y'_1  \: d^3
p'_{y_2} \: \mathcal{D}^{(1)} (y_1, y'_1 ,   p_{y_2},  p'_{y_2} ;  p_{y_c} ) \phi(y'_1 ,
p'_{y_2}, p_{y_c} )  . \label{6.8;77}
\end{equation}
The integral kernel in (\ref{6.8;77}) has the form, see \cite{1,jpasubmit} 
\begin{eqnarray}
\mathcal{D}^{(1)} (y_1, y'_1 ,   p_{y_2},  p'_{y_2} ;  p_{y_c} ) = \frac
1{2^{\frac{3N-2}2}\pi^{\frac{3N-4}2} \gamma^3}
\left\{|p_{y_r}|^{\zeta} + (\epsilon')^{\frac \zeta2} \right\}^{-1} \chi_1
(|p_{y_r}|)\left| v_{12}\bigl( (2\mu_{12})^{-1} y_1 \bigr) \right|^{\frac 12}
\nonumber\\
\times  G\Bigl((p_{y_r}^2 + \epsilon)^{\frac 12};y_1, y'_1\Bigr)  \exp\left(
i\beta \gamma^{-1} y'_1 \cdot (p_{y_2} - p'_{y_2})\right) \:\widehat{\bigl|
v_{23}\bigr|^{\frac 12}} (\gamma^{-1}(p_{y_2} - p'_{y_2})), \quad
\end{eqnarray}
where $\beta := -m_2 \hbar / ((m_1 + m_2)\sqrt{2\mu_{12}})$, $\gamma :=
\hbar/\sqrt{2M_{12;3}}$ and $p_{y_r}^2 = p_{y_2}^2 + p_{y_c}^2 $.
The function $G(k; y_1, y'_1)$ denotes the integral kernel of the operator
$(-\Delta_{y_1} + k^2 + \mu^{-1} (\epsilon)(v_{12})_+)^{-1}$ acting in $L^2 (\mathbb{R}^3)$. By the arguments in \cite{jpasubmit} (around Eqs. (18)--(19))
\begin{equation}
 0 \leq G(k; y_1, y'_1) \leq \frac{e^{-k|y_1 -y'_1|}}{|y_1 -y'_1|}  \label{6.8;80}
\end{equation}
away from $y_1 = y'_1$.
Using the estimate
\begin{equation}
  \| D^{(1)} (\epsilon, \epsilon') \|^2 \leq \sup_{ p_{y_c}} \int \bigl|\mathcal{D}^{(1)} (y_1, y'_1
,   p_{y_2},  p'_{y_2} ;  p_{y_c} )\bigr|^2 d^3 y_1  d^3 y'_1 d^3 p_{y_2} d^3
p'_{y_2}
\end{equation}
and the upper bound (\ref{6.8;80}) we get
\begin{gather}
\| D^{(1)} (\epsilon, \epsilon') \|^2 \leq C_0 \sup_{ p_{y_c}} \int_{|p_{y_2}|\leq 1}
\left\{\bigl( p^2_{y_2} + |p_{y_c}|^2\bigr)^{\frac \zeta2} + (\epsilon')^{\frac
\zeta2} \right\}^{-2}
\left( p^2_{y_2} + |p_{y_c}|^2 + \epsilon \right)^{-\frac 12} d^3 p_{y_2} \nonumber\\
\leq C_0  \int_{|p_{y_2}|\leq 1}  \frac{d^3 p_{y_2}}{|p_{y_2}|^{2\zeta+1}} \leq
2\pi C_0 (1-\zeta)^{-1}
\end{gather}
where
\begin{equation}
 C_0 := \frac 1{2^{3N-3} \pi^{3N-5} \gamma^6} \left\{\int \Bigl|\widehat{\bigl|
v_{23}\bigr|^{\frac 12}} (\gamma^{-1}p_{y_2})\Bigr|^2 d^3 p_{y_2} \right\}
 \left\{\int \Bigl|v_{12}\bigl( (2\mu_{12})^{-1} y_1 \bigr)\Bigr| d^3 y_1
\right\}
\end{equation}
is finite due to $v_{ik} \in L^1 (\mathbb{R}^3)$. Therefore,
$\sup_{\epsilon, \epsilon' >0} \| D^{(1)} (\epsilon, \epsilon') \| <\infty $. The proof that
$\sup_{\epsilon, \epsilon'>0} \| D^{(2)} (\epsilon, \epsilon') \| <\infty $ is trivial (c. f. proof of
Lemma~2 in \cite{jpasubmit}).
\end{proof}

The following corollaries to Theorem~\ref{3.8;8} provide further necessary estimates. We
agreed to set $\|\varphi(\epsilon)\| = 1$ for $\epsilon \in [0, \beps]$, while for $\psi(\epsilon)$ we can
prove
\begin{corollary} \label{20.9;6}
\begin{equation}
  C_\psi := \sup_{\epsilon >0}\|\psi(\epsilon)\| < \infty
\end{equation}
\end{corollary}
\begin{proof}
 The proof easily follows from (\ref{11.01.12/1}) and Theorem~\ref{3.8;8}.
\end{proof}
\begin{corollary}\label{24.9;9}
 For $\epsilon , \epsilon' >0$ let us set
\begin{equation}\label{11.10;1}
 \eta(\epsilon, \epsilon') := |v_{12}|^{\frac12} B_{12}^{-1}(\epsilon') \bigl( H_0 +
\epsilon\bigr)^{-\frac 12} \varphi(\epsilon) ,
\end{equation}
where $B_{12} (\epsilon)$ was defined in (\ref{7.8;15}) and $\zeta \in (0, 1)$. Then
\begin{equation}
 C_\eta := \sup_{\epsilon >0} \sup_{0 < \epsilon' \leq \epsilon}\|\eta(\epsilon, \epsilon')\| < \infty . 
\end{equation}
\end{corollary}
\begin{proof}
Using (\ref{11.01.12/1}) we can write
\begin{equation}
  \eta(\epsilon, \epsilon') =  |v_{12}|^{\frac12} B_{12}^{-1}(\epsilon') \psi(\epsilon) . 
\end{equation}
Now the result follows from (\ref{6.8;6}).
\end{proof}
Note that due to continuity of $\varphi(\epsilon)$ on $[0, \beps]$ and
compactness of the interval
\begin{equation}
 \sup_{\substack{\epsilon_{1, 2} \in [0, \beps] \\ |\epsilon_1 - \epsilon_2| <
d}} \|\varphi(\epsilon_1) - \varphi(\epsilon_2)\| = \hbox{o} (d)  \quad \quad (\textnormal{when $d \to 0$}).  \label{7.8;25}
\end{equation}
The following is also true
\begin{corollary}\label{19.9;10}
Let $\alpha \in [1, \frac 32 )$. For all $\varepsilon > 0$ there exists
$\delta > 0$ such that
\begin{equation}
\sup_{\epsilon ' > 0}  \sup_{s \in \mathbb{R}^3} \| G_j^\alpha (s, \epsilon')
(\varphi(\epsilon_1)  - \varphi(\epsilon_2))\| < \varepsilon \label{7.8;21}
\end{equation}
for all $|\epsilon_1 - \epsilon_2 | < \delta$ and $j = 1, 2, \ldots, N$.
\end{corollary}
\begin{proof}
 Without loosing generality let us set $j = 1$. For any $r \in \mathbb{R}_+ / \{0\}$ we have
\begin{gather}
 \sup_{\epsilon ' > 0}  \sup_{s \in \mathbb{R}^3} \| G_1^\alpha (s, \epsilon')
(\varphi(\epsilon_1)  - \varphi(\epsilon_2))\|^2 =
\sup_{0 <\epsilon ' <r}  \sup_{s \in \mathbb{R}^3} \| G_1^\alpha (s, \epsilon')
(\varphi(\epsilon_1)  - \varphi(\epsilon_2))\|^2 \nonumber\\
\leq r^{-\alpha} \| \varphi(\epsilon_1)  - \varphi(\epsilon_2) \|^2 + \sup_{0
<\epsilon ' <r}  \sup_{s \in \mathbb{R}^3} \int_{|p_1 + s|^2 \leq r}
\frac{\left|\hat \varphi(\epsilon_1)  - \hat \varphi(\epsilon_2)\right|^2
}{\left[ (p_1 + s)^2 + \epsilon' \right]^\alpha} d^{3N-3}p,
\end{gather}
where $\hat \varphi = F_1 \varphi$ and $p := (p_1, \ldots, p_{N-1})$. Let us choose $\alpha' \in (\alpha, \frac32 )$. For
the last term we can write
\begin{equation}
 \sup_{0 <\epsilon ' <r}  \sup_{s \in \mathbb{R}^3} \int_{|p_1 + s|^2 \leq r}
\frac{\left|\hat \varphi(\epsilon_1)  - \hat \varphi(\epsilon_2)\right|^2
}{\left[ (p_1 + s)^2 + \epsilon' \right]^\alpha} d^{3N-3}p \leq 2 (2r)^{\alpha'
- \alpha} M_{\alpha'},
\end{equation}
where the constant
\begin{equation}
 M_{\alpha'} := \sup_{\epsilon >0} \sup_{\epsilon ' >0}  \sup_{s \in
\mathbb{R}^3} \|G_1^{\alpha'} (s, \epsilon') \varphi(\epsilon)\|^2
\end{equation}
is finite by Theorem~\ref{3.8;8}. Summarizing,
\begin{equation}
\sup_{\epsilon ' > 0}  \sup_{s \in \mathbb{R}^3} \| G_1^\alpha (s, \epsilon')
(\varphi(\epsilon_1)  - \varphi(\epsilon_2))\|^2 \leq  r^{-\alpha} \|
\varphi(\epsilon_1)  - \varphi(\epsilon_2) \|^2 + 2 (2r)^{\alpha' - \alpha}
M_{\alpha'}
\end{equation}
Now (\ref{7.8;21}) for $j = 1$ would hold if the following two inequalities are fulfilled
\begin{gather}
4 (2r)^{\alpha' - \alpha} M_{\alpha'} \leq \varepsilon^2  \label{7.8;22} , \\
2 r^{-\alpha} \| \varphi(\epsilon_1)  - \varphi(\epsilon_2) \|^2 \leq
\varepsilon^2 \label{7.8;23} . 
\end{gather}
Let us fix $r$ so that (\ref{7.8;22}) is satisfied. Then by (\ref{7.8;25}) we can always
choose an appropriate $\delta >0$ so that
(\ref{7.8;23}) holds for $|\epsilon_1 - \epsilon_2| < \delta$.
\end{proof}

Using (\ref{11.01.12/1}) let us define
\begin{equation}
 \tilde \varphi(\epsilon)  := - \mu^{-1} (\epsilon) V \psi (\epsilon) \quad \quad (\epsilon \in (0, \beps]). \label{18.9;1}
\end{equation}
Obviously
\begin{equation}
\bigl(H_0 + \epsilon\bigr)^{\frac 12}  \varphi(\epsilon) = \bigl(H_0 + \epsilon\bigr) \psi(\epsilon) =  \tilde \varphi(\epsilon).
\end{equation}
\begin{corollary}\label{18.9;4}
The norm limits $\psi(0)= \lim_{\epsilon \to +0} \psi(\epsilon)$ and $\tilde \varphi(0)=
\lim_{\epsilon \to +0} \tilde \varphi(\epsilon)$ exist and $\tilde \varphi(\epsilon), \psi(\epsilon)$ are norm--continuous on $[0, \beps]$. The following is also true for
$L > 0$
\begin{equation}\label{19.9;3}
 \sup_{\epsilon \in [0, \beps]}\left\|\bigl(1-\chi_L (|x|^2)\bigr) \tilde \varphi(\epsilon)\right\| = \hbox{o} (L^{-1}) \quad \quad (\textnormal{when $L \to \infty$})
\end{equation}
\end{corollary}
\begin{proof}
Let us first prove the following statement: for all $\varepsilon > 0$ there exist $\delta , \delta' > 0$ such that
\begin{gather}
\sup_{\substack{\epsilon_{1, 2} \in (0, \beps] \\ |\epsilon_1 - \epsilon_2| <\delta}} \|\psi (\epsilon_1) - \psi(\epsilon_2)\| < \varepsilon \label{18.9;2} \\
\sup_{\substack{\epsilon_{1, 2} \in (0, \beps] \\ |\epsilon_1 - \epsilon_2| <\delta'}} \|\tilde \varphi (\epsilon_1) - \tilde \varphi (\epsilon_2)\|  < \varepsilon  \label{18.9;3}
\end{gather}
Note that (\ref{18.9;3}) easily follows from (\ref{18.9;2}) since $V$ is relatively $H_0$ bounded with a relative bound zero and
$\sup_{\epsilon \in (0, \beps]} \|H_0 \psi(\epsilon)\| < \infty$ (see f. e. Lemma~1 in \cite{1}). For the same reason
$\sup_{\epsilon \in (0, \beps] } \|\tilde \varphi(\epsilon)\|< \infty$. The norm--continuity of $\psi(\epsilon), \tilde \varphi(\epsilon)$ on $[0, \beps]$
follows directly from (\ref{18.9;2})--(\ref{18.9;3}). So suppose that $0 <\epsilon_1 \leq \epsilon_2 \leq \beps$. Then
\begin{gather}
 \|\psi (\epsilon_1) - \psi(\epsilon_2)\| \leq \left\|\bigl(H_0 + \epsilon_1\bigr)^{-\frac 12}\bigl[\varphi(\epsilon_1) - \varphi(\epsilon_2)\bigr]\right\| +
\left\|\Bigl[\bigl(H_0 + \epsilon_1\bigr)^{-\frac 12} - \bigl(H_0 + \epsilon_2\bigr)^{-\frac 12}\Bigr]\varphi(\epsilon_2)\right\| \nonumber\\
\leq \left\|G_1 (0, \epsilon_1) \bigl[\varphi(\epsilon_1) - \varphi(\epsilon_2)\bigr]\right\| + \left\|\Bigl[\bigl(H_0 + \epsilon_1\bigr)^{-\frac 12} - \bigl(H_0 + \epsilon_2\bigr)^{-\frac 12}\Bigr]\varphi(\epsilon_2)\right\| . 
\end{gather}
By Corollary~\ref{19.9;10} we can choose $\delta > 0$ so that
\begin{equation}\label{18.9;9}
 \|\psi (\epsilon_1) - \psi(\epsilon_2)\| < \varepsilon/2 + \left\|\Bigl[\bigl(H_0 + \epsilon_1\bigr)^{-\frac 12} - \bigl(H_0 + \epsilon_2\bigr)^{-\frac 12}\Bigr]\varphi(\epsilon_2)\right\|
\end{equation}
for $|\epsilon_1 - \epsilon_2| < \delta$. Let us introduce $f:\mathbb{R}_+ \to \mathbb{R}_+ $ such that $f(r) = r^{2/3}$ for $r \in [0, 1]$ and $f(r) = 1$ for $r \geq 1$. Now we consider
the last term in (\ref{18.9;9}).
\begin{gather}
 \left\|\Bigl[\bigl(H_0 + \epsilon_1\bigr)^{-\frac 12} - \bigl(H_0 + \epsilon_2\bigr)^{-\frac 12}\Bigr]\varphi(\epsilon_2)\right\| \nonumber\\
\leq
\left\|\Bigl[\bigl(H_0 + \epsilon_1\bigr)^{-\frac 12} - \bigl(H_0 + \epsilon_2\bigr)^{-\frac 12}\Bigr]f(H_0 + \epsilon_1)\right\| \left\|\left[f(H_0 + \epsilon_1)\right]^{-1}\varphi(\epsilon_2)\right\|  . \label{18.9;10}
\end{gather}
The expression $f(H_0 + \epsilon_1)$ in (\ref{18.9;10}) is to be understood in terms of functional calculus of self--adjoint operators. 
Using the Fourier transform it is easy to see that
\begin{equation}
 \left\|\Bigl[\bigl(H_0 + \epsilon_1\bigr)^{-\frac 12} - \bigl(H_0 + \epsilon_2\bigr)^{-\frac 12}\Bigr]f(H_0 + \epsilon_1)\right\| \leq|\epsilon_1 - \epsilon_2|^{1/6}
\end{equation}
for $|\epsilon _2 - \epsilon_1| \leq 1$. The second norm in (\ref{18.9;10}) can be estimated as follows 
 \begin{equation}
  \left\|\left[f(H_0 + \epsilon_1)\right]^{-1}\varphi(\epsilon_2)\right\| \leq 1 +  \|G^{\frac 43}_1 (0 , \epsilon_1) \varphi(\epsilon_2) \| , 
 \end{equation}
where  the last norm is uniformly bounded by Theorem~\ref{3.8;8}. Thus we can always set $\delta > 0$ to ensure that the last term in  (\ref{18.9;9}) is less that $\varepsilon/2$. Eq.~(\ref{19.9;3}) is a
trivial consequence of the norm--continuity of $\tilde \varphi(\epsilon)$ on $[0, \beps]$.
\end{proof}

\section{Proof of Main Theorem}\label{1.08;2}

Throughout this section we assume that the N--particle Hamiltonian (\ref{1.08;3})--(\ref{1.08;4}) satisfies the assumption R1. 
Let us by $\mathfrak{C}_j$ for $j = 1, 2, \ldots, N$ denote the subsystem
containing $N-1$ particles, where the particle $j$ is missing.
For each subsystem $\mathfrak{C}_j$ we introduce the operators
\begin{gather}
 K_j (\epsilon):= -\bigl(H_0 + \epsilon\bigr)^{-1/2} V_{\{j\}} \bigl(H_0 +
\epsilon\bigr)^{-1/2} \\
L_j (\epsilon) := -\bigl(H_0 + \epsilon\bigr)^{-1/2} (V-V_{\{j\}}) \bigl(H_0 +
\epsilon\bigr)^{-1/2} ,
\end{gather}
where $\epsilon >0$ and $V_{\{j\}}$ was defined in (\ref{7.8;1a}).
Clearly, $K_j, L_j \in \mathfrak{B}(L^2(\mathbb{R}^{3N-3}))$ and
\begin{equation}
 K (\epsilon)= K_j(\epsilon) + L_j(\epsilon) \quad \quad (j = 1, \ldots, N).
\end{equation}
After an appropriate Fourier transform $K_j (\epsilon)$ becomes the BS operator for the subsystem
$\mathfrak{C}_j$.
Suppose that the subsystem $\mathfrak{C}_j$ (as a system of $N-1$ particles) is
at critical coupling and satisfies the conditions of Theorem~\ref{3.8;3}.
Using Theorem~\ref{3.8;3} and (\ref{7.01;41}) for each such subsystem  we can define $\beps_j$, 
$\varphi_j (\epsilon) \in L^2 (\mathbb{R}^{3N-6})$, $\mu_j (\epsilon)$. Inequality
(\ref{mubound}) reads then $1- \mu_j (\epsilon) \geq a_\mu^{(j)}\epsilon$	for
$\epsilon \in [0, \beps_j]$.
It is convenient to set $\beps := \min_j \beps_j$ and $a_\mu := \min_j
a_\mu^{(j)}$ (where minima are taken over all such $j$ for which
$\mathfrak{C}_j$ satisfies the conditions of Theorem~\ref{3.8;3}). The function $\mu_j (\epsilon)$ is defined as in (\ref{7.01;41}), where $\beps := \min_j \beps_j$. 
Note that $\mu_j (\epsilon) \geq (1+\omega/2)^{-1}$, where $\omega$ is defined in R1. 
Let us redefine the definitions saying that for $\epsilon \in [0,  \beps]$ the
functions $\varphi_j (\epsilon) , \tilde \varphi_j (\epsilon)$ are defined through
Theorem~\ref{3.8;3} and (\ref{18.9;1}) respectively 
and $\varphi_j (\epsilon) = \tilde \varphi_j (\epsilon) = 0$ if $\epsilon > \beps$.

 Now suppose that the subsystem $\mathfrak{C}_1$ is at critical coupling. We use
Jacobi coordinates $x_1, \ldots, x_{N-1}$,
which we have already introduced in Sec.~\ref{1.08;1}. Then $\varphi_1 (\epsilon)$
depends explicitly
on the coordinates as $\varphi_1 (\epsilon; x_r)$. Similar to (\ref{7.8;31}) for $\epsilon >0$ we
define the projection operator $\hat P_1 (\epsilon)$,
which acts on $\hat f (p_1, x_r) \in L^2(\mathbb{R}^{3N-3})$ as follows
\begin{equation}\label{29.8;3}
 \hat P_1 (\epsilon)  \hat f := 
 \varphi_1 (\epsilon + p_1^2; x_r) {\displaystyle \int \hat f (p_1, x_r)}
\varphi^*_1 (\epsilon + p_1^2; x_r) \: d^{3N-6} x_r 
\end{equation}
Its Fourier--transformed version is $P_1 (\epsilon) := \mathcal{F}_1^{-1}
\hat P_1 (\epsilon)\mathcal{F}_1$, where $\mathcal{F}_1$ was defined in
(\ref{7.8;51}).
Additionally, let us introduce the operator functions $\mathcal{P}_1 , \mathfrak{m}_1 : \mathbb{R}_+ \to \mathfrak{B} (L^2 (\mathbb{R}^{3N-3}))$, and 
$\mathfrak{g}_1 : \mathbb{R}_+ / \{0\}\to \mathfrak{B} (L^2 (\mathbb{R}^{3N-3}))$, which act on
$f(x_1, x_r) \in L^2(\mathbb{R}^{3N-3})$ as follows
\begin{gather}
 \mathfrak{m}_1 (\epsilon)f  = \mathcal{F}_1^{-1}  \mu_1 (\epsilon + p_1^2)
(\mathcal{F}_1 f) , \\
\mathcal{P}_1 (\epsilon)  := \mathfrak{m}_1  (\epsilon) P_1  (\epsilon) , \\
\mathfrak{g}_1 (\epsilon)f  = \mathcal{F}_1^{-1}  \left(\left[1-\mu_1 (\epsilon
+ p_1^2) \right]^{-\frac 12} - 1\right) (\mathcal{F}_1 f).  \label{24.9;26}
\end{gather}
Similarly, for each subsystem $\mathfrak{C}_j$ at critical coupling we define
$\hat P_j , P_j, \mathcal{P}_j , \mathfrak{g}_j ,  \mathfrak{m}_j, \tilde \varphi_j (\epsilon)$,
where for each $j$ one has to choose appropriate Jacobi coordinates.
If the subsystem $\mathfrak{C}_j$ is not at critical coupling we simply set
$\hat P_j , P_j, \mathcal{P}_j , \mathfrak{g}_j, \mathfrak{m}_j , \tilde \varphi_j (\epsilon) = 0$.
According to Lemma~\ref{7.8;53} the operator
\begin{equation}
\mathcal{R}_j (\epsilon) := \mathfrak{g}_j (\epsilon) P_j (\epsilon)= P_j (\epsilon) \mathfrak{g}_j (\epsilon) \quad \quad (\epsilon >0)\label{20.9;11}
\end{equation}
satisfies
\begin{equation}\label{8.8;1}
  \mathcal{R}_j (\epsilon) = \bigl( 1- \mathcal{P}_j (\epsilon) \bigr)^{-1/2} -
1.
\end{equation}

From definitions it is clear that the operator $\mathcal{P}_j (\epsilon)$ is not norm--continuous. We shall now construct its norm--continuous
analogue. Let us introduce a continuous function $\theta: \mathbb{R} \to \mathbb{R}$ depending on a parameter $\gamma >0$ such that
$\theta(r) = 1$ if $r\in [1-\gamma , 1+\gamma]$; $\theta(r) = 0$ if $r\in \mathbb{R}/(1-2\gamma , 1+2\gamma)$; in the intervals $[1-2\gamma, 1-\gamma]$ and
$[1+\gamma, 1+2\gamma]$ the function $\theta(r)$ is linear, see Fig.~\ref{28.08;1}. Recall that $\mu_j (\epsilon)$ on $[0, \beps]$ is continuous and monotone decreasing. Let us set
$\gamma$ so that the following equation is fulfilled
\begin{equation}
 1 - 2\gamma = \max_j \mu_j (\beps/2) ,
\end{equation}
where the maximum is taken over all such $j$ that $\mathfrak{C}_j$ is at critical coupling. If $\mathfrak{C}_1$ is at critical coupling then the operators
$\mathfrak{m}^{(c)}_1 (\epsilon), \mathcal{P}^{(c)}_1 (\epsilon) $ given by expressions 
\begin{gather}
 \mathfrak{m}^{(c)}_1 (\epsilon)f  := \mathcal{F}_1^{-1}  \theta\bigl( \mu_1 (\epsilon + p_1^2) \bigr)
(\mathcal{F}_1 f) \label{29.8;11}\\
\mathcal{P}^{(c)}_1 (\epsilon)  :=  \mathfrak{m}^{(c)}_1 (\epsilon) P_1  (\epsilon)
\end{gather}
are norm--continuous on $[0, \infty)$ (the superscript ``c'' stands for continuous). Similarly, using appropriate coordinates
one defines $\mathfrak{m}^{(c)}_j (\epsilon), \mathcal{P}^{(c)}_j (\epsilon) $ if
$\mathfrak{C}_j$ is at critical coupling and sets  $\mathfrak{m}^{(c)}_j (\epsilon), \mathcal{P}^{(c)}_j (\epsilon) = 0$ otherwise.
\begin{figure}
\begin{center}
\includegraphics[width=0.7\textwidth]{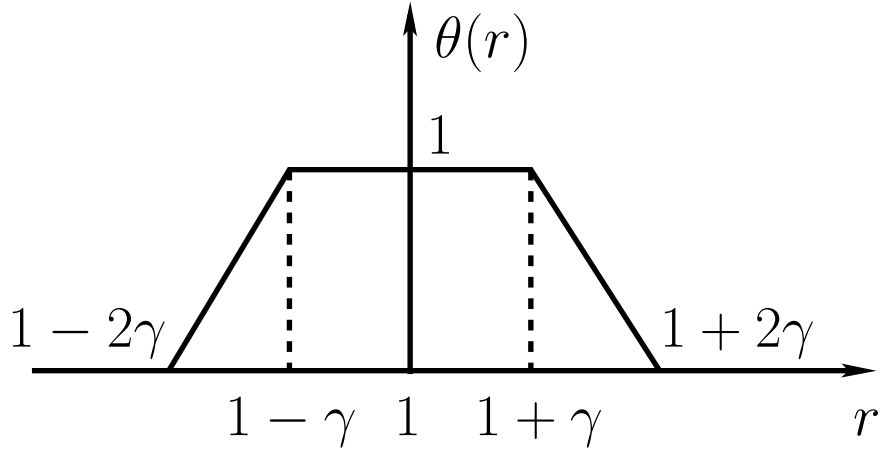}
\caption{Behavior of the continuous function $\theta(r)$.}
\label{28.08;1}
\end{center}
\end{figure}
It is easy to see that $\mathfrak{m}_j (\epsilon) - \mathfrak{m}^{(c)}_j (\epsilon) \geq 0$, therefore 
\begin{equation}\label{28.08;3}
\mathcal{P}_j (\epsilon) - \mathcal{P}^{(c)}_j (\epsilon) \geq 0 \quad \quad (j = 1, \ldots, N).
\end{equation}
By construction of $\mathcal{P}^{(c)}_j$ and (\ref{19.9;11})
\begin{equation}\label{19.9;1}
 \pmb{\eta} := \max_{j = 1, \ldots, N} \sup_{\epsilon >0 } \sup \sigma\left( K_j
(\epsilon) - \mathcal{P}^{(c)}_j (\epsilon)\right) < 1.
\end{equation}

The following theorem, which is used for counting the eigenvalues, is the
central ingredient in the proof of Theorem~\ref{1.08;9}.
\begin{theorem}\label{8.8.;8}
Suppose that $H$ defined in (\ref{1.08;3})--(\ref{1.08;4}) is such that $\sigma_{ess} (H) = [0, \infty)$. Then the following relation
holds for $N_\epsilon$ defined in (\ref{7.8;55})
\begin{equation}
 N_\epsilon = \# (\evs (\mathcal{T}_k (\epsilon)) > 1) \quad \quad (k = 1,
\ldots, N),
\end{equation}
where $\epsilon >0$ and
\begin{gather}
 \mathcal{T}_k (\epsilon) := K(\epsilon) - \sum_{i=1}^k \mathcal{P}_i (\epsilon)
+ \mathcal{M}_k (\epsilon) , \\
\mathcal{M}_1 (\epsilon):= \mathcal{R}_1 L_1 \mathcal{R}_1 + \mathcal{R}_1 L_1 +
L_1 \mathcal{R}_1 , \label{8.8;2}
\end{gather}
and the operators $\mathcal{M}_k (\epsilon)$ for $k = 2, 3, \ldots, N$ are
determined through the following recurrence relation
\begin{gather}
\mathcal{M}_k (\epsilon) = \bigl(1+ \mathcal{R}_k\bigr)
\mathcal{M}_{k-1}\bigl(1+ \mathcal{R}_k\bigr)  + \mathcal{R}_k \left\{  L_k -
\sum_{i=1}^{k-1} \mathcal{P}_i \right\} \mathcal{R}_k \nonumber\\
+ \mathcal{R}_k \left\{  L_k - \sum_{i=1}^{k-1} \mathcal{P}_i \right\}  +
\left\{  L_k - \sum_{i=1}^{k-1} \mathcal{P}_i \right\} \mathcal{R}_k . \label{8.8.7}
\end{gather}
Besides, $\sigma_{ess} (\mathcal{T}_k (\epsilon )) \cap (1, \infty) = \emptyset$ for $k = 1, \ldots, N$.
\end{theorem}
\begin{proof}
 The proof consists in applying the BS principle $N$ times. As
the first step let us write
\begin{equation}
 1-K(\epsilon) = 1- \mathcal{P}_1 (\epsilon) - \left[K(\epsilon) - \mathcal{P}_1
(\epsilon) \right] ,
\end{equation}
where the operator on the lhs has $N_\epsilon$ negative eigenvalues, see (\ref{7.8;55}). By Theorem~\ref{31.07;16} we also have
$\sigma_{ess} ( 1-K(\epsilon)) \subset [0, \infty)$. Obviously, for $\epsilon >0$
\begin{equation}
 1- \mathcal{P}_j (\epsilon) \geq 1- \mu_j (\epsilon) >0 \quad \quad (j = 1, \ldots, N)
\end{equation}
Applying the BS principle
(see a remark after Theorem~\ref{31.07;16}) we obtain
\begin{equation}
 N_\epsilon = \#\bigl(\evs(  \mathcal{T}_1 (\epsilon)) > 1 \bigr)
\end{equation}
where
\begin{equation}
 \mathcal{T}_1 (\epsilon ) = \bigl(  1- \mathcal{P}_1 \bigr)^{-1/2} \bigl[ K -
\mathcal{P}_1\bigr] \bigl(  1- \mathcal{P}_1 \bigr)^{-1/2} .
\end{equation}
Besides, $\sigma_{ess} (\mathcal{T}_1 (\epsilon )) \cap (1, \infty) = \emptyset$. Using  (\ref{8.8;1}) and $K(\epsilon) = K_1(\epsilon) + L_1(\epsilon)$ we get
\begin{equation}
 \mathcal{T}_1 (\epsilon ) = K - \mathcal{P}_1 + L_1 \mathcal{R}_1 +
\mathcal{R}_1 L_1 + \mathcal{R}_1 L_1 \mathcal{R}_1 = K - \mathcal{P}_1 +
\mathcal{M}_1
\end{equation}
where we have used $\mathcal{R}_1(K_1 - \mathcal{P}_1) = 0$ and (\ref{8.8;2}). Now
we do the second step and use that
\begin{equation}
 1 - \mathcal{T}_1 (\epsilon ) = 1- \mathcal{P}_2 (\epsilon) -
\left[\mathcal{T}_1 (\epsilon) - \mathcal{P}_2 (\epsilon) \right]
\end{equation}
has $N_\epsilon$ negative eigenvalues (counting multiplicities). Therefore, again by the BS principle
\begin{equation}
 N_\epsilon = \#\bigl(\evs(  \mathcal{T}_2 (\epsilon)) > 1 \bigr),
\end{equation}
where
\begin{gather}
 \mathcal{T}_2 (\epsilon ) = \bigl(  1- \mathcal{P}_2 \bigr)^{-1/2} \bigl[
\mathcal{T}_1 (\epsilon) - \mathcal{P}_2\bigr] \bigl(  1- \mathcal{P}_2
\bigr)^{-1/2} \nonumber\\
= \bigl(  1+ \mathcal{R}_2 \bigr) \bigl[ K - \mathcal{P}_1  - \mathcal{P}_2 +
\mathcal{M}_1  \bigr] \bigl(  1+ \mathcal{R}_2 \bigr) \nonumber\\
= K - \mathcal{P}_1 - \mathcal{P}_2  + \bigl(  1+ \mathcal{R}_2 \bigr)
\mathcal{M}_1 \bigl(  1+ \mathcal{R}_2 \bigr) + \mathcal{R}_2 \bigl[ K -
\mathcal{P}_1  - \mathcal{P}_2 \bigr] \mathcal{R}_2 \nonumber\\
+ \mathcal{R}_2 \bigl[ K - \mathcal{P}_1  - \mathcal{P}_2 \bigr] +  \bigl[ K -
\mathcal{P}_1  - \mathcal{P}_2 \bigr] \mathcal{R}_2
\end{gather}
has $N_\epsilon$ eigenvalues larger than one. From the BS principle it also follows that $\sigma_{ess} (\mathcal{T}_2 (\epsilon )) \cap (1, \infty) = \emptyset$.
For the last three terms in square
brackets we substitute $K = K_2 + L_2$ and use $\mathcal{R}_2(K_2 -
\mathcal{P}_2) = (K_2 - \mathcal{P}_2)\mathcal{R}_2  = 0$.
This leads to
\begin{equation}
 \mathcal{T}_2 = K - \mathcal{P}_1 - \mathcal{P}_2 + \mathcal{M}_2 ,
\end{equation}
where $\mathcal{M}_2 $ is defined through (\ref{8.8.7}). Proceeding in the same way, that is writing each time
\begin{equation}
 1 - \mathcal{T}_k (\epsilon ) = 1- \mathcal{P}_{k+1} (\epsilon) -
\left[\mathcal{T}_k (\epsilon) - \mathcal{P}_{k+1} (\epsilon) \right] \quad (k = 3, \ldots, N-1)
\end{equation}
and applying the BS principle we prove the theorem.
\end{proof}
Let us define the norm--continuous operator function $\mathcal{G}_N : \mathbb{R}_+ \to \mathfrak{B}(L^2 (\mathbb{R}^{3N-3}))$ as
\begin{equation}
 \mathcal{G}_N (\epsilon) := K(\epsilon) - \sum_{j=1}^N  \mathcal{P}^{(c)}_j (\epsilon).  \label{13.8;11}
\end{equation}
\begin{lemma}\label{13.8;4}
 There exist $q, \varepsilon_0 >0$ such that for $\epsilon \in [0, \varepsilon_0]$
\begin{equation}
 \sigma_{ess}  \bigl(\mathcal{G}_N (\epsilon)\bigr) \subset (- \infty , 1-q).
\end{equation}
\end{lemma}
The proof of Lemma~\ref{13.8;4} would be given later. Let us define 
\begin{equation}
 \mathcal{A}_\infty := \left\{ W:\mathbb{R}_+ /\{0\} \to \mathfrak{B}(L^2 (\mathbb{R}^{3N-3})) \Bigl| \; \;  \sup_{\epsilon >0} \| W(\epsilon)\| < \infty \right\} . 
\end{equation}
By definition $W \in \mathfrak{J}\subset \mathcal{A}_\infty$  {\it iff} for all $\beta >0$ there exists a
decomposition $W = W^{B}  + W^{HS}$, 
where $W^{B}, W^{HS} \in \mathcal{A}_\infty$ satisfy the following inequalities
\begin{gather}
\sup_{\epsilon >0} \| W^{B} (\epsilon)\| < \beta \\
\sup_{\epsilon >0} \| W^{HS} (\epsilon)\|_{HS} < \infty .\label{13.8;8}
\end{gather}
(Eq.~(\ref{13.8;8}) implies that $W^{HS} (\epsilon)$ is a Hilbert-Schmidt operator for all $\epsilon >0$). Obviously, if $W \in \mathfrak{J}$ then $W(\epsilon)$ is a 
compact operator for all $\epsilon >0$. 
\begin{lemma}\label{13.8;9}
 $\mathcal{A}_\infty$ is an algebra and $\mathfrak{J}$ is a two--sided ideal in
$\mathcal{A}_\infty$.
\end{lemma}
The proof of Lemma~\ref{13.8;9} is a trivial consequence of the Hilbert--Schmidt class properties and we omit it.
\begin{lemma}\label{13.8;1}
 The function $\mathcal{M}_N\in \mathcal{A}_\infty$ defined in (\ref{8.8.7}) is such that $\mathcal{M}_N  \in \mathfrak{J}$.
\end{lemma}
The proof of Lemma~\ref{13.8;1} would be given later.

\begin{proof}[Proof of Theorem~\ref{1.08;9}]
 Let us assume by contradiction that $N_\epsilon \to \infty$. Then by Theorem~\ref{8.8.;8}
\begin{equation}
 \lim_{\epsilon \to +0 }\# (\evs (\mathcal{T}_N (\epsilon)) > 1) = \infty . \label{13.8;31a}
\end{equation}
Let us define
\begin{equation}
 \mathcal{T'}_N (\epsilon) := \mathcal{T}_N (\epsilon) + \sum_{j = 1}^N \left[ \mathcal{P}_j (\epsilon) - \mathcal{P}^{(c)}_j (\epsilon) \right].
\end{equation}
Using (\ref{13.8;11}) we can write
\begin{equation}
 \mathcal{T'}_N(\epsilon) = \mathcal{G}_N (0) + \left\{ \mathcal{G}_N (\epsilon) -
\mathcal{G} (0)\right\} + \mathcal{M}_N(\epsilon) . \label{13.8;32}
\end{equation}
By norm--continuity of $ \mathcal{G} (\epsilon)$ and by Lemma~\ref{13.8;4} there exist $\varepsilon_0, q >0$ such that for all $\epsilon \in (0, \varepsilon_0)$
\begin{gather}
  \|\mathcal{G}_N (\epsilon) - \mathcal{G}_N (0)\| < q , \label{13.8;33}\\
 \mathcal{G}_N (0) \leq 1 - 3q + \mathcal{C}_f \label{13.8;34}
\end{gather}
where $\mathcal{C}_f$ is a fixed finite rank self--adjoint operator.
By Lemma~\ref{13.8;1} we can write the decomposition
\begin{equation}
 \mathcal{M}_N (\epsilon) = \mathcal{M}_N^B (\epsilon) + \mathcal{M}_N^{HS} \label{13.8;35}
(\epsilon) ,
\end{equation}
where
\begin{gather}
\|\mathcal{M}_N^B (\epsilon)\|< q \label{13.8;36} \\
\sup_{\epsilon \in (0,\beps]} \| \mathcal{M}_N^{HS} (\epsilon) \|_{HS} =\vartheta < \infty . \label{13.8;37}
\end{gather}
On one hand, from (\ref{13.8;31a}) we infer that for any $n \in \mathbb{Z}_+$ there is
$\epsilon \in (0, \varepsilon_0)$ and an orthonormal set
$\phi_1 , \ldots, \phi_n \in L^2 (\mathbb{R}^{3N-3})$ such that $(\phi_i, \mathcal{T}_N(\epsilon) \phi_i) >
1$ holds for $i = 1, \ldots, n$. Due to (\ref{28.08;3}) $(\phi_i, \mathcal{T'}_N(\epsilon) \phi_i) >
1$ holds as well for $i = 1, \ldots, n$. With (\ref{13.8;32})--(\ref{13.8;36}) this results in
\begin{equation}
 \bigl| \bigl(\phi_i , [\mathcal{C}_f + \mathcal{M}_N^{HS}
(\epsilon)]\phi_i)\bigr| > q \quad \quad (i = 1, \ldots, n).\label{13.8;38}
\end{equation}
On the other hand, from Lemma~\ref{30.07;2} and (\ref{13.8;37}), (\ref{13.8;38}) it follows that $n
\leq(\|\mathcal{C}_f\|_{HS} +\vartheta)^2/q^2$,which contradicts $n$ being arbitrary positive integer.
\end{proof}

Our next aim is to prove Lemma~\ref{13.8;4}. Note that the operator $H_0^{1/2} \bigl(H_0 + \epsilon\bigr)^{-1/2}  $ and its 
adjoint $\bigl(H_0 + \epsilon\bigr)^{-1/2}  H_0^{1/2}$ are uniformly bounded for $\epsilon >0$ (the second operator
can obviously be extended from $D(H_0^{1/2})$ to the whole Hilbert space by the BLT theorem). Let us define
\begin{equation}
  \mathcal{P}'_j (\epsilon) := \bigl(H_0 + \epsilon\bigr)^{-1/2}  H_0^{1/2}  \mathcal{P}^{(c)}_j (0)  H_0^{1/2} \bigl(H_0 + \epsilon\bigr)^{-1/2}  \quad \quad (\epsilon >0)
\end{equation}
\begin{lemma}\label{8.8;19}
The following is true
\begin{equation}
 \lim_{\epsilon \to +0} \left\| \mathcal{P}'_j (\epsilon) - \mathcal{P}^{(c)}_j
(\epsilon) \right\| = 0 \quad \quad (j = 1, \ldots, N)
\end{equation}
 \end{lemma}
\begin{proof}
Note that $\wlim_{\epsilon \to +0} \mathcal{P}'_j (\epsilon) = \mathcal{P}^{(c)}_j
(0)$ because
\begin{equation}
 \sslim_{\epsilon \to +0}  H_0^{1/2} \bigl(  H_0 + \epsilon\bigr)^{-1/2} = 1.
\end{equation}
So the lemma would be proved if we would show that $\mathcal{P}'_j (\epsilon)$
form a Cauchy sequence for $\epsilon \to +0$.
We follow the same recipe as in Lemma~\ref{31.07;2}. It is enough to prove that
\begin{equation}
 \mathcal{D}_j (\epsilon) := \bigl(H_0 + \epsilon\bigr)^{-1/2} H_0^{1/2}  P_j
(0)
\end{equation}
forms a Cauchy sequence for $\epsilon \to +0$. Repeating the arguments from
Lemma~\ref{31.07;2} we obtain for $\epsilon_2 \geq \epsilon_1 >0$
\begin{gather}
\left\|\mathcal{D}_j (\epsilon_2) - \mathcal{D}_j (\epsilon_1)\right\| \nonumber \\
\leq \left\|\left\{(H_0 + \epsilon_2)^{\frac 12} - (H_0 + \epsilon_1)^{\frac
12}\right\}
(H_0 + \epsilon_1)^{-\frac 12} (H_0 + \epsilon_2)^{- \frac 12} H_0^{1/2}  P_j
(0)\right\| \nonumber\\
\leq |\epsilon_2 - \epsilon_1|^{\frac 12} |\epsilon_2|^{- \frac 14}
\left\|H_0^{\frac 12} (H_0 + \epsilon_1)^{- \frac 12} (H_0 + \epsilon_2)^{-
\frac 14 } P_j (0)\right\| \nonumber \\
\leq |\epsilon_2 - \epsilon_1|^{\frac 12} |\epsilon_2|^{- \frac 14}
C_0 , \label{14.8;41}
\end{gather}
where
\begin{equation}
 C_0 := \sup_{\epsilon >0} \left\|(H_0 + \epsilon)^{- \frac 14} P_j (0)\right\|. \label{14.8;23}
\end{equation}
It remains to show that $C_0$ in (\ref{14.8;23}) is finite. Without loss of generality let us set $j = 1$.
 We have
\begin{equation}\label{29.8;1}
  C_0 = \sup_{\epsilon >0}\|(p_1^2 + \epsilon - \Delta_{x_r})^{- \frac 14} \hat P_1 (0) \| =
\sup_{\epsilon >0}\sup_{\epsilon' >0}\|(\epsilon' + \epsilon - \Delta_{x_r})^{- \frac 14} \varphi_1 (\epsilon')\| , 
\end{equation}
where the last norm is that of $L^2 (\mathbb{R}^{3N-6})$. The expression on the rhs of (\ref{29.8;1}) is finite due to Theorem~\ref{3.8;8}.
From (\ref{14.8;41}) it follows that $D_j (\epsilon)$ form
a Cauchy sequence for $\epsilon \to +0$.
\end{proof}

Consider a hermitian sesquilinear form
$\mathfrak{q}_j (f, g ) = \bigl( H_0^{\frac 12} f ,  \mathcal{P}^{(c)}_j (0)  H_0^{\frac 12} g \bigr)$ with the domain $D( H_0^{\frac 12}) \times D( H_0^{\frac 12})$. Let us first
show that $\mathfrak{q}_j$ is bounded. Using norm--continuity we get
\begin{gather}
 \Bigl\| \bigl[\mathfrak{m}^{(c)}_j (0)\bigr]^{\frac 12} P_j (0) H_0^{\frac 12} f\Bigr\| = \lim_{\epsilon \to +0}  \Bigl\| \bigl[\mathfrak{m}^{(c)}_j (\epsilon)\bigr]^{\frac 12} P_j (\epsilon) H_0^{\frac 12} f\Bigr\| \nonumber\\
= \lim_{\epsilon \to +0}  \Bigl\| \bigl[\mathfrak{m}_j (\epsilon)\bigr]^{-1} \bigl[\mathfrak{m}^{(c)}_j (\epsilon)\bigr]^{\frac 12} P_j (\epsilon) K_j (\epsilon)  H_0^{\frac 12} f\Bigr\| \label{29.8;12}
\end{gather}
Recall that we chose $\beps$ so that $\mu_j (\beps) \geq (1 + \omega/2)^{-1}$. Hence, 
\begin{equation}
 \bigl[\mathfrak{m}_j (\epsilon)\bigr]^{-1} \bigl[\mathfrak{m}^{(c)}_j (\epsilon)\bigr]^{\frac 12}  \leq (1+ \omega/2)^{1/2}
\end{equation}
Therefore, we can rewrite (\ref{29.8;12}) as
\begin{gather}
 \Bigl\| \bigl[\mathfrak{m}^{(c)}_j (0)\bigr]^{\frac 12} P_j (0) H_0^{\frac 12} f\Bigr\|  \leq  (1+ \omega/2)^{1/2} \lim_{\epsilon \to +0} \| K_j (\epsilon) H_0^{\frac 12} f \| \\
=   (1+ \omega/2)^{1/2} \lim_{\epsilon \to +0} \| (H_0 + \epsilon)^{- \frac 12 } V_{\{j\}} (H_0 +
\epsilon)^{- \frac 12 }  H_0^{\frac 12} f \| \leq c_j \| (H_0 + \epsilon)^{- \frac
12 }  H_0^{\frac 12} f\| \leq c_j \|f\|, \label{29.8;41}
\end{gather}
where
\begin{equation}
 c_j := (1+ \omega/2)^{1/2} \sup_{\epsilon >0} \| (H_0 + \epsilon)^{- \frac 12 } V_{\{j\}} \| < \infty
\end{equation}
From (\ref{29.8;41}) it follows that $|\mathfrak{q}_j (f, g )| \leq c_j \|f\|\|g\|$. Hence, there exists a self--adjoint operator
$\mathcal{Z}_j \in \mathfrak{B} (L^2 (\mathbb{R}^{3N-3}))$ such that
\begin{equation}
 \mathfrak{q}_j (f, g )  = (f, \mathcal{Z}_j g ). \label{8.1;1}
\end{equation}
It is easy to check that
\begin{equation}
  \mathcal{P}'_j (\epsilon) = \bigl(H_0 + \epsilon\bigr)^{-1/2}  \mathcal{Z}_j \bigl(H_0 + \epsilon\bigr)^{-1/2}  \quad \quad (\epsilon >0).  \label{14.8;1}
\end{equation}

\begin{lemma}\label{8.8;20}
  Suppose that $H$ defined in (\ref{1.08;3})--(\ref{1.08;4}) satisfies R1. Then there
exists $\lambda_0 > 1 $ such that $\sigma_{ess} (\tilde H(\lambda_0)) =
[0, \infty) $, where
\begin{equation}\label{esshard}
  \tilde H(\lambda) := H_0 + \lambda V + \lambda \sum_{i = 1}^N  \mathcal{Z}_i .
\end{equation}
\end{lemma}
\begin{proof}
The operator $\tilde H (\lambda)$ is self--adjoint on $D(H_0) \subset L^2 (\mathbb{R}^{3N-3})$. 
Let $J_s \in C^{2}(\mathbb{R}^{3N-3})$ for $s=1,2,\ldots, N$ denote the Ruelle--Simon partition of
unity, see Definition
3.4 and Proposition 3.5 in \cite{cikon} and also \cite{jpasubmit,coulombgridnev}. One has
$J_s \geq 0$, $\sum_s J^{2}_s =1$ and $J_s (k x) =J_s (x)$ for $k
\geq 1$ and $|x|= 1$. Besides there exists $C > 0$ such that for $i \neq s$
\begin{equation}\label{ims5}
    \supp J_s \cap \{ x | |x| > 1 \} \subset \{x|\; |r_i - r_s | \geq C |x|\} .
\end{equation}
We shall use the following version of the IMS formula, see eq.~(42) in \cite{coulombgridnev}
\begin{equation}\label{06.02/1}
 \Delta = \sum_{s, s' = 1}^N J_s J_{s'} \Delta J_{s'}J_s  + 2\sum_{s=1}^N
|\nabla J_s|^2
\end{equation}
The previous equation can be obtained from the standard IMS formula (Theorem 3.2
in \cite{cikon}) if one notes that $N^2$ functions $J_s J_{s'}$ satisfy
$\sum_{s, s '} (J_s J_{s'})^2 = 1 $. With the help of (\ref{06.02/1}) we can
write
\begin{equation}
    \tilde H (\lambda) = \sum_{s=1}^N J^2_s \left[H_0 + \lambda V_{\{s\}}  + \lambda
\mathcal{Z}_s \right]  J^2_s
+ \sum_{s=1}^N \sum_{\substack{s'=1\\ s\neq s'}}^N J_s J_{s'} \tilde H_{s s'}
(\lambda)J_{s'} J_s + \mathcal{C}_1 + \mathcal{C}_2 + \mathcal{C}_3 \label{ims}
\end{equation}
where
\begin{gather}
\tilde H_{s s'} (\lambda) := H_0 + \lambda V_{\{s, s'\}} \quad \quad (s \neq s')\\
\mathcal{C}_1 := \lambda \sum_{s=1}^N \sum_{\substack{i=1\\i \neq s}}^N J_s^2
v_{is} J_s^2 + 2\sum_{s=1}^N |\nabla J_s|^2 \label{C1} \\
\mathcal{C}_2 := \lambda \sum_{s=1}^N \sum_{\substack{s' = 1\\ s'\neq s}}^N
J_s^2 J_{s'}^2 \left\{ v_{ss'} + \sum_{\substack{i=1\\i \neq s, s'}}^N \left(
v_{is} + v_{is'} \right) \right\}  \label{C2}\\
\mathcal{C}_3 := \lambda \sum_{s=1}^N \sum_{m =1}^N \sum_{\substack{k = 1 \\ k
\neq m}}^N J^2_m  \mathcal{Z}_s  J^2_k \label{C3}
\end{gather}
By the standard arguments in the proof of the HVZ theorem the lemma would be
proved if we can show that $\mathcal{C}_{1, 2, 3}$ are relatively $H_0$ compact
and
the operators under the sums in (\ref{ims}) are non--negative for some $\lambda_0
>1$.  From the derivation of the HVZ theorem, see \cite{teschl,cikon}, it follows that the
operators
$\mathcal{C}_{1,2}$ are relatively $H_0$ compact. To prove the same for
$\mathcal{C}_3$ it suffices to show that
$J^2_m  \mathcal{Z}_s$ for $m \neq s$ is relatively $H_0$ compact.
Without loosing generality we consider only $J^2_2  \mathcal{Z}_1$. Let $\phi \in C_0^\infty(\mathbb{R}_+)$ be such that $\phi(r) = 1$ if $r \in [0, 1]$,
$\phi(r) \in [0,1]$ for $r \in [1, 2]$ and
$\phi(r)=0$ if $r \in [2, \infty)$. Then by definition $\phi_L (r) := \phi(L^{-1}r)$, where $L >0$.
We
have
\begin{equation}\label{71}
 J^2_2  \mathcal{Z}_1 = J^2_2 \phi_L (x_r^2) \mathcal{Z}_1 +
J^2_2  \bigl(1- \phi_L (x_r^2)\bigr) \mathcal{Z}_1 ,
\end{equation}
where $x_r^2 = x_2^2 + \ldots + x_{N-1}^2$.
Obviously, the operator $J^2_2 \phi_L (x_r^2)$ is
relatively $H_0$ compact for all $L >0$.
It remains to show that the norm of the second term in (\ref{71}) can
be made as small as pleased by choosing $L$ large enough. Before we estimate this term let us introduce the operator $\hat Y_1$ similar to the expression in (\ref{29.8;3}),
which acts on $\hat f (p_1, x_r) \in L^2(\mathbb{R}^{3N-3})$ as follows
\begin{equation}\label{29.8;314}
 \hat Y_1 (\epsilon)  \hat f := 
 \tilde \varphi_1 (\epsilon + p_1^2; x_r) {\displaystyle \int \hat f (p_1, x_r)}
\tilde \varphi^*_1 (\epsilon + p_1^2; x_r) \: d^{3N-6} x_r
\end{equation}
while $\tilde \varphi$ was defined in (\ref{18.9;1}). The Fourier--transformed version we denote as $Y_1 := \mathcal{F}_1^{-1}
\hat Y_1 (\epsilon)\mathcal{F}_1$ (the Fourier transform $\mathcal{F}_1$ was defined in (\ref{7.8;51})). By (\ref{8.1;1}), (\ref{14.8;1}) for any $f, g \in C_0^\infty (\mathbb{R}^{3N-3})$
\begin{gather}
 \left| \bigl(f, (1-\phi_L)\mathcal{Z}_1 g\bigr) \right| = \left| \bigl(H_0^{\frac 12} (1-\phi_L)f, \mathcal{P}^{(c)}_1 (0) H_0^{\frac 12} g\bigr) \right| \nonumber\\
= \lim_{\epsilon \to +0} \left|  \bigl((H_0+\epsilon)^{\frac 12} (1-\phi_L)f, \mathfrak{m}_1^{(c)} (\epsilon) P_1 (\epsilon) (H_0+\epsilon)^{\frac 12} g\bigr) \right| \nonumber\\
= \lim_{\epsilon \to +0} \left|  \bigl(\mathfrak{m}_1^{(c)} (\epsilon) f,  (1-\phi_L) Y_1 (\epsilon) g\bigr) \right| \leq \|f\| \|g\| \sup_{\epsilon \in (0, \beps]}\|(1-\phi_L) Y_1 (\epsilon)\| \nonumber\\
\leq \|f\| \|g\| \sup_{\epsilon \in (0, \beps]}\|(1-\phi_L) \tilde \varphi_1 (\epsilon)\| \sup_{\epsilon \in (0, \beps]}\|\tilde \varphi_1 (\epsilon)\| ,
\end{gather}
where the last two norms are that of $L^2 (\mathbb{R}^{3N-6})$. Now by Corollary~\ref{18.9;4} it follows that the norm of the second term on the rhs of (\ref{71}) can be
made as small as pleased by choosing $L$ large and, hence, $\mathcal{C}_3$ is relatively $H_0$ compact.

Due to R1 $\tilde H_{ss'}\bigl( 1+\omega \bigr) \geq
0$ and $\lambda_0 \in (1, 1+\omega)$ ensures that all terms under the double sum in (\ref{ims}) are non--negative operators.
Thus to prove the Lemma it remains to show that with appropriate $\lambda_0 > 1$
the operator
in square brackets in (\ref{ims}) is non--negative. Or equivalently, that for
some $\lambda_0 > 1$ there is a sequence $\epsilon_n \to +0$ such that
\begin{equation}
  H_0 + \lambda_0 V_{\{s\}}  + \lambda_0 \mathcal{Z}_s  +\epsilon_n \geq 0  .  \label{19.9;17}
\end{equation}
We have
\begin{gather}
 H_0 + \lambda V_{\{s\}}  + \lambda \mathcal{Z}_s  +\epsilon_n \nonumber\\
 =\bigl(H_0 + \epsilon_n\bigr)^{1/2}  \left\{ 1 - \lambda \Bigl[ K_s
(\epsilon_n) -
\mathcal{P}^{(c)}_s (\epsilon_n)\Bigr] + \lambda \Bigl(\mathcal{P'}_s (\epsilon_n) -
\mathcal{P}^{(c)}_s (\epsilon_n)\Bigr)\right\}
\bigl(H_0 + \epsilon_n\bigr)^{1/2} , \label{19.9;14}
\end{gather}
where $\epsilon_n \in (0, \beps)$.
By Lemma~\ref{8.8;19}  for any $\varepsilon > 0$ we can choose $\epsilon_n \to +0$ such
that
\begin{equation}
 \Bigl\|\mathcal{P'}_s (\epsilon_n) - \mathcal{P}^{(c)}_s (\epsilon_n)\Bigr\| <
\varepsilon
\end{equation}
Hence, by  (\ref{19.9;1}) from (\ref{19.9;14}) follows
\begin{gather}
 H_0 + \lambda V_{\{s\}}  + \lambda \mathcal{Z}_s  +\epsilon_n \nonumber\\
\geq \bigl(H_0 + \epsilon_n\bigr)^{1/2}  \left\{ 1 - \lambda (\pmb{\eta} +
\varepsilon) \right\}
\bigl(H_0 + \epsilon_n\bigr)^{1/2} .
\end{gather}
Setting $\lambda_0 \leq (\pmb{\eta} + \varepsilon)^{-1}$ makes (\ref{19.9;17}) hold.
Taking $\varepsilon$ sufficiently small we ensure that $\lambda_0 \in(1, 1+\omega)$.

The inclusion
$[0, \infty) \subset \sigma_{ess} (\tilde H(\lambda)) $ for all $\lambda > 0$ is
standard and we omit its proof (it is not used in other proofs anyway). 
\end{proof}
\begin{proof}[Proof of Lemma~\ref{13.8;4}]
Let us first consider $\epsilon \in (0, \beps]$. We can write
\begin{equation}
 \mathcal{G}_N(\epsilon) = \mathcal{G'} (\epsilon) + \sum_{i=1}^N \left[\mathcal{P'}_i (\epsilon) - \mathcal{P}^{(c)}_i (\epsilon)\right] ,
\end{equation}
where
\begin{equation}
  \mathcal{G'} (\epsilon):= K(\epsilon) - \sum_{i=1}^N \mathcal{P'}_i (\epsilon) .
\end{equation}
If we set $A=H_0$ and $B = V + \sum_{i=1}^N \mathcal{Z}_i$ in the
BS principle (Theorem~\ref{31.07;16}) then from (\ref{14.8;1}) and Lemma~\ref{8.8;20} it follows that $\sigma_{ess} \bigl(  \mathcal{G'} (\epsilon)
\bigr) \subset (-\infty , \lambda_0^{-1}]$ if $\epsilon \in (0, \beps]$. Now the result follows from Lemma~\ref{8.8;19} and Theorem~9.5 in \cite{weidmann}.
\end{proof}

\begin{proof}[Proof of Lemma~\ref{13.8;1}]
We shall prove by induction that $\mathcal{M}_k (\epsilon)  \in \mathfrak{J}$
for $k= 1, 2, \ldots, N$. We make the following induction assumption.
Suppose that for $\mathcal{M}_k$ the following holds: (a) $\mathcal{M}_k \in
\mathfrak{J}$; (b) for any $s, s' \geq k+1$ one has
$\mathcal{M}_k \mathcal{R}_s ,  \mathcal{R}_s\mathcal{M}_k ,  \mathcal{R}_s
\mathcal{M}_k \mathcal{R}_{s'}  \in \mathfrak{J}$.
Let us first show that the induction
assumption is fulfilled for $k=1$. That $\mathcal{M}_1 \in \mathfrak{J}$ follows from (\ref{8.8;2})
and Lemma~\ref{19.9;21}.
Checking (b) is also straightforward if one applies Lemmas~\ref{19.9;21}, \ref{19.9;22}, \ref{19.9;23}.
For example,
\begin{gather}
 \mathcal{R}_s\mathcal{M}_1 \mathcal{R}_{s'} = \left[ \mathcal{R}_s  P_1\right]
\left[ \mathcal{R}_1 \mathcal{L}_1  \mathcal{R}_1 \right]
\left[P_1 \mathcal{R}_{s'}  \right] + \left[ \mathcal{R}_s  P_1\right] \left[
\mathcal{R}_1 \mathcal{L}_1 \mathcal{R}_{s'}  \right] \nonumber \\
+ \left[ \mathcal{R}_s  \mathcal{L}_1 \mathcal{R}_1 \right] \left[  P_1
\mathcal{R}_{s'} \right] \quad \quad (s, s' \geq 2) , \label{20.9;1}
\end{gather}
where we have used $P_i \mathcal{R}_i = \mathcal{R}_i P_i  = \mathcal{R}_i$. All
expressions in square brackets are elements of $\mathfrak{J}$ according to Lemmas~\ref{19.9;21}, \ref{19.9;22}, \ref{19.9;23}, 
hence, the lhs of (\ref{20.9;1}) also belongs to $\mathfrak{J}$ 
according to Lemma~\ref{13.8;9}. 
The implication $k \rightarrow k+1$ is proved similarly. The fact that
$\mathcal{M}_{k+1} \in \mathfrak{J}$ follows directly from the induction
assumption and Lemmas~\ref{19.9;21}, \ref{19.9;22}, \ref{19.9;23}. Let us consider, for example,
$\mathcal{R}_s \mathcal{M}_{k+1} $ for $s \geq k+2$. By (\ref{8.8.7}) we obtain
\begin{gather}
 \mathcal{R}_s \mathcal{M}_{k+1}  = \left[ \mathcal{R}_s \mathcal{M}_{k} \right]
+ \left[ \mathcal{R}_s P_{k+1} \right] \left[\mathcal{R}_{k+1} \mathcal{M}_{k}
\right]
+ \left[\mathcal{R}_s \mathcal{M}_{k} \mathcal{R}_{k+1}  \right]
+ \left[\mathcal{R}_s P_{k+1} \right] \left[\mathcal{R}_{k+1} \mathcal{M}_{k}
\mathcal{R}_{k+1} \right] \nonumber\\
+ \left[\mathcal{R}_s P_{k+1} \right]  \left[\mathcal{R}_{k+1} \mathcal{L}_{k+1}
\mathcal{R}_{k+1} - \sum_{i=1}^k \left[ \mathcal{R}_{k+1} P_i \right]
\mathcal{P}_i \left[P_i \mathcal{R}_{k+1}\right]\right] \nonumber\\
+ \left[\mathcal{R}_s P_{k+1} \right]  \left[\mathcal{R}_{k+1} \mathcal{L}_{k+1}
- \sum_{i=1}^k \left[\mathcal{R}_{k+1} P_i \right]\mathcal{P}_i \right]
+  \left[\mathcal{R}_{s} \mathcal{L}_{k+1} \mathcal{R}_{k+1}  - \sum_{i=1}^k
\left[ \mathcal{R}_{s} P_i \right] \mathcal{P}_i \left[P_i \mathcal{R}_{k+1}
\right]\right] \nonumber
\end{gather}
Again, according to Lemmas~~\ref{19.9;21}, \ref{19.9;22}, \ref{19.9;23} and the induction assumption all expressions in
square brackets belong to $\mathfrak{J}$.
\end{proof}

\begin{lemma} \label{19.9;21}
For $i \neq j$
\begin{equation}
 \sup_{\epsilon > 0} \left\|\mathcal{R}_i \bigl(H_0 + \epsilon\bigr)^{-1/2}
|v_{ij}|^{1/2}\right\|_{HS} < \infty . \label{20.9;51}
\end{equation}
\end{lemma}
\begin{proof}
Without loosing generality we can consider the integral operator
$\mathcal{C}(\epsilon):= \mathcal{F}_1 \mathcal{R}_1 \bigl(H_0 +
\epsilon\bigr)^{-1/2} |v_{12}|^{1/2} \mathcal{F}^{-1}_1$.
We use the Jacobi coordinates $x_1, \ldots, x_{N-1}$ and from Fig.~\ref{Fig:1} one finds $r_1 - r_2 = \gamma_1 x_1 + \gamma_2 x_2$, where
\begin{gather}
 \gamma_1 = \left[ \frac{M}{2m_1 (M - m_1)} \right]^{\frac 12} , \label{21.9;32}\\
 \gamma_2 = - \left[ \frac{M- m_1 - m_2}{2m_2 (M - m_1)} \right]^{\frac 12} . 
\end{gather}
The integral operator $\mathcal{C}(\epsilon)$ acts on $f(p_1 , x_r) \in L^2
(\mathbb{R}^{3N-3})$ as
follows
\begin{equation}
 \bigl(\mathcal{C}(\epsilon) f\bigr) (p_1, x_r)= \int \mathcal{C}(p_1, p'_1,
x_r, x'_r; \epsilon) f(p'_1, x_r') d^3 p'_1 d^{3N-6}x'_r ,
\end{equation}
where the integral kernel is
\begin{gather}
\mathcal{C}(p_1, p'_1, x_r, x'_r; \epsilon)  = \frac{\gamma_1^3}{(2 \pi)^{3/2}}
\varphi_1 (\epsilon + p_1^2; x_r ) \psi_1 (\epsilon + p_1^2; x'_r)
\left\{[1-\mu_1(\epsilon + p_1^2)]^{-\frac 12} - 1\right\} \nonumber\\
 \times \widehat{\bigl| v_{12}\bigr|^{\frac 12}} \bigl(\gamma_1 (p_1 -
p'_1)\bigr)
\exp
\left\{ i\gamma_2 (p'_1 - p_1)\cdot  x_2 \right\} , \label{20.9;3}
\end{gather}
and $\psi_1$ is expressed through $\varphi_1$ through (\ref{11.01.12/1}).
The hat in (\ref{20.9;3}) denotes a standard Fourier transform in $\mathbb{R}^3$.
Now we calculate the square of the Hilbert--Schmidt norm
\begin{gather}
 \left\|\mathcal{C}(\epsilon)\right\|_{HS}^2 = \int \left| \mathcal{C}(p_1,
p'_1, \xi, \xi'; \epsilon)\right|^2 \; d^3 p_1 d^3 p'_1 d^{3N-6}\xi d^{3N-6}\xi'
\nonumber\\
\leq \gamma_1^3 C^2_\psi C_0 \int_{ p_1^2 \leq \beps - \epsilon} \left\{[1-\mu_1(\epsilon +
p_1^2)]^{-\frac 12} - 1\right\}^2 d^3 p_1 \label{31.01/02},
\end{gather}
where $C_\psi$ was defined in Corollary~\ref{20.9;6} and
\begin{equation}
 C_0 :=   (2\pi)^{-3} \max_{i < k} \int \left|  \widehat{\bigl| v_{ik}\bigr|^{\frac
12}} \bigl( t\bigr) \right|^2 d^3t
\end{equation}
is finite since $|v_{ik}|^{\frac 12} \in L^2 (\mathbb{R}^3)$. Using (\ref{mubound}) we
obtain
\begin{gather}
 \left\|\mathcal{C}(\epsilon)\right\|_{HS}^2 \leq  \gamma_1^3 C^2_\psi C_0 \int_{ |p_1|\leq
\sqrt{\beps}} \left(a_\mu^{-\frac 12} |p_1|^{-1} - 1\right)^2 d^3p_1 \nonumber \\
= (4\pi)\gamma_1^3  C^2_\psi C_0  \left[a_\mu^{-1} \sqrt{\beps} - a_\mu^{-\frac 12} \beps +
(1/3) \beps^{\frac 32}\right] . 
\end{gather}
\end{proof}
Another less trivial estimate is given by
\begin{lemma}\label{19.9;22}
  For all $1 \leq i < j \leq N$ and $1 \leq i < s \leq N $
\begin{equation}
\sup_{\epsilon > 0}
\left\|\mathcal{R}_i \bigl(H_0 + \epsilon\bigr)^{-1/2} v_{is} \bigl(H_0 +
\epsilon\bigr)^{-1/2} \mathcal{R}_j\right\|_{HS} < \infty
\end{equation}
\end{lemma}
\begin{proof}
 For $s = j$ the result easily follows from Lemma~\ref{19.9;21}. So without loosing
generality it suffices to prove that
\begin{equation}
\sup_{\epsilon > 0} \left\|\mathcal{C}_1 (\epsilon) \right\|_{HS} < \infty ,\label{20.9;25}
\end{equation}
where we have defined
\begin{equation}
\mathcal{C}_1 (\epsilon):= \mathcal{R}_1 \bigl(H_0 + \epsilon\bigr)^{-1/2}
v_{13} \bigl(H_0 + \epsilon\bigr)^{-1/2} \mathcal{R}_2 .
\end{equation}
By (\ref{20.9;11}) we have
\begin{gather}
 \mathcal{C}_1 (\epsilon) =
P_1 (\epsilon) \mathfrak{g}_1 (\epsilon)\bigl(H_0 + \epsilon\bigr)^{-1/2} v_{13}
\bigl(H_0 + \epsilon\bigr)^{-1/2} \mathfrak{g}_2 (\epsilon) P_2 (\epsilon) \nonumber \\
= P_1 (\epsilon)\mathfrak{g}_2 (\epsilon) \bigl(H_0 + \epsilon\bigr)^{-1/2}
\mathfrak{g}_1 (\epsilon) v_{13} \bigl(H_0 + \epsilon\bigr)^{-1/2}P_2 (\epsilon) .
\end{gather}
We can write $\mathcal{C}_1 (\epsilon) = \mathcal{C}_2 (\epsilon) -
\mathcal{C}_3 (\epsilon)$, where
\begin{gather}
\mathcal{C}_2 (\epsilon):= P_1 (\epsilon)[\mathfrak{g}_2 (\epsilon) +1]\bigl(H_0
+ \epsilon\bigr)^{-1/2} \mathfrak{g}_1 (\epsilon) v_{13} \bigl(H_0 +
\epsilon\bigr)^{-1/2}P_2 (\epsilon) , \label{20.9;41}\\
\mathcal{C}_3 (\epsilon):= \mathcal{R}_1 (\epsilon)\bigl(H_0 + \epsilon\bigr)^{-1/2} v_{13} \bigl(H_0 + \epsilon\bigr)^{-1/2}P_2 (\epsilon) .
\end{gather}
From Lemma~\ref{19.9;21} it easily follows that $\sup_{\epsilon > 0} \left\|\mathcal{C}_3
(\epsilon) \right\|_{HS} < \infty $. Therefore, (\ref{20.9;25}) reduces to proving
that
\begin{equation}
 \sup_{\epsilon > 0} \left\|\mathcal{C}_2 (\epsilon) \right\|_{HS} < \infty . \label{20.9;42}
\end{equation}
Apart from the sets of coordinates $x_i , z_i$ depicted in Fig.~\ref{Fig:1} (left) we shall need the third set of coordinates
$t_1, t_2, \ldots, t_{N-1}$ depicted in Fig.~\ref{Fig:1} (right).
Thereby $t_1 = z_1$
and $t_2 = \sqrt{2\mu_{13}} (r_3 - r_1)$. The coordinate $t_3$ points in the
direction from the center of mass of the particles $\{4, 5, \ldots, N\}$
to the center of mass of the particle pair
$\{1,3\}$ and $t_i = x_i = z_i$ for $i \geq 4$. The scales are set so as to make the kinetic energy
operator take the form $H_0 = - \sum_i \Delta_{t_i} $. The coordinates are connected through
\begin{gather}
 t_2 = b_{22} z_2 + b_{23} z_3 \\
 t_3 = b_{32} z_2 + b_{33} z_3
\end{gather}
where
\begin{gather}
b_{22} =  -\left[ \frac{m_3 (M-m_2)}{(M - m_1-m_2)(m_1 +m_3)} \right]^{\frac 12}\\
b_{23} = \left[ \frac{m_1(M- m_1 -m_2 -m_3)}{(m_1 +m_3)(M - m_1 -m_2)} \right]^{\frac 12},
\end{gather}
and $b_{32} = b_{23}$, $b_{33} = - b_{22}$. We also set $b_{11} = 1$, $b_{12} = b_{13} = b_{21} = b_{31}  = 0$. Then $b_{ik}$  
are entries of the $3\times 3$ orthogonal matrix $b$. 
From the $2\times 2$ matrix in
(\ref{21.9;1})--(\ref{21.9;2}) we can construct the $3 \times 3 $ orthogonal matrix $a$ by setting $a_{33} = 1$ and $a_{13} = a_{23} = a_{32} = a_{31} = 0$. Then
\begin{equation}
 t_i = \sum_{i=1}^3 c_{ik} x_k \quad \quad (i = 1, 2, 3),
\end{equation}
where $c_{ik}$ are elements of the orthogonal matrix $c = ab$. 

The full and partial Fourier transforms associated with $t_i$ are 
\begin{gather}
 (F_t f ) (p_{t_1}, p_{t_2}, \ldots, p_{t_{N-1}}):= \frac 1{(2\pi)^{\frac{3N-3}2}} \int e^{-i\sum_{k = 1}^{N-1} p_{t_k} \cdot t_k} f(t_1, \ldots, t_{N-1}) d^3 t_1 \ldots d^3 t_{N-1} \\
(\mathcal{F}_t f ) (t_1, p_{t_2}, \ldots, p_{t_{N-1}}):= \frac 1{(2\pi)^{\frac{3N-6}2}} \int e^{-i\sum_{k = 2}^{N-1} p_{t_k} \cdot t_k} f(t_1, \ldots, t_{N-1}) d^3 t_2 \ldots d^3 t_{N-1}
\end{gather}
For shorter notation we introduce the tuples
$p_{t_r} := (p_{t_2}, p_{t_3}, \ldots, p_{t_{N-1}})$ and $p_{t_c} := (p_{t_3},
p_{t_4}, \ldots, p_{t_{N-1}})$.
In analogy with the proof of Theorem~\ref{3.8;8} let us introduce the operator $ B_t
(\epsilon)$, which acts on $f \in L^2(\mathbb{R}^{3N-3})$ as
\begin{equation}
 B_t (\epsilon) f = F_t^{-1} (1+ (p_{t_1}^2 + \epsilon)^{\frac \zeta2})(F_t
f)  +F_t^{-1} (|p_{t_c}|^\zeta - 1)\chi_1 (|p_{t_c}|) (F_t f).
\end{equation}
and $\zeta = 3/4$ (in fact, we could take any $\zeta \in (1/2 , 1)$). For
all $\epsilon >0$ the operators $B_t (\epsilon)$ and $B_t^{-1} (\epsilon)$
are bounded. Inserting the identity $B_t B_t^{-1} = 1$  into (\ref{20.9;41})
and using $[B_t , v_{13}] =0$ we get
\begin{equation}
 \mathcal{C}_2 (\epsilon) = \mathcal{C}_4 (\epsilon) \sign (v_{13})
\mathcal{C}_5 (\epsilon)
\end{equation}
where
\begin{gather}
 \mathcal{C}_4 (\epsilon)  := P_1 (\epsilon)_1 [\mathfrak{g}_2 (\epsilon) + 1] 
\bigl(H_0 + \epsilon\bigr)^{-1/2} B_t (\epsilon) \mathfrak{g}_1
(\epsilon)\bigl|v_{13}\bigr|^{1/2} , \\
\mathcal{C}_5 (\epsilon) = \bigl|v_{13}\bigr|^{1/2} \bigl(H_0 +
\epsilon\bigr)^{-1/2}  B^{-1}_t (\epsilon) P_2 (\epsilon) . 
\end{gather}
Now (\ref{20.9;42}) follows from the inequalities
\begin{gather}
 \sup_{\epsilon > 0} \| \mathcal{C}_4 (\epsilon)\|_{HS} < \infty , \\
\sup_{\epsilon > 0} \| \mathcal{C}_5 (\epsilon)\| < \infty . \label{24.9;6}
\end{gather}
By construction of coordinates $r_1- r_3 = \gamma_1 x_1 + \gamma_2' x_2 + \gamma_3 x_3$, where $\gamma_1$ was defined in (\ref{21.9;32}) and
\begin{gather}
 \gamma_2' = \left[ \frac{m_2}{2(M - m_1)(M - m_1 -m_2 )} \right]^{\frac 12} , \\
 \gamma_3 = - \left[ \frac{M- m_1 - m_2 -m_3}{2m_3 (M - m_1-m_2)} \right]^{\frac 12} . 
\end{gather}
Let us first consider the operator $ \mathcal{\hat C}_4 (\epsilon) :=
\mathcal{F}_1 \mathcal{C}_4 (\epsilon)   \mathcal{F}^{-1}_1$, which  acts on $f(p_1 , x_r) \in L^2 (\mathbb{R}^{3N-3})$ as
follows
\begin{equation}
 \bigl(\mathcal{\hat C}_4 (\epsilon) f\bigr) (p_1, x_r)= \int \mathcal{C}_4
(p_1, p'_1, x_r, x'_r; \epsilon) f(p'_1, x_r') d^3 p'_1 d^{3N-6}x'_r ,
\end{equation}
whereby  the integral kernel is (c. f. eq.~(\ref{20.9;3}))
\begin{gather}
\mathcal{C}_4 (p_1, p'_1, x_r, x'_r; \epsilon)  = \frac{\gamma_1^3}{(2
\pi)^{3/2}} \varphi_1 (\epsilon + p_1^2; x_r ) \psi'_1 (p_1, x'_r)
\left\{[1-\mu_1(\epsilon + p_1^2)]^{-\frac 12} - 1\right\} \nonumber\\
 \times \widehat{\bigl| v_{13}\bigr|^{\frac 12}} \bigl(\gamma_1 (p_1 -
p'_1)\bigr)
\exp
\left\{ i(p'_1 - p_1)\cdot(\gamma_2' x_2 + \gamma_3 x_3 )\right\} . \label{20.9;54}
\end{gather}
In (\ref{20.9;54}) we have introduced the function $\psi_1' = \mathcal{F}_1 F_1^{-1}  \hat \psi'_1$,
where
\begin{gather}
\hat \psi'_1 = \hat \varphi_1 (\epsilon +p_1^2; p_r) \left[1- \mu_2 (\epsilon +
p_{t_1}^2)\right]^{- \frac 12} \left[p_1^2 + p_r^2 + \epsilon\right]^{- \frac
12} \nonumber\\
\times \left\{1+ (p_{t_1}^2 + \epsilon)^{\frac \zeta2} + (|p_{t_c}|^\zeta -
1)\chi_1 (|p_{t_c}|) \right\} \label{21.9;31}
\end{gather}
and in (\ref{21.9;31}) one has to substitute $p_{t_1} = a_{11} p_1 + a_{12} p_2$ and
\begin{equation}
 p_{t_c}^2 = \sum_{i=4}^{N-1}p_i^2 + (c_{31}p_1 + c_{32}p_2 + c_{33}p_3)^2.
\end{equation}
Repeating the arguments in Lemma~\ref{19.9;21} we obtain
\begin{equation}
 \|\mathcal{C}_4 (\epsilon)\|_{HS} \leq (4\pi) \gamma_1^3 C^2_{\psi'} C_0  \left[a_\mu^{-1}
\sqrt{\beps} - a_\mu^{-\frac 12} \beps + (1/3) \beps^{\frac 32}\right] ,
\end{equation}
where by definition
\begin{equation}
 C^2_{\psi'} =  \sup_{\epsilon >0} \sup_{p_1 \in \mathbb{R}^3}\int \left|
\psi'_1 (p_1, x_r)\right|^2 d^{3N-6} x_r =
\sup_{\epsilon >0}  \sup_{p_1 \in \mathbb{R}^3} \int \left| \hat \psi'_1 (p_1,
p_r)\right|^2 d^{3N-6} p_r \label{21.9;61}
\end{equation}
To further estimate the expression in (\ref{21.9;61}) we use the inequality
\begin{equation}
\frac{\left\{1+ (p_{t_1}^2 + \epsilon)^{\frac \zeta2} + (|p_{t_c}|^\zeta -
1)\chi_1 (|p_{t_c}|) \right\}^2 }{p_{t_1}^2 + p_{t_2}^2 + p_{t_c}^2 + \epsilon}
\leq \frac{4}{(p_{t_1}^2 + \epsilon)^{1-\zeta}}  \label{24.9;1}
\end{equation}
Let us check (\ref{24.9;1}) for $|p_{t_c}| < 1$. By (\ref{8.6;61})
\begin{gather}
 \frac{\left\{ (p_{t_1}^2 + \epsilon)^{\frac \zeta2} + |p_{t_c}|^\zeta
\right\}^2 }{p_{t_1}^2 + p_{t_2}^2 + p_{t_c}^2 + \epsilon} \leq
 \frac{4\left\{ p_{t_1}^2 + p_{t_c}^2 + \epsilon \right\}^\zeta }{p_{t_1}^2 +
p_{t_c}^2 + \epsilon} \leq \frac{4}{(p_{t_1}^2 + \epsilon)^{1-\zeta}}
\end{gather}
Similarly, one proves that (\ref{24.9;1}) holds for $|p_{t_c}| \geq 1$. Substituting
(\ref{24.9;1}), (\ref{21.9;31}) into (\ref{21.9;61})  and using that
$p_1^2 + p_r^2 = p_{t_1}^2 + p_{t_2}^2 + p_{t_c}^2$ we obtain the estimate
\begin{gather}
 C^2_{\psi'} \leq 4 \sup_{\epsilon >0} \sup_{p_1 \in \mathbb{R}^3} \int \left|
\hat \varphi_1 (\epsilon +p_1^2; p_r) \right|^2 \left[1- \mu_2 (\epsilon +
p_{t_1}^2)\right]^{-1} (p_{t_1}^2 + \epsilon)^{\zeta-1} d^{3N-6} p_r \nonumber\\
\leq 4 a_\mu^{-1}
 \sup_{\epsilon >0} \sup_{p_1 \in \mathbb{R}^3} \int \left| \hat \varphi_1
(\epsilon +p_1^2; p_r) \right|^2 \left[\epsilon + (a_{11}p_1 + a_{12} p_2)^2
\right]^{-2+ \zeta} d^{3N-6} p_r \nonumber\\
\leq 4 a_\mu^{-1} \sup_{\epsilon >0} \sup_{p_1 \in \mathbb{R}^3} \sup_{s\in
\mathbb{R}^3}  \int \left| \hat \varphi_1 (\epsilon +p_1^2; p_r) \right|^2
\left[\epsilon + (s + a_{12} p_2)^2 \right]^{-2+ \zeta} d^{3N-6} p_r \nonumber\\
\leq 4 a_\mu^{-1} \sup_{0 < \epsilon' < \epsilon }  \sup_{s\in \mathbb{R}^3}
\int \left| \hat \varphi_1 (\epsilon; p_r) \right|^2 \left[\epsilon' + (s +
a_{12} p_2)^2 \right]^{-2+ \zeta} d^{3N-6} p_r , \label{24.9;4}
\end{gather}
where we have used (\ref{mubound}). For $\zeta = 3/4$ the expression on the rhs of (\ref{24.9;4}) is finite due to
Theorem~\ref{3.8;8}. It remains to prove (\ref{24.9;6}). Like in Corollary~\ref{24.9;9} let us define $\eta_2 (\epsilon;t_r) \in L^2 (\mathbb{R}^{3N-6})$
as
\begin{equation}
 \eta_2 (\epsilon;t_r) := |v_{13}|^{\frac 12} \bigl(-\Delta_{t_r} +
\epsilon\bigr)^{-\frac 12} \tilde B_t^{-1} (\epsilon) \varphi_2 (\epsilon). \label{24.9;10}
\end{equation}
In (\ref{24.9;10}) we use the inverse of $\tilde B_t (\epsilon) \in \mathfrak{B}
(L^2 (\mathbb{R}^{3N-6}))$, which acts on $f \in L^2(\mathbb{R}^{3N-6})$ as
\begin{equation}
 \tilde B_t (\epsilon) f = (1+ \epsilon^{\frac \zeta2} ) f  +\mathcal{F}_{t}^{-1}
(|p_{t_c}|^\zeta - 1)\chi_1 (|p_{t_c}|) (\mathcal{F}_t f).
\end{equation}
and $\zeta = 3/4$. Eq.~(\ref{24.9;10}) is equivalent to the expression (\ref{11.10;1}) corresponding to the subsystem $\mathfrak{C}_2$ (though $\tilde B_t (\epsilon)$ and 
$B_{12} (\epsilon)$ are defined using different coordinates, they are, in fact, equal, see the discussion around Eq.~(8), (9) in \cite{jpasubmit}). 
Then the operator $F_t \mathcal{C}_5 (\epsilon) F_t^{-1}$ acts on
$\hat f(p_{t_1}, p_{t_r}) \in L^2 (\mathbb{R}^{3N-3})$ as follows
\begin{equation}
 \bigl(F_t \mathcal{C}_5 (\epsilon) F_t^{-1} \hat f\bigr)(p_{t_1}, p_{t_r}) =
\hat \eta_2^* (p_{t_1}^2 + \epsilon; p_{t_r}) \int \hat \varphi_2 (p_{t_1}^2 +
\epsilon; p'_{t_r}) \hat f(p_{t_1}, p'_{t_r}) d^{3N-6}p'_{t_r} ,
\end{equation}
where $\hat \eta_2 (\epsilon) = \mathcal{F}_t \eta_2 (\epsilon)$, $\hat \varphi
(\epsilon) = \mathcal{F}_t \varphi (\epsilon)$. Now it is trivial to show that
\begin{equation}
 \| \mathcal{C}_5 (\epsilon)\|^2  \leq \sup_{\epsilon >0}  \sup_{p_{t_1} \in
\mathbb{R}^3} \int \left| \hat \eta_2 (p_{t_1}^2 + \epsilon; p_{t_r}) \right|^2
d^{3N-6} p_{t_r} =
 \sup_{\epsilon >0} \int \left| \eta_2 (\epsilon; t_r) \right|^2 d^{3N-6}  t_r. \label{24.9;12}
\end{equation}
The rhs in (\ref{24.9;12}) is finite due to Corollary~\ref{24.9;9}.
\end{proof}

Note that in the expression (\ref{29.8;3}) for the operator $\hat P_1$ the function $\hat \varphi_1 (\epsilon + p_1^2; p_r)$ depends also on $p_1$ through
its first argument.
This is a source of trouble when one attempts to prove, for example,  that
$\sup_{\epsilon >0}\|P_1 (\epsilon) P_2 (\epsilon)\|_{HS} < \infty$.
Therefore, in the expression for $\hat P_j$ it makes sense to approximate  $\hat
\varphi_j$ by a function, which is
piecewise constant in the first argument. Namely, for $\epsilon \in (0, \beps]$,
$n \in \mathbb{Z}_+ $ and $k = 1, \ldots , n-1$
let us define $\hat P_1^{(k)} (\epsilon, n) $ as an operator, which acts on
$\hat f (p_1, p_r) \in L^2(\mathbb{R}^{3N-3})$ as follows
\begin{equation}\label{24.9;43}
 \hat P^{(k)}_1 (\epsilon, n) \hat f := \begin{cases}
 \hat \varphi^*_1 (\epsilon_k ; p_r) {\displaystyle \int \hat \varphi_1
(\epsilon_k , p'_r) f (p_1, p'_r)}  \: d^{3N-6} p'_r &\text{if $p_1^2 +
\epsilon\in (\epsilon_k, \epsilon_{k+1}]$},\\
0& \text{if $p_1^2 + \epsilon \notin (\epsilon_k, \epsilon_{k+1}]$} ,
\end{cases}
\end{equation}
where $\epsilon_k := k\beps n^{-1}$. We define $ P_1^{(k)}
(\epsilon, n) := F_1^{-1} \hat P_1^{(k)} (\epsilon, n) F_1$.
Similarly, using appropriate coordinates one defines $\hat P_j^{(k)} (\epsilon,
n) , P_j^{(k)} (\epsilon, n)$ for $j = 1, \ldots, N$.
\begin{lemma}\label{24.9;41}
 The following approximation formulas hold for $j, s = 1, \ldots, N$
\begin{gather}
 \sup_{\epsilon \in (0, \beps]}\left\|  P_j (\epsilon) - \sum_{k=1}^{n - 1}
P^{(k)}_j (\epsilon, n) \right\| = \hbox{o} (1/n) \\
\sup_{\epsilon \in (0, \beps]} \left\| \mathfrak{ g}_s (\epsilon)  P_j
(\epsilon) - \sum_{k=1}^{n - 1}  \mathfrak{ g}_s (\epsilon)  P^{(k)}_j
(\epsilon, n) \right\| = \hbox{o} (1/n) \quad \quad (j \neq s)
\end{gather}
\end{lemma}
\begin{proof}
Without loosing generality we set $j = 1$ and $s=2$. Generally, suppose $f, h,
g, g' \in \mathcal{H}$, where $\|f\| = \|h\|=1$ and $\mathcal{H}$ denotes some
Hilbert space.
Then the norm of the difference of projections
can be trivially estimated as follows 
\begin{equation}
 \|g(f, \cdot ) - g' (h,  \cdot)\| \leq \|g - g'\| + \|g\| \|f-h\| \label{24.9;22}
\end{equation}
and, consequently,
\begin{equation}
 \|f( f, \cdot ) - h(h,  \cdot )\| \leq 2\|f-h\|. \label{24.9;21}
\end{equation}
Using (\ref{24.9;21}) we get
\begin{gather}
  \sup_{\epsilon \in (0, \beps]}\left\| \hat P_1 (\epsilon) - \sum_{k=1}^{n - 1}
\hat P^{(k)}_1 (\epsilon, n) \right\| \leq 2 \sup_{\epsilon \in (0, \beps]}
\sup_{k = 1, \ldots, n} \sup_{p_1^2 \in [\epsilon_k - \epsilon, \epsilon_{k+1} -
\epsilon)} \| \varphi_1 (p_1^2 + \epsilon) - \varphi_1 (\epsilon_k)\| \nonumber\\
\leq 2\sup_{|\epsilon' - \epsilon|\leq \beps n^{-1} } \| \varphi_1 (\epsilon') -
\varphi_1 (\epsilon)\| = \hbox{o} (n^{-1}) \label{24.9;24}
\end{gather}
due to (\ref{7.8;25}). The norms on the rhs in (\ref{24.9;24}) are that of $L^2
(\mathbb{R}^{3N-6})$. For the second approximation formula using (\ref{24.9;22}) we
get
\begin{gather}
\sup_{\epsilon \in (0, \beps]} \left\| \mathfrak{ g}_2 (\epsilon)  P_1
(\epsilon) - \sum_{k=1}^{n - 1}  \mathfrak{ g}_2 (\epsilon)  P^{(k)}_1
(\epsilon, n) \right\| \nonumber \\
\leq  \sup_{\epsilon \in (0, \beps]} \left\| (\mathfrak{\hat g}_2 (\epsilon) +1)
\hat P_1 (\epsilon) - \sum_{k=1}^{n - 1}  (\mathfrak{\hat g}_2 (\epsilon) +1)
\hat P^{(k)}_1 (\epsilon, n) \right\| +
\hbox{o}(n^{-1}),
\end{gather}
where $\mathfrak{ \hat g}_2 (\epsilon) := F_1 \mathfrak{ g}_2 (\epsilon)F_1^{-1}$
Using (\ref{24.9;26}) and (\ref{24.9;28})--(\ref{24.9;29}) we estimate the last term as follows
\begin{gather}
\sup_{\epsilon \in (0, \beps]} \left\| (\mathfrak{\hat g}_2 (\epsilon) +1) \hat
P_1 (\epsilon) - \sum_{k=1}^{n - 1}  (\mathfrak{\hat g}_2 (\epsilon) +1)  \hat
P^{(k)}_1 (\epsilon, n) \right\|^2  \leq \nonumber\\
\sup_{\epsilon \in (0, \beps]}
\sup_{k = 1, \ldots, n} \sup_{p_1^2 \in [\epsilon_k - \epsilon, \epsilon_{k+1} -
\epsilon)} \left\|\left[1-\mu_2 \bigl(\epsilon + (a_{11} p_1 + a_{12} p_2
)^2\bigr)\right]^{- \frac 12} \bigl(\hat \varphi_1 (p_1^2 + \epsilon) - \hat
\varphi_1 (\epsilon_k)\bigr) \right\| \nonumber\\
+ \sup_{\epsilon \in (0, \beps]} \sup_{k = 1, \ldots, n} \sup_{p_1^2 \in
[\epsilon_k - \epsilon, \epsilon_{k+1} - \epsilon)} \left\|\left[1-\mu_2
\bigl(\epsilon + (a_{11} p_1 + a_{12} p_2 )^2\bigr)\right]^{- \frac 12} \hat
\varphi_1 (p_1^2 + \epsilon)  \right\| \nonumber\\
\times \sup_{\epsilon \in (0, \beps]} \sup_{k = 1, \ldots, n} \sup_{p_1^2 \in
[\epsilon_k - \epsilon, \epsilon_{k+1} - \epsilon)} \| \varphi_1 (p_1^2 +
\epsilon) - \varphi_1 (\epsilon_k)\|
\end{gather}
Now applying (\ref{mubound}) we continue the last line as
\begin{gather}
 \leq a_{12}^{-1} a_\mu^{- \frac 12} \sup_{|\epsilon - \epsilon'| < \beps
n^{-1}}  \sup_{\epsilon'' >0} \sup_{s \in \mathbb{R}^3} \left\|G_2 (s,
\epsilon'') \bigl(\varphi_1 (\epsilon') - \varphi_1 (\epsilon)\bigr)\right\| \nonumber \\
+ a_{12}^{-1} a_\mu^{- \frac 12} \left\{\sup_{\epsilon \in (0, \beps]}
\sup_{\epsilon' >0} \sup_{s \in \mathbb{R}^3} \left\|G_2 (s, \epsilon')
\varphi_1 (\epsilon)\right\|  \right\} \times
\left\{\sup_{|\epsilon' - \epsilon|\leq \beps n^{-1} } \| \varphi_1 (\epsilon')
- \varphi_1 (\epsilon)\|\right\} \nonumber \\
= \hbox{o}(n^{-1}),
\end{gather}
where we have used Theorem~\ref{3.8;8}, Corollary~\ref{19.9;10} and (\ref{7.8;25}).
\end{proof}

It remains to prove
\begin{lemma}\label{19.9;23}
$\mathcal{R}_i P_j , P_i P_j\in \mathfrak{J}$ for $i \neq j$.
\end{lemma}
\begin{proof}
 Without loosing generality it suffices to prove that $\mathcal{R}_1 P_2 \in
\mathfrak{J}$. By Lemma~\ref{24.9;41} it is enough to prove that
\begin{gather}
 \sup_{\epsilon > 0} \left\|P_1^{(i)}(\epsilon, n) \bigl(\mathfrak{g}_1
(\epsilon) + 1 \bigr) P_2^{(k)}(\epsilon, n)  \right\|_{HS} < \infty \label{24.9;44} , \\
\sup_{\epsilon > 0} \left\|P_1^{(i)}(\epsilon, n)  P_2^{(k)}(\epsilon, n)
\right\|_{HS} < \infty \label{24.9;45} . 
\end{gather}
for any given $i, k = 1, \ldots, n-1$ and $n \in \mathbb{Z}_+$.
We shall consider only (\ref{24.9;44}), (\ref{24.9;45}) is proved analogously. Eq.~(\ref{24.9;44}) follows from the
inequality
\begin{equation}
\sup_{\epsilon > 0} \left\|\mathcal{C}_L \mathcal{C}_R\right\|_{HS} < \infty ,
\end{equation}
where $\mathcal{C}_L = F_1 P_1^{(i)}(\epsilon, n) \bigl(\mathfrak{g}_1 (\epsilon)
+ 1 \bigr)  F_1^{-1}$ and $\mathcal{C}_R = F_1 P_2^{(k)}(\epsilon, n) F_1^{-1}$ are
integral
operators with the kernels
\begin{gather}
\mathcal{C}_L (p_1, p_2 , p_c ; p'_1, p'_2 , p'_c) = \hat \varphi^*_1 (\epsilon_i
, p_2, p_c) \left[1-\mu_1(\epsilon + p_1^2)\right]^{-\frac 12}
\chi_{(\epsilon_i, \epsilon_{i+1}]} (|p_1|) \hat \varphi_1 (\epsilon_i , p'_2,
p'_c) \nonumber\\
\times \delta (p_1 - p'_1 )\\
\mathcal{C}_R (p_1, p_2 , p_c ; p'_1, p'_2 , p'_c) = \hat \varphi^*_2 (\epsilon_k
, a_{21}p_1 + a_{22}p_2, p_c)
\hat \varphi_2 (\epsilon_k , a_{21}p'_1 + a_{22}p'_2, p'_c)  \nonumber \\
\times \chi_{(\epsilon_k, \epsilon_{k+1}]} (|a_{11}p_1 + a_{12}p_2|)  \delta
(a_{11}p_1 + a_{12}p_2 - a_{11}p'_1 - a_{12}p'_2)
\end{gather}
where $ \chi_\Omega : \mathbb{R} \to \mathbb{R}$ is the characteristic function
of the interval $\Omega \subset \mathbb{R}$ (the delta--function is needed formally here to compute the product kernel).

Let us define the function $D:\mathbb{R}^3 \times \mathbb{R}^3  \to \mathbb{C}$
as
\begin{equation}
  D (p_2, q_2) := \int \hat \varphi_1 (\epsilon_i ; p_2, p_c)  \hat \varphi_2^*
(\epsilon_k ; q_2, p_c) d^{3N-9} p_c.
\end{equation}
By the Cauchy--Schwarz inequality
\begin{equation}
 |D (p_2, q_2)|^2 \leq \rho_1 (p_2) \rho_2 (q_2) , \label{24.9;72}
\end{equation}
where
\begin{gather}
 \rho_1 (p_2) := \int \bigl|\hat \varphi_1 (\epsilon_i ; p_2, p_c) \bigr|^2
d^{3N-9} p_c \\
 \rho_2 (q_2) := \int \bigl|\hat \varphi_2 (\epsilon_k ; q_2, p_c) \bigr|^2
d^{3N-9} p_c,
\end{gather}
and by normalization
\begin{equation}
  \int \rho_1 (z) d^3 z = \int \rho_2 (z) d^3 z = 1. \label{24.9;71}
\end{equation}

A straightforward calculation shows that the kernel of the product $\mathcal{C}_L \mathcal{C}_R$ has the form
\begin{gather}
(\mathcal{C}_L \mathcal{C}_R) (p_1, p_2 , p_c ; p'_1, p'_2 , p'_c) =
|a_{12}|^{-3}\hat \varphi_1 (\epsilon_i , p_2, p_c) \left[1-\mu_1(\epsilon +
p_1^2)\right]^{-\frac 12}
\chi_{(\epsilon_i, \epsilon_{i+1}]} (|p_1|) \nonumber\\
\times D\bigl(a_{12}^{-1}a_{11} (p'_1 - p_1) + p'_2 , a_{21} p_1 +
a_{22}a_{12}^{-1} a_{11} (p'_1 - p_1)  + a_{22}p'_2 \bigr) \nonumber\\
\times \hat \varphi_2 (\epsilon_k , a_{21} p'_1 + a_{22} p'_2  , p'_c)
\chi_{(\epsilon_k, \epsilon_{k+1}]} (|a_{11}p'_1 + a_{12}p'_2|)
\end{gather}
Using (\ref{24.9;72}), (\ref{24.9;71}) and (\ref{mubound}) the square of the Hilbert--Schmidt norm
can be estimated as follows
\begin{gather}
\left\|\mathcal{C}_L \mathcal{C}_R\right\|^2_{HS} = \int \left|(\mathcal{C}_L
\mathcal{C}_R) (p_1, p_2 , p_c ; p'_1, p'_2 , p'_c)\right|^2
d^3 p_1 d^3 p_2 d^{3N-9} p_c d^3 p'_1 d^3 p'_2 d^{3N-9} p'_c \nonumber\\
\leq |a_{12}|^{-6} a_\mu^{-1}\int d^3 p_1 d^3 p_2 d^{3N-9} p_c d^3 p'_1 d^3 p'_2 d^{3N-9}
p'_c \bigl|\hat \varphi_1 (\epsilon_i , p_2, p_c)\bigr|^2 \bigl[\epsilon + p_1^2
\bigr]^{-1} \nonumber\\
\times \rho_1 \bigl(a_{12}^{-1}a_{11} (p'_1 - p_1) + p'_2 \bigr) \rho_2
\bigl(a_{21} p_1 + a_{22}a_{12}^{-1} a_{11} (p'_1 - p_1)  + a_{22}p'_2 \bigr) \nonumber\\
\times \bigl|\hat \varphi_2 (\epsilon_k , a_{21} p'_1 + a_{22} p'_2  , p'_c)
\bigr|^2
= |a_{12}|^{-6} a_\mu^{-1} \int d^3 p_1 d^3 p_2 d^3 p'_1 d^3 p'_2  \bigl[\epsilon + p_1^2
\bigr]^{-1} \rho_1 (p_2) \nonumber\\
\times \rho_1 \bigl(a_{12}^{-1}a_{11} (p'_1 - p_1) + p'_2 \bigr) \rho_2
\bigl(a_{21} p_1 + a_{22}a_{12}^{-1} a_{11} (p'_1 - p_1)  + a_{22}p'_2 \bigr)
\rho_2 \bigl(a_{21} p'_1 + a_{22} p'_2 \bigr) \nonumber\\
= |a_{12}|^{-6} a_\mu^{-1} \int d^3 p_1 d^3 p'_1 d^3 p'_2  \bigl[\epsilon + p_1^2
\bigr]^{-1} \rho_1 \bigl(a_{12}^{-1}a_{11} (p'_1 - p_1) + p'_2 \bigr) \nonumber \\
\times \rho_2 \bigl(a_{12}^{-1} (p'_1 - p_1) + a_{21}p'_1 + a_{22}p'_2 \bigr)
\rho_2 \bigl(a_{21} p'_1 + a_{22} p'_2 \bigr),
\end{gather}
where we have used the relation $a_{11}a_{22} - a_{12}a_{21} = 1$.
Now we make the change of variables
\begin{gather}
 \xi_1 := a_{12}^{-1} a_{11} (p'_1 - p_1 ) + p'_2 \\
\xi_2 := a_{12}^{-1} (p'_1 - p_1 ) + a_{21}p'_1 + a_{22} p'_2 \\
\xi_3 := a_{21}p'_1 + a_{22}p'_2
\end{gather}
The inverse transformation has the form
\begin{gather}
 p_1 = -a_{21}^{-1} a_{22} \xi_1 + a_{21}^{-1} \xi_2  \\
p'_1 = -a_{21}^{-1} a_{22} \xi_1 + a_{21}^{-1} a_{22} a_{11}\xi_2 - a_{12}\xi_3\\
p'_2 = \xi_1 - a_{11} \xi_2 + a_{11} \xi_3
\end{gather}
This gives
\begin{gather}
\left\|\mathcal{C}_L \mathcal{C}_R\right\|^2_{HS} \leq |a_{12}|^{-5} |a_{21}|^{-1} a_\mu^{-1} \int d^3 \xi_1 d^3 \xi_2 \rho_1 (\xi_1) \rho_2 (\xi_2)
\left\{\epsilon + a_{21}^{-2} (\xi_2 - a_{22}\xi_1 )^2 \right\}^{-1} \\
\leq |a_{12}|^{-5} |a_{21}| a_\mu^{-1}  \sup_{\epsilon >0} \sup_{s \in \mathbb{R}^3} \int d^3 \xi_2\rho_2 (\xi_2)
\left\{\epsilon + (\xi_2 + s)^2 \right\}^{-1}
\end{gather}
That the rhs is finite follows from Theorem~\ref{3.8;8}. \end{proof}

In conclusion, let us explain why the proof of Theorem~\ref{1.08;9} does not apply to the three--particle 
case, where the Efimov effect is possible \cite{yafaev,merkuriev,sobolev}.
For simplicity let us assume that the pair--interactions are bounded and finite
range, particle pairs $\{2,3\}$ and $\{1,3\}$ have zero energy resonances and
$v_{12} = 0$. Theorem~\ref{3.8;3} applies in this case as well
but instead of (\ref{mubound}) one has $1 - \mu (\epsilon) = a_\mu \sqrt\epsilon +
\mathcal{O} (\epsilon)$, see \cite{klaus1} or \cite{yafaev2}
(this makes the operators $\mathcal{R}_j$ defined in (\ref{20.9;11}) less singular
compared to the case of Theorem~\ref{1.08;9}!). Lemma~\ref{13.8;4} and Theorem~\ref{8.8.;8} apply without change. However, in
case $N=3$ the proof of
Theorem~\ref{1.08;9}
breaks down because $\mathcal{M}_3 \notin \mathfrak{J}$. Indeed, from definition
in Theorem~\ref{8.8.;8} it follows that $\mathcal{M}_3 = \mathcal{M}_2$, while the
expression for
$\mathcal{M}_2$ contains the term $\mathcal{R}_1 L_1 \mathcal{R}_2 +
\mathcal{R}_2 L_1 \mathcal{R}_1 = \mathcal{R}_1 \mathfrak{m}_2 \mathcal{R}_2 +
\mathcal{R}_2 \mathfrak{m}_2 \mathcal{R}_1$. Slightly modifying the analysis in
\cite{yafaev,sobolev} one can show that for $\epsilon \to +0$ the eigenvalues of the last operator accumulate at the point lying
in the interval $(1, \infty)$, which implies that $\mathcal{M}_3 \notin \mathfrak{J}$.

\appendix
\section{The Birman Schwinger Principle}\label{31.07;14}

The Birman--Schwinger (BS) principle was first independently formulated in \cite{birman,schwinger} and since then remains an indispensable tool of counting the
eigenvalues of Schr\"odinger operators, see f. e.
\cite{reed,klaus1,klaus2,traceideals,liebyy,liebthirring}.
The form of the BS operator that we use here is
unconventional, and it appears useful to reprove
some standard statements.
So suppose $A\geq 0$ is a self--adjoint operator and $B$ is a symmetric operator on a separable Hilbert space $\mathcal{H}$, where
$D(A^{\frac 12}) \subseteq D(B)$. (Note, that we do not require $B$ being positive as it is usually  done in the proofs of the BS principle).
It follows immediately that
$B (A+
\varepsilon)^{-1/2} \in \mathfrak{B}(\mathcal{H})$ for all $\varepsilon >0$. Indeed,
\begin{equation}
 B (A+
\varepsilon)^{-1/2} = \overline{B} (A+
\varepsilon)^{-1/2} , \label{26.9;1}
\end{equation}
where $\overline{B} = \bigl( B^*\bigr)^*$ is a closure of $B$. It is easy to see that the operator on the rhs of (\ref{26.9;1}) is a closed 
operator defined on the whole $\mathcal{H}$, hence, it is bounded by the closed graph theorem. It follows that $(A+
\varepsilon)^{-1/2} B \in \mathfrak{B}(\mathcal{H})$ as well. 
\begin{remark}
 Suppose that $A\geq 0 , B$ are self--adjoint operators with domains $D(A), D(B)$ respectively. If $\|(A+\epsilon_0)^{-1/2} B f\|$ for some $\epsilon_0 >0$ 
is uniformly bounded for all $f \in D(B)$ 
then $D(A^{\frac 12}) \subseteq D(B)$. Indeed, suppose $g \in D\bigl((A+\epsilon_0)^{\frac 12}\bigr) = D(A^{\frac 12})$. Then $|(g, Bf)| = 
\bigl| \bigl((A+\epsilon_0)g, (A+\epsilon_0)^{-1/2} Bf\bigr)\bigr|$ is uniformly bounded for all $f \in D(B)$, which by self--adjointness of $B$ means that $g \in D(B)$. 
\end{remark}

Let us define the BS operator $K: \mathbb{R}_+/\{0\} \to \mathfrak{B}(\mathcal{H})$ as
\begin{equation}
 K(\varepsilon) := - (A + \varepsilon)^{-1/2}B(A + \varepsilon)^{-1/2} \quad
\quad (\varepsilon >0) ,
\end{equation}
where $K(\varepsilon)$ is, clearly, self--adjoint.

We shall consider two cases: in the first case the BS principle is formulated in
form of an inequality (Theorem~\ref{31.07;15}) and in the second case one has an exact
equality (Theorem~\ref{31.07;16}).
\begin{theorem}\label{31.07;15}
Let $A\geq 0$ be a self--adjoint operator and $B$ a symmetric operator such that $D(A^{\frac 12}) \subseteq D(B)$.  
Suppose $A+B$ is self--adjoint on $D(A)$ and has $n$ eigenvalues (counting
multiplicities) lying in the interval
$(-\infty , -\varepsilon_0 )$, where $\varepsilon_0 >0$. Then $\#(\evs(K(\varepsilon_0)) > 1) \geq n$.
\end{theorem}
\begin{proof}
By conditions of the theorem there are $0 < \epsilon_1 < \epsilon_2 <\ldots \epsilon_s$ and orthonormal $\phi^{(j_i)}_i \in D(A)$ for $i = 1,  \ldots, s$ and $j_i = 1, \ldots, m_i$ such
that
\begin{equation}
 (A + B) \phi^{(j_i)}_i = - (\varepsilon_0 + \epsilon_i) \phi^{(j_i)}_i  \quad \quad (j_i = 1, \ldots, m_i;\; \; \; i = 1, \ldots, s)   \label{20.07;1} .
\end{equation}
The number $m_i$ is the geometric multiplicity of the eigenvalue $-(\varepsilon_0 + \epsilon_i) $ in (\ref{20.07;1}) and $\sum_{i=1}^s m_i = n$.
 Rearranging the terms in (\ref{20.07;1}) and acting on both sides by $ (A + \varepsilon_0)^{-1}$ (where clearly $\Ker (A + \varepsilon_0)^{-1} = \{0\}$) we obtain  that
$-(A + \varepsilon_0)^{-1}(B + \epsilon_i) $ for each $i = 1, \ldots, s$
has an eigenvalue equal to $1$ with the multiplicity larger or equal to $m_i$. In fact, it is exactly equal to $m_i$. Indeed, suppose
\begin{equation}
 -(A + \varepsilon_0)^{- \frac 12} (A + \varepsilon_0)^{- \frac 12} (B + \epsilon_i) \phi = \phi, \label{27.9;1}
\end{equation}
for some $\phi \in \mathcal{H}$. Since $(A + \varepsilon_0)^{- \frac 12} B$ is bounded from (\ref{27.9;1}) it follows that $\phi \in D(A^{\frac 12})$ and
$\phi \in D(B)$. Then by (\ref{27.9;1}) $\phi \in D(A)$ and $(A + B) \phi = - (\varepsilon_0 + \epsilon_i) \phi$.

Let us introduce the operator function $M: \mathbb{R} \to \mathfrak{B}(H)$, where
\begin{equation}
 M(r) := K(\varepsilon_0) - r (A+ \varepsilon_0)^{-1} . \label{25.9;1}
\end{equation}
Using the well--known fact that $\sigma(CD)/\{0\} = \sigma(DC)/\{0\}$ for any
bounded $C, D$ (see f. e. \cite{deift}) we
conclude that $1$ is an eigenvalue of $M(\epsilon_i)$ with the multiplicity $m_i$ for $i= 1, \ldots, s$.

We can assume that $\sigma_{ess} \bigl(K(\varepsilon_0)\bigr) \cap (1, \infty) = \emptyset$,
otherwise the theorem is proved. By the inequality $M(r) \leq K(\varepsilon_0)$ for $r \geq 0$ we conclude
that $\sigma_{ess} \bigl(M(r)\bigr) \cap (1, \infty) = \emptyset$ for $r \geq 0$. The operator $M(\epsilon_s)$ has an eigenvalue equal to $1$ with the multiplicity
$m_s$. On one hand, $M(\epsilon_{s-1})$ has an eigenvalue equal to $1$ with the  multiplicity
$m_{s-1}$. On the other hand, by the min--max principle for eigenvalues (Vol. 1, Theorem XIII.1 in \cite{reed}) and the inequality $M(\epsilon_{s-1}) < M(\epsilon_s)$ the
operator $M(\epsilon_{s-1})$ has $m_s$ eigenvalues (counting multiplicities), which are strictly larger than 1. Thus
$\#(\evs(M(\epsilon_{s-1})) \geq 1) \geq m_s + m_{s-1}$. Similarly, $\#(\evs(M(\epsilon_{s-2})) \geq 1) \geq m_s + m_{s-1} + m_{s-2}$. Proceeding in the same way we find that
$M(0)= K(\varepsilon_0)$ has at least $n = m_1 + \dotsb + m_s$ eigenvalues (counting multiplicities), which are strictly larger than 1.
%
%
%
\end{proof}

The proof of the next theorem largely repeats that of Proposition~2.2 in \cite{klaus2}.

\begin{theorem}\label{31.07;16}
Let $A > 0$ be a self--adjoint operator and $B$ a symmetric operator such that $D(A^{\frac 12}) \subseteq D(B)$.  
Suppose that $A + \mu B$ is self--adjoint on $D(A)$ for all $\mu \in [0, 1]$ and $\sigma_{ess} (A
+ B)\cap (-\infty, 0) = \emptyset$. Then: (a) $\sigma_{ess}
\bigl( K(\varepsilon)\bigr) \cap (1, \infty) = \emptyset$ for all $\varepsilon >0$; (b) if
$A+B$ has $n$ eigenvalues counting multiplicities in the interval
$(-\infty , -\varepsilon_0 )$ for $\varepsilon_0 >0$ then $K(\varepsilon_0)$
has exactly $n$ eigenvalues counting multiplicities in the interval $(1, \infty)$.
\end{theorem}
\begin{proof}
Note that by the min--max principle
$\sigma_{ess} (A + \mu B)\cap (-\infty, 0) = \emptyset$ for $\mu \in (0,1]$.
Following \cite{klaus2} we write
\begin{equation}
A + \mu B - z = \mu (A - z) \bigl[\mu^{-1} + (A-z)^{-1} B \bigr] \quad \quad (z
< 0)  \label{20.07;2}
\end{equation}
First, we show that $z \in \sigma \bigl(A + \mu B\bigr)$ if and only if
$\mu^{-1} \in \sigma \bigl(K(-z)\bigr)$. Or
\begin{equation}
 z \notin \sigma \bigl(A + \mu B\bigr) \Longleftrightarrow \mu^{-1} \notin \sigma \bigl(K(-z)\bigr),
\end{equation}
which is the same.
Indeed, if $\mu^{-1} \notin  \sigma \bigl(K(-z)\bigr)$ or equivalently
$\mu^{-1} \notin \sigma \bigl((A-z)^{-1} B\bigr) $ then by (\ref{20.07;2}) the operator $A + \mu B - z$ has a bounded
inverse. Conversely, if $z \notin \sigma \bigl(A + \mu B\bigr)$ then
\begin{equation}
 \bigl[\mu^{-1} + (A-z)^{-1} B \bigr]^{-1} = \bigl( A + \mu B - z\bigr)^{-1} (A
- z)
\end{equation}
The operator on the rhs is bounded because $D(A + \mu B) = D(A)$, see f.~e.
Lemma~6.2 in \cite{teschl} and the remark after Proposition~2.2 in
\cite{klaus2}.
Therefore, $\mu^{-1} \notin \sigma \bigl((A-z)^{-1} B\bigr) $, or, equivalently,
$\mu^{-1} \notin \sigma \bigl(K(-z)\bigr)$.

Now let us fix $z_0 <0$ and prove that
$\sigma_{ess} \bigl( K(-z_0) \bigr) \cap (1, \infty) = \emptyset$. By contradiction, suppose that $\mu_0 < 1$ is
such that $\mu_0^{-1} \in \sigma_{ess} \bigl(K(-z_0)\bigr)$. By the established
equivalence
$z_0 \in \sigma \bigl(A + \mu_0 B\bigr)$ is an eigenvalue. Because negative eigenvalues
of $A + \mu B$ are strictly monotonic \cite{reed} in $\mu$ the point $\mu_0^{-1}$ must be an
isolated
point of the spectrum of $K(-z_0)$. Since $\mu_0^{-1} \in \sigma_{ess}
\bigl(K(-z_0)\bigr)$ by the spectral theorem it can only be an infinitely degenerate eigenvalue. Thus there is
an orthonormal set
$\varphi_i $, $i = 1,2, \ldots$ such that $\mu_0 K(-z_0 ) \varphi_i =  \varphi_i$.
It follows that $ \varphi_i \in D\bigl((A-z_0)^{1/2}\bigr)$. The vectors
$\tilde  \varphi_i := (A-z_0)^{-1/2} \varphi_i$ form a linearly independent set and
$ \tilde \varphi_i \in D(A)$. Acting on both sides of the equation
$\mu_0 K(-z_0 ) \varphi_i =  \varphi_i$ with $(A-z_0)^{1/2}$ gives $-\mu_0
(A-z_0)^{-1} B \tilde \varphi_i = \tilde \varphi_i$. Acting on both sides of this
equation with $(A-z_0)$
we see that the degeneracy of $\mu_0^{-1}$ as an eigenvalue of $K(-z_0)$ does not
exceed the degeneracy of $z_0$ as an eigenvalue of $A + \mu_0 B$. Therefore, in
the interval
$(1,\infty)$ the spectrum of $K(-z_0)$ can contain only eigenvalues of finite
multiplicity and (a) is proved. Let us prove (b). Suppose $K(\varepsilon_0)$
in the interval $(1, \infty)$ has $n_0$ eigenvalues. By Theorem~\ref{31.07;15} (b) would be proved if we can show
that $A+B$ has at least $n_0$ eigenvalues in the interval $(-\infty, -\varepsilon_0)$. Let us choose $\delta >0$ so that
$n_0$ eigenvalues of $K(\varepsilon_0)$ are larger than $1+ \delta$. Due to (a) $\sigma_{ess} \bigl(M(r)\bigr) \cap (1, \infty) = \emptyset$ for
$r \geq 0$, where $M(r)$ was defined in (\ref{25.9;1}).
Let us set
\begin{equation}
 \mu_k (r) := \inf_{\phi_1, \ldots, \phi_{k-1} \in \mathcal{H}}
\sup_{\substack{\|\psi\| = 1 \\\psi \in [\phi_1, \ldots, \phi_{k-1}]^{\bot}}} \bigl(\psi, M(r) \psi\bigr)
\end{equation}
Clearly, $\mu_1 (0), \ldots \mu_{n_0} (0) \in (1+\delta, \infty)$ are eigenvalues of $K(\epsilon_0)$. For each $k = 1, \ldots, n_0 $ the function $\mu_k (r)$ is continuous and
on the set $\mathbb{R}_+ \cap \{r| \mu_k (r) \in [1+ \delta/4, \infty) \}$ it is also monotone decreasing.
Using  Lemma~\ref{25.9;2} it is easy to show that $\mu_k (L) \leq 1 + \delta/2$ for $L$ large enough. Therefore, we conclude that
there exist $\epsilon_1 < \dotsb  < \epsilon_s$ and the vectors $\varphi_i^{(j_i)}$ such that
\begin{equation}
 M(\epsilon_i) \varphi_i^{(j_i)} = (1+\delta) \varphi_i^{(j_i)}  \quad \quad (i = 1, \ldots, s;\; \; \; j_i = 1, \ldots, m_i) \label{25.9;7}
\end{equation}
and $m_1 + \dotsb + m_{s} = n_{0}$. For each $i = 1, \ldots, s$ the vectors $\varphi_i^{(1)}, \ldots, \varphi_i^{(m_i)}$ are orthonormal.
Again, we set $\tilde  \varphi^{(j_i)}_i := (A+ \varepsilon_0 )^{-1/2} \varphi^{(j_i)}_i$, where $\tilde  \varphi^{(j_i)}_i \in D(A)$ and for each $i$ the vectors
$\tilde \varphi^{(1)}_i , \ldots, \tilde \varphi^{(m_i)}_i$ form a linearly independent set. From (\ref{25.9;7}) it follows that
\begin{equation}
 \left[(1+ \delta) A + B \right] \tilde \varphi^{(j_i)}_i = [- \varepsilon_0 - \delta\varepsilon_0 - \epsilon_i ] \tilde \varphi^{(j_i)}_i \quad \quad (i = 1, \ldots, s;\; \; \; j_i = 1, \ldots, m_i)
\end{equation}
Hence, the operator $(1+ \delta) A + B $ has at least $n_0$ eigenvalues in the interval $( -\infty , - \varepsilon_0) $. The same is true for $A + B$ due to
$A + B < (1+ \delta) A + B$.
\end{proof}
\begin{remark}
Suppose that $A, B \in \mathfrak{B}(\mathcal{H})$ are self--adjoint operators such that $A$ is bounded from
below by a positive constant and $\sigma_{ess} (A + B)\cap (-\infty, 0) = \emptyset$. One can
define a bounded self--adjoint operator $K(0) := A^{-1/2} B A^{-1/2}$.
The proof of Theorem~\ref{31.07;16} can be easily utilized to prove the
following statement:  (a) $\sigma_{ess}
\bigl( K(0)\bigr) \cap (1, \infty) = \emptyset$; (b) if
$A+B$ has $n$ eigenvalues counting multiplicities lying in the interval
$(-\infty , 0 )$ then $K(0)$ has exactly $n$ eigenvalues counting
multiplicities in the interval $(1, \infty)$.
\end{remark}
\begin{lemma}\label{25.9;2}
Suppose $f \in \mathcal{H}$ and $A \geq 0$ is a self--adjoint operator acting in $\mathcal{H}$ with domain $D(A)$.
Then for any given $c, \varepsilon_0 >0 $ one can find $L >0$ such that $f(f, \cdot) - L(A+\varepsilon_0)^{-1} \leq c$.
\end{lemma}
\begin{proof}
Without loosing generality we can set $\|f\| = 1$.  Let $\mathbb{P}_\Omega$
denote the spectral projections of the operator $A$. Using that $\sslim_{k \to \infty} \mathbb{P}_{[0, k]} = 1$ we can set $k_0 > 0$ so that
$\|f - f' \| < c/2$, where $f' = \mathbb{P}_{[0, k_0 ]} f$. Then $\|f(f, \cdot) - f'(f', \cdot)\| < c $  and we get
\begin{gather}
f(f, \cdot) - L(A+\varepsilon_0)^{-1} \leq c  + \mathbb{P}_{[0, k_0 ]} \bigl(  f'(f', \cdot) - L(A+\varepsilon_0)^{-1} \bigr)
\mathbb{P}_{[0, k_0 ]} \\ \leq c + \mathbb{P}_{[0, k_0 ]} \bigl(  f'(f', \cdot) - L(k_0 +\varepsilon_0)^{-1} \bigr) \mathbb{P}_{[0, k_0 ]} \leq c ,
\end{gather}
where we set $L = 2(k_0 + \varepsilon_0 )$.
\end{proof}

\end{document}